\documentclass[11pt]{article}
\usepackage{matharticle}
\usepackage{geometry}
\usepackage{mathrsfs}
\usepackage{color}
\usepackage[usenames,dvipsnames,svgnames,table]{xcolor}
\usepackage[colorlinks=true, linkcolor=red, urlcolor=blue, citecolor=gray]{hyperref}
\usepackage{algorithm}
\usepackage{algpseudocode}
\usepackage[font=small]{caption}
\usepackage{subcaption}
\usepackage{float}
\usepackage{wrapfig}
\usepackage{enumitem}
\usepackage{booktabs}
\definecolor{verylight}{gray}{0.90}
\newcommand{\algoname}[1]{\textnormal{\textsc{#1}}}

\newcommand{\That}{\widehat{T}}

\newcommand{\dtv}{d_{\mathrm{TV}}}
\newcommand{\dkl}{d_{\mathrm{KL}}}
\newcommand{\rank}{\mathrm{rank}}
\newcommand{\diag}{\mathrm{diag}}
\newcommand{\toep}{\mathrm{Toep}}

\newcommand{\avg}{\mathrm{avg}}
\newcommand{\Otilde}{\widetilde{O}}
\newcommand{\Omegatilde}{\widetilde{\Omega}}

\newcommand{\specialcell}[2][c]{%
  \begin{tabular}[#1]{@{}c@{}}#2\end{tabular}}

\title{Sample Efficient Toeplitz Covariance Estimation}

\author{Yonina C. Eldar\\ Weizmann Institute of Science \\ \texttt{yonina.eldar@weizmann.ac.il} \and Jerry Li \\ Microsoft Research\\ \texttt{jerrl@microsoft.com} \and Cameron Musco\\UMass Amherst\\ \texttt{cmusco@cs.umass.edu} \and Christopher Musco\\New York University\\ \texttt{cmusco@nyu.edu}}

  \usepackage{nth}
  \usepackage{intcalc}

\begin{document}
\maketitle

\begin{abstract}
We study the sample complexity of estimating the covariance matrix $T$ of a distribution $\mathcal{D}$ over $d$-dimensional vectors, under the assumption that $T$ is Toeplitz.
This assumption arises in many signal processing problems, where the covariance between any two measurements only depends on the time or distance between those measurements.\footnote{In other words, measurements are drawn from a `wide-sense' stationary process.} We are interested in estimation strategies that may choose to view only a \emph{subset} of entries in each vector sample $x \sim \mathcal{D}$, which often equates to reducing hardware and communication requirements in applications ranging from wireless signal processing to advanced imaging. 
Our goal is to minimize both 1) the number of vector samples drawn from $\mathcal{D}$ and 2) the number of entries accessed in each sample. 

We provide some of the first non-asymptotic bounds on these sample complexity measures that exploit $T$'s  Toeplitz structure, and by doing so, significantly improve on results for generic covariance matrices. These bounds follow from a novel analysis of classical and widely used estimation algorithms (along with some new variants), including methods based on selecting entries from each vector sample according to a so-called \emph{sparse ruler}.

In addition to results that hold for any Toeplitz $T$, we further study the important setting when $T$ is close to low-rank, which is often the case in practice. 
We show that methods based on sparse rulers perform even better in this setting, with sample complexity scaling sublinearly in $d$. Motivated by this finding, we develop a new covariance estimation strategy that further improves on existing methods in the low-rank case: when $T$ is rank-$k$ or nearly rank-$k$, it achieves sample complexity depending polynomially on $k$ and only logarithmically on $d$. 

Our results utilize tools from random matrix sketching, leverage score based sampling techniques for continuous time signals, and sparse Fourier transform methods. In many cases, we pair our upper bounds with matching or nearly matching lower bounds.
\end{abstract}
%


\thispagestyle{empty} 
\clearpage
\setcounter{page}{1}
\section{Introduction}
\label{sec:intro}
Estimating the covariance matrix of a distribution $\mathcal{D}$ over vectors in $\C^d$ given independent samples $x^{(1)}, \ldots, x^{(n)} \sim \mathcal{D}$ is a fundamental statistical problem. In signal processing, many applications require this problem to be solved under the assumption that $\mathcal{D}$'s covariance matrix, $T \in \R^{d\times d}$, is a \emph{symmetric Toeplitz matrix}. Specifically, $T_{a,b} = T_{c,d}$ whenever $|a - b| = |c-d|$.

Toeplitz structure naturally arises when entries in the random vector correspond to measurements on a spatial or temporal grid, and the covariance between measurements only depends on the distance between them, as in a stationary process. To name just a few applications (see  \cite{romero2016compressive} for more), estimation algorithms for Toeplitz covariance matrices are used in:
\begin{itemize}
\item Direction-of-arrival (DOA) estimation for signals received by antenna arrays  \cite{KrimViberg:1996}, which allows, for example, cellular networks to identify and target signal transmissions based on the geographic location of devices \cite{DahlmanMildhParkvall:2014,BogaleLe:2016}. 
\item Spectrum sensing for cognitive radio \cite{MaLiJuang:2009,CohenTsiperEldar:2018}, which allows a receiver to estimate which parts of the frequency spectrum are in use at any given time.
\item Medical and radar imaging processing \cite{SnyderOSullivanMiller:1989,RufSwiftTanner:1988, fuhrmann1991application, BrookesVrbaRobinson:2008,AslMahloojifar:2012,CohenEldar:2018}.
\end{itemize}
In these applications, the goal is to find an approximation to $T$ using as few samples as possible. In contrast to generic covariance estimation, there is significant interest in algorithms that consider only a \emph{subset} of the entries in each of the samples $x^{(1)},\ldots, x^{(n)}$. This subset can be chosen in any way, deterministically or randomly. This suggests
two separate measures of sample complexity: 
\begin{description}
	\smallskip\item[1. \hspace{-.15em}Vector sample complexity\hspace{-.15em} (VSC).] How many vector samples $x^{(1)}, \ldots, x^{(n)} \sim \mathcal{D}$ does an algorithm require to estimate the Toeplitz covariance matrix $T$ up to a specified tolerance?
	\item[2. \hspace{-.15em}Entry sample complexity\hspace{-.15em} (ESC).] How many entries, $s$, does the algorithm view in each $x^{(\ell)}$? 
\end{description}
There is typically a trade off between VSC and ESC, which makes different algorithms suited to different applications. 
For example, in direction-of-arrival estimation, fewer vector samples means a shorter acquisition time, while fewer entry samples means that fewer active receiving antennas are needed. 
An algorithm with high ESC and low VSC might give optimal performance in terms of acquisition speed, while one with low ESC and higher VSC would minimize  hardware costs.
In many cases, it is also natural to simply consider the combination of VSC and ESC:
\begin{description}
	\item[3. \hspace{-.15em}Total sample complexity\hspace{-.15em} (TSC).] How many total entries $n \cdot s$  does the algorithm read across all sampled vectors to estimate the covariance $T$?
\end{description}

\begin{wrapfigure}{r}{0.55\textwidth}
	\vspace{-1em}
	\captionsetup{width=.95\linewidth}
	\centering
	\begin{subfigure}[t]{0.25\textwidth}
		\centering
		\includegraphics[width=1\textwidth]{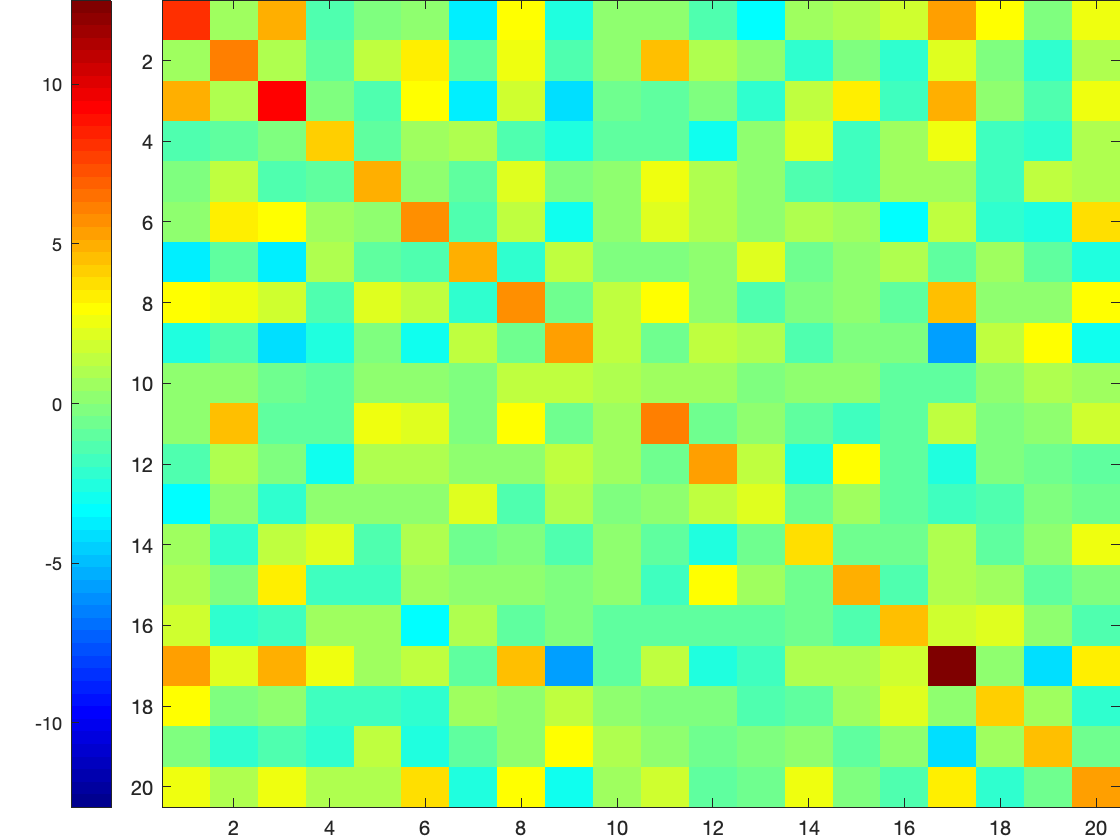}
		\caption{\small General covariance.}
	\end{subfigure}
	~
	\begin{subfigure}[t]{0.25\textwidth}
		\centering
		\includegraphics[width=1\textwidth]{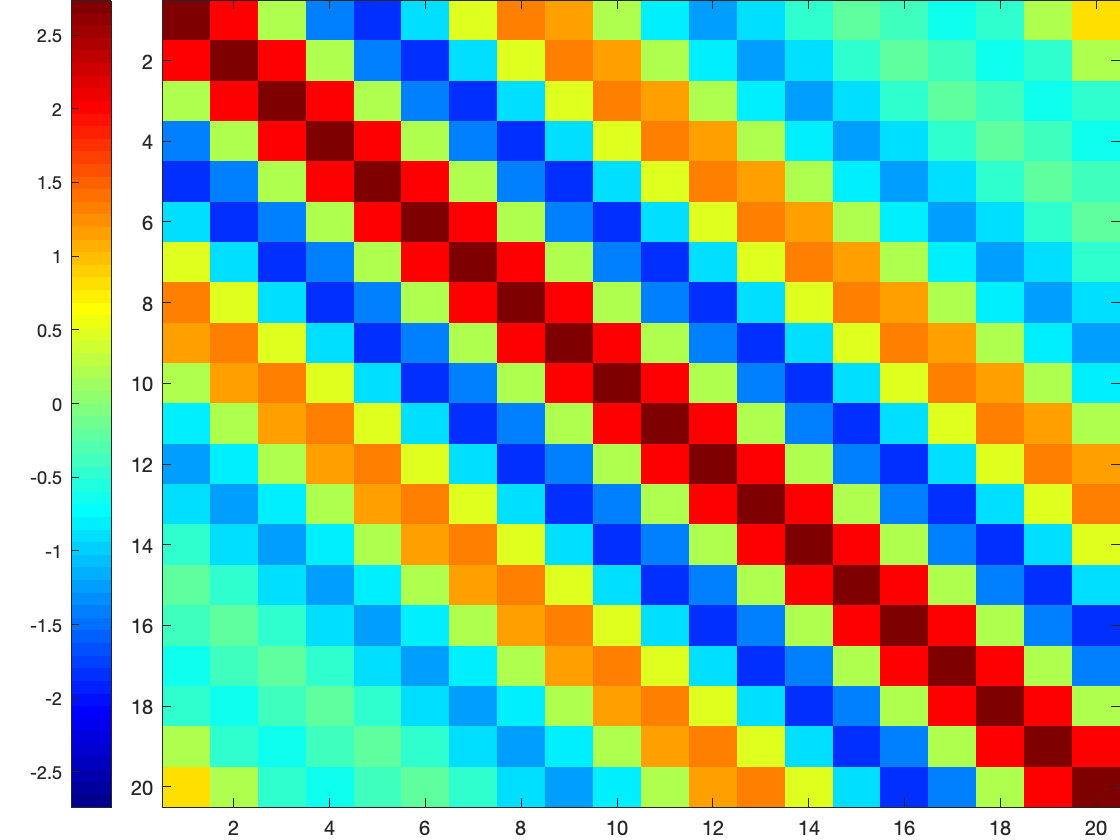}
		\caption{\small Toeplitz covariance.}
	\end{subfigure}
	\vspace{-0em}
	\caption{\small We study methods for estimating Toeplitz covariance matrices from samples of random vectors. This requires learning $\Theta(d)$ parameters, in contrast to $\Theta(d^2)$ for generic positive semidefinite covariance matrices.}
	\label{fig:params}
	\vspace{-2em}
\end{wrapfigure}

In this paper we introduce new estimation algorithms and analyze some common practical approaches, to obtain non-asymptotic bounds (sometimes tight) on these fundamental sample complexity measures.
Surprisingly, relatively little was known previously in the non-asymptotic setting, beyond results for general  covariance matrices without Toeplitz structure. Intuitively, since Toeplitz matrices are parameterized by just $O(d)$ variables (one for each diagonal) we might expect to learn them with fewer samples.

\subsection{Our contributions in brief} 
We first prove that it is in fact possible to estimate Toeplitz covariance matrices with lower vector sample complexity than general covariance matrices. We focus on the common case when $x^{(1)},\ldots,x^{(n)}$ are drawn from a $d$-dimensional Gaussian distribution with covariance $T$ and consider the goal of recovering an approximation $\tilde{T}$ with spectral norm error $\|T - \tilde{T}\|_2 \leq \eps \|T\|_2$. We show that with entry sample complexity (ESC) $s = d$ (i.e., we view all entries in each vector sample) a fast and simple algorithm\footnote{Compute the sample covariance of $x^{(1)}, \ldots, x^{(n)}$ and average its diagonals to make it Toeplitz.} achieves vector sample complexity (VSC) $O\left(\frac{\log(d/\eps) \cdot \log d}{  \eps^2}\right)$. 

This bound is significantly less than the $\Omega \left({d}/{\eps^2}\right)$ required for generic covariance matrices, and serves as a baseline for the rest of our results, which focus on methods with ESC $s < d$. In this setting, we give a number of new upper and lower bounds, which we review in detail, along with our main techniques, in Section \ref{sec:our_results}. 
We briefly describe our contributions here and summarize the bounds in Table \ref{table:overviewtable}. Several of our results apply in the important practical setting (see e.g., \cite{KrimViberg:1996,ChiEldarCalderbank:2013,qiao2017gridless}) where $T$ is both Toeplitz and either exactly or approximately rank-$k$ for some $k \ll d$ .

\begin{table}[h]
\setlength{\tabcolsep}{10pt}
\centering
\footnotesize
\renewcommand{\arraystretch}{1.2}
\begin{tabular}{cccc}
\toprule
Algorithm & ESC & VSC upper bound & TSC lower bound \\
\midrule
\multicolumn{4}{c}{\cellcolor{verylight}{\bf general $d\times d$ Toeplitz matrix}} \vspace{0.1cm}\\
Full samples  & $d$ & $\Otilde \Paren{{1/\eps^2}}$~(Thm.~\ref{thm:linear}) & $\Omega \Paren{{d/\eps^2}}$~(folklore) \\[.2cm]
$\Theta(\sqrt{d})$-sparse ruler & $\Theta(\sqrt{d})$ & $\widetilde{O} \Paren{{d/\eps^2} }$~(Thm.~\ref{thm:32}) & $\Omega \Paren{{d^{3/2}/\eps^2}}$~(Thm.~\ref{thm:ruler-lower-bound}) \\[.2cm]
$\Theta(d^\alpha)$-sparse ruler, $\alpha \in [\frac{1}{2},1]$ & $\Theta(d^{\alpha})$ & $\Otilde\Paren{d^{2 - 2\alpha}/\eps^2}$~(Thm.~\ref{thm:general-ruler-bound}) & $\Omega \Paren{{d^{3-3\alpha}/\eps^2}}$~(Thm.~\ref{thm:ruler-lower-bound})\\[.2cm]
\multicolumn{4}{c}{\cellcolor{verylight}{\bf rank-$k$ Toeplitz matrix}} \vspace{0.1cm}\\
Prony's method & $2k$ & $\Otilde \Paren{{1/\eps^2}}$~(Thm.~\ref{thm:prony-inexact-1}) & See below \\[.2cm]
Any non-adaptive method & -- & -- & $\Omega(k/\log k)$~(Thm.~\ref{thm:nonadaptive-lb}) \\[.2cm]
\multicolumn{4}{c}{\cellcolor{verylight}{\bf approximately rank-$k$ Toeplitz matrix}} \vspace{0.1cm}\\
$\Theta(\sqrt{d})$-sparse ruler & $\Theta(\sqrt{d})$ & $\widetilde{O} \Paren{{\min(k^2, d)/\eps^2} }$~(Thm.~\ref{thm:ruler}) & $\Omega(\sqrt{d} k/\epsilon^2)$ (Thm. \ref{thm:ruler-lower-bound}) \\[.2cm]
\specialcell{Algorithm \ref{alg:sft}\\ (sparse Fourier trans. based)} & $\Otilde \Paren{k^2}$
& $\Otilde{\Paren{ k + {1/\eps^2}}}$~(Thm.~\ref{thm:sftT}) &  See below \\[.2cm]
Any adaptive method & -- & -- & $\Omega (k)$~(Thm.~\ref{thm:adaptive-lb})\\
\bottomrule
\end{tabular}
\caption{Overview of our results on estimating a $d$-dimensional Toeplitz covariance from samples to spectral norm error $\eps\|T\|_2$ (Problem~\ref{prob:main}). The above results hold for success probability $2/3$ and we let $\Otilde(\cdot)$ and $\Omegatilde(\cdot)$ hide $\poly (\log d, \log k, \log 1 / \eps)$ factors.
See the corresponding theorems for full statements of the results.
}
\label{table:overviewtable}
\vspace{-0cm}
\end{table}

First, in Section \ref{sec:ruler}, we give a non-asymptotic analysis of a widely used estimation scheme that reads $\Theta(\sqrt{d})$ entries per vector according to a \emph{sparse ruler} \cite{Moffet:1968,romero2016compressive}. 
Sparse ruler methods are important because, up to constant factors, they minimize ESC among all methods that read a \emph{fixed subset} of entries from each vector sample. This is a natural constraint because, in many applications, the subset of entries read is encoded in the signal acquisition hardware. For algorithms that read a fixed subset of entries, it is easy to exhibit Toeplitz covariance matrices that can only be estimated if this subset has size $\Omega(\sqrt{d})$.

While optimal in terms of ESC, at least without additional assumptions on $T$, the VSC of sparse ruler methods was not previously well understood.
We prove a VSC bound that is nearly linear in $d$, which we show is tight up to logarithmic factors. 
Moreover, we introduce a new class of ruler based methods that read anywhere between $\Theta(\sqrt{d})$ and $d$ entries per vector, smoothly interpolating between the ideal ESC of standard sparse rulers and the ideal $O(\log^2d)$ VSC of fully dense sampling. 

Beyond these results, we theoretically confirm a practical observation: sparse ruler methods perform especially well when $T$ is low-rank or close to low-rank. The VSC of these methods decreases with rank, ultimately allowing for total sample complexity \emph{sublinear} in $d$. 
Inspired by this finding, Section \ref{sec:fourier} is devoted to developing methods specifically designed for estimating low-rank, or nearly low-rank, Toeplitz covariance matrices. We develop algorithms based on combining classical tools from harmonic analysis (e.g., the Carath\'{e}odory-Fej\'{e}r-Pisarenko decomposition of Toeplitz matrices) with methods from randomized numerical linear algebra \cite{Woodruff:2014,drineas2016randnla}.
 
 In particular, our work builds on connections between random matrix sketching, leverage score based sampling, and sampling techniques for continuous Fourier signals, which have found a number of recent applications \cite{ChenKanePrice:2016,ChenPrice:2018, ChenPrice:2019, AvronKapralovMusco:2017, AvronKapralovMusco:2019}. 
 Ultimately, we develop an estimation algorithm with total sample complexity depending just \emph{logarithmically} on $d$ and polynomially on $k$ for any $T$ that is approximately rank-$k$. The method reads $\tilde O(k^2)$ entries from each $x^{(\ell)}$ using a fixed pattern, which is constructed randomly. To the best of our knowledge, this approach is the first general Toeplitz covariance estimation algorithm with ESC $< \sqrt{d}$. It provides a potentially powerful alternative to sparse ruler based methods, so tightening our results (potentially to linear in $k$) is an exciting direction for future work. In Section~\ref{sec:lower} we demonstrate that a linear dependence on $k$ is unavoidable for total sample complexity, even for adaptive sampling algorithms.
 
 \subsection{Comparison to prior work}\label{sec:comp}
 An in-depth discussion of prior work is included in Section \ref{sec:our_results}. Our results on general \emph{full-rank} Toeplitz covariance matrices give much tighter bounds on VSC than prior work. While some results address the variance of estimators (including those based on sparse rulers) for each entry in $T$ \cite{ArianandaLeus:2012}, the challenge is that obtaining a bound on $\|T - \tilde{T}\|_2$ requires understanding potentially strong correlations between the errors on different matrix entries. A naive analysis leads to a VSC bound of $\Omega(d^2/\epsilon^2)$ for both full samples and sparse ruler samples, which is significantly worse than our respective bounds of $\tilde{O}(1/\epsilon^2)$ and $\tilde{O}(d/\epsilon^2)$ proven in Theorems \ref{thm:linear} and \ref{thm:32}. Most closely related to our work is that of Qiao and Pal \cite{qiao2017gridless}. While they consider a different error metric, their analysis yields VSC $O(d/\epsilon^2)$ for sparse ruler samples to estimate $T$ to relative error in the Frobenius norm. In this case, $O(d/\epsilon^2)$ samples can be obtained without considering correlations between the errors on different matrix entries. Qiao and Pal also consider the case when $T$ is approximately low-rank, and like us show that an even lower ESC is possible in this setting. They actually give ESC $O(\sqrt{k}),$ which goes beyond any of our bounds, however their analysis requires strong  assumptions on the Carath\'{e}odory-Fej\'{e}r-Pisarenko decomposition of $T$. 
 
For low-rank or nearly low-rank covariance matrices, our total sample complexity  bounds depend  polynomially in the rank $k$ and \emph{sublinearly on $d$}, in fact just logarithmically in Theorems \ref{thm:prony-inexact-1} and \ref{thm:sftT}. When $T$ is not assumed to be Toeplitz, TSC must depend linearly in $d$, even when $T$ is low-rank. This dependence is reflected in prior work, which obtains TSC bounds of at best $O(dk/\epsilon^2)$, for any ESC between $2$ and $d$ \cite{GonenRosenbaumEldar:2016}.  Our work critically takes advantage of Toeplitz and low-rank structure simultaneously to surpass such bounds.




\section{Discussion of Results}\label{sec:our_results}
Formally, we study sampling schemes and recovery algorithms for the following problem:
\begin{problem}[Covariance estimation]\label{prob:main} For a positive semidefinite matrix $T \in \R^{d \times d}$, given query access to the entries of i.i.d. samples drawn from a  $d$-dimension normal distribution with covariance $T$, i.e., $x^{(1)},\ldots, x^{(n)} \sim \mathcal{N}(0, T)$, return $\tilde T$ satisfying, with probability $\ge 1-\delta$:
	\begin{align*}
	\norm{T - \tilde{T}}_2 \le \eps \norm{T}_2,
	\end{align*}
	where $\norm{\cdot }_2$ denotes the operator (spectral) norm. The \emph{total sample complexity (TSC) } of computing $\tilde{T}$ is the total number of \textbf{entries} that the algorithm must read from $x^{(1)},\ldots, x^{(n)}$, combined. 
\end{problem}
As discussed, in addition to total sample complexity, we would like to understand the tradeoff between the number of vector samples used by an algorithm and the maximum number of entries viewed in each vector individually (VSC and ESC, respectively). Problem \ref{prob:main} assumes a normal distribution for simplicity, but our results also hold when the $x^{(j)}$ are distributed as $x^{(j)} \sim T^{1/2} y^{(j)}$, where $y^{(j)}$ is an isotropic sub-gaussian random variable.
We also conjecture that similar bounds hold for more general classes of random variables.
Finally, we note that when $T$ is Toeplitz, our algorithms will return  $\tilde{T}$ that is also Toeplitz, which is useful in many applications.

\subsection{What is known}
\label{sec:whats_known}
Standard matrix concentration results can be used to bound the error of estimating a generic covariance matrix $T$ when $\tilde T$ is set to the empirical covariance $\frac{1}{n}\sum_{j=1}^n x^{(j)}{x^{(j)}}^T$  \cite{vershynin2010introduction}. Computing this estimate requires reading all $d$ entries of each sample, yielding the following standard bound:
\begin{claim}[General covariance estimation]\label{thm:naive}
	For any positive semidefinite $T \in \R^{d \times d}$, Problem \ref{prob:main} can be solved with vector sample complexity $O\left ( \frac{d + \log(1/\delta)}{\eps^2} \right )$, entry sample complexity $d$, and total sample complexity $O \left ( \frac{d^2 + d \log(1/\delta)}{\eps^2} \right )$. These bounds are achieved by setting $\tilde{T} = \frac{1}{n}\sum_{j=1}^n x^{(j)}{x^{(j)}}^T$.
\end{claim}
Claim \ref{thm:naive} applies for generic covariance matrices, which are specified by $\Theta(d^2)$ parameters instead of the $\Theta(d)$ required for Toeplitz matrices. For other matrices specified by $\Theta(d)$ parameters, it is possible to improve on Claim \ref{thm:naive}. When $T$ is diagonal, the entries of a sample $x^{(j)} \sim \mathcal{N}(0,T)$ are independent Gaussians with variances equal to $T$'s diagonal. So each diagonal entry can be approximated to relative error $1 \pm \eps$ with probability $\geq 1 - \delta/d$ using $O \left ({\log(d/\delta)}/{\eps^2} \right )$ samples. Applying a union bound, Problem \ref{prob:main} can be solved with $O \left ({\log(d/\delta)}/{\eps^2} \right )$ full vector samples.

Closer to our setting, it is possible to prove the same bound for any \emph{circulant} covariance matrix: i.e., a Toeplitz matrix where each row is a cyclic permutation of the first. To see why, we use the fact that any circulant matrix, $T$, can be written as $T = F D F^*$ where $F \in \C^{d \times d}$ is the DFT matrix and $D$ is diagonal. If we transform a sample $x \sim \mathcal{N}(0,T)$ by forming $F^*x$, we obtain a random vector with diagonal covariance, allowing us to apply the diagonal result.
Overall we have:
\begin{claim}[Diagonal and circulant covariance estimation]\label{thm:diagonal}
	For any diagonal or circulant positive semidefinite $T \in \R^{d \times d}$,  Problem \ref{prob:main} can be solved with vector sample complexity $O\left (\frac{\log(d/\delta)}{\eps^2} \right )$, entry sample complexity $d$, and total sample complexity $O \left ( \frac{d \log(d/\delta)}{\eps^2} \right)$. These bounds are achieved by setting $\tilde{T} = \diag\left(\frac{1}{n}\sum_{j=1}^n x^{(j)}{x^{(j)}}^T\right)$ when $T$ is diagonal, or $\tilde{T} = F \diag\left(\frac{1}{n}\sum_{j=1}^n F^*x^{(j)}{x^{(j)}}^TF\right)F^*$ when $T$ is circulant. Here $\diag(A)$ returns the matrix $A$ but with all off-diagonal entries set to 0. 
\end{claim}
See Appendix \ref{app:additional} for a more formal proof of this known result.
Unfortunately, Claim \ref{thm:diagonal} does not extend to general Toeplitz covariance matrices. While all Toeplitz matrices have some Fourier structure (as will be discussed later, they can be diagonalized by Fourier matrices with ``off-grid'' frequencies), they cannot be diagonalized using a single known basis like circulant matrices. 

\subsection{What is new: full rank matrices}
Our first result is that Claim \ref{thm:diagonal} can nevertheless be matched for Toeplitz matrices using a very simple algorithm: we compute the empirical covariance matrix $\bar T$ of $x^{(1)}, \ldots, x^{(n)}$ and then form $\tilde T$ by averaging its  diagonals to obtain a Toeplitz matrix (see Algorithm \ref{alg:sparse}). Surprisingly, unlike the method for circulant matrices, this algorithm does not explicitly use the Fourier structure of the Toeplitz matrix, but achieves a similar sample complexity, up to logarithmic factors:
\begin{theorem}[Near linear sample complexity]\label{thm:linear}
For any positive semidefinite Toeplitz matrix $T \in \R^{d \times d}$, Problem \ref{prob:main} can be solved by Algorithm \ref{alg:sparse} 
with vector sample complexity $O \left ( \frac{ \log d\log \left (d/\eps \delta\right )}{\eps^2} \right )$, entry sample complexity $d$, and total sample complexity $O \left ( \frac{ d  \log d\log \left (d/\eps \delta\right )}{\eps^2} \right )$.
\end{theorem}
The averaging method used for Theorem \ref{thm:linear} has been suggested before as a straightforward way to improve a naive sample covariance estimator \cite{cai2013optimal}. Our theorem shows that this improvement can be very significant. For example, when $T$ is the identity, the sample covariance truly requires ${\Omega}(d)$ vector samples to converge. Averaging reduces this complexity to $O(\log^2 d)$.

\begin{wrapfigure}{l}{0.5\textwidth}
	\vspace{-.5em}
	\captionsetup{width=.99\linewidth}
	\centering
	\includegraphics[width=.45\textwidth]{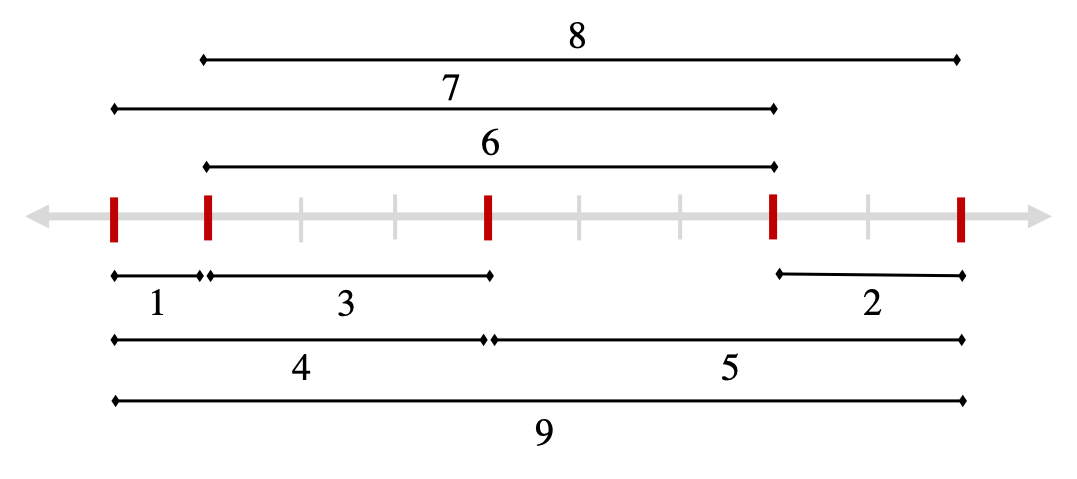}
	\caption{\small A sparse ruler $R = \{1,2,5,8,10\}$ for $d = 10$. Here $|R| = 5$, but the ruler still represents all distances in $0, \ldots, d-1$, so we can use it to select entry samples for estimating any $10\times 10$ Toeplitz covariance matrix.}
	\label{fig:sparse_ruler}
	\vspace{-.5em}
\end{wrapfigure}
Moreover, averaging gives a natural strategy for solving the Toeplitz problem with entry sample complexity $o(d)$. In particular, dating back to the work of \cite{Moffet:1968}, who sought to minimize the number of receivers in directional antenna arrays, signal processing applications have widely applied the averaging method with less than $d$ samples taken from $x\sim\mathcal{N}(0,T)$ according to a so called \textbf{sparse ruler} \cite{romero2016compressive}. 

The idea is elegant and ingenious: a symmetric Toeplitz matrix $T$ is fully described by a vector $a \in \R^{d}$ where $a_i$ is the value on the $i^{\text{th}}$ diagonal of $T$, i.e., the covariance of samples spaced $i$ steps apart. So, to estimate each entry in $a$, we only need to view a subset of entries $R\subseteq [d]$ from $x\sim \mathcal{N}(0,T)$ such that there are two entries, $u,v$ in $R$ at distance $s$ apart, for \emph{every distance} $s \in 0, \ldots,d-1$. If we average $x^{(\ell)}_ux^{(\ell)}_v$ for $\ell \in [n]$, we will eventually converge on an estimate of $a_s$. It was noticed by \cite{Moffet:1968} and others \cite{PillaiBar-NessHaber:1985} that such a subset $R$ can always be found with $|R| =  \Theta(\sqrt{d})$. 
Formally we define:
\begin{definition}[Ruler]\label{def:ruler}
	A subset $R \subseteq [d]$ is a \emph{ruler}  if for all $s = 0, \ldots, d - 1$, there exist $j,k \in R$ so that $s =|j - k|$. We let $R_s \eqdef \{(j,k) \in R \times R: |j-k| = s\}$ denote the set of (ordered) pairs in $R \times R$ with distance $s$. We say $R$ is ``sparse'' if $|R| < d$. 
\end{definition}
Sparse rulers are also called ``sparse linear arrays'' in the literature \cite{WuZhuYan:2016} and are closely related to other constructions like Golumb rulers, which are also used in signal processing applications \cite{Babcock:1953}. 
It is known that any sparse ruler must have $|R| \gtrsim 1.557 \sqrt{d}$ \cite{ErdosGal:1948,Leech:1956}, and it is possible to nearly match  this bound for all $d$ with a linear time constructible ruler:
\begin{claim}[$\Theta(\sqrt{d})$ Sparse Ruler]\label{clm:sparse}
	For any $d$, there is an explicit ruler $R$ of size $|R| = 2\lceil\sqrt{d}\rceil-1$. 
\end{claim}
\begin{proof}
	It suffices to take $R = \{1, \ldots, \lceil\sqrt{d}\rceil \} \cup \{d, d - \lceil\sqrt{d}\rceil, d - 2\lceil\sqrt{d}\rceil, \ldots, d - (\sqrt{d}-2)\lceil\sqrt{d}\rceil\}$.
\end{proof}
Using samples taken from $x^{(1)}, \ldots, x^{(\ell)}$ according to a fixed $\Theta(\sqrt{d})$ sparse ruler and averaging to estimate $a$ (formalized in Algorithm \ref{alg:sparse}) gives a $\tilde{T}$ that clearly converges to $T$ as $n\rightarrow \infty$. However, while commonly used for its potential to save on hardware cost, power, and space in applications, it has been unknown how much the sparse ruler strategy sacrifices in terms of \emph{vector sample complexity} compared to looking at all entries in $x^{(1)}, \ldots, x^{(\ell)}$. We bound this complexity:

\begin{theorem}[Sparse ruler sample complexity]\label{thm:32}
	For any positive semidefinite Toeplitz matrix $T \in \R^{d \times d}$, Problem \ref{prob:main} can be solved by Algorithm \ref{alg:sparse} using any $\Theta(\sqrt{d})$ sparse ruler (e.g., the construction from Claim \ref{clm:sparse})
	with vector sample complexity $O \left ( \frac{d\log \left (d/\eps \delta\right )}{\eps^2} \right )$, entry sample complexity $\Theta(\sqrt{d})$, and total sample complexity $O \left ( \frac{ d^{3/2} \log \left (d/\eps \delta\right )}{\eps^2} \right )$.
\end{theorem}

More generally, in Section \ref{sec:ruler} we introduce a new class of rulers with sparsity $\Theta(d^\alpha)$ for $\alpha \in [1/2,1]$. These rulers give an easy way to trade off entry sample complexity and vector sample complexity, with bounds interpolating smoothly between Theorem \ref{thm:linear} and Theorem \ref{thm:32}. 

Note that the overall sample complexity  of Theorem \ref{thm:32} is worse than that of Theorem \ref{thm:linear} by roughly a $\sqrt{d}$ factor.
This is in fact inherent: we provide a lower bound in Theorem \ref{thm:ruler-lower-bound}.
\begin{theorem}[Informal, see Theorem~\ref{thm:ruler-lower-bound}]
For any $\eps > 0$ sufficiently small, any algorithm that takes samples from a ruler $R$ with sparsity $|R| = O(\sqrt{d})$ and solves Problem~\ref{prob:main} with success  probability $1-\delta \ge 1/10$ requires vector sample complexity $\Omega (d / \eps^2)$.
\end{theorem}
Thus, in terms of total sample complexity, the ruler-based algorithms are worse than taking full samples. 
However, the ``hard'' case for the sparse ruler methods appears to be when $T$ is close to an identity matrix -- i.e., it is nearly rank $d$. In practice, it is much more common for $T$ to be nearly $k$ rank for some $k \ll d$. For example, this is the case in DOA estimation when an antenna array is detecting the direction of just $k$ different signal sources \cite{KrimViberg:1996, ShakeriArianandaLeus:2012,ChiEldarCalderbank:2013}. Experimentally, the performance of sparse ruler methods  suffers less when $T$ is low-rank. In fact, in terms of TSC, they often significantly \emph{outperform} algorithms that look at all entries in each $x^{(\ell)}$ (see Figure \ref{fig:sparse_vs_full}). 

\begin{figure}[h]
	\centering
	\captionsetup{width=1\linewidth}
	\begin{subfigure}[t]{0.38\textwidth}
		\centering
		\includegraphics[width=1\textwidth]{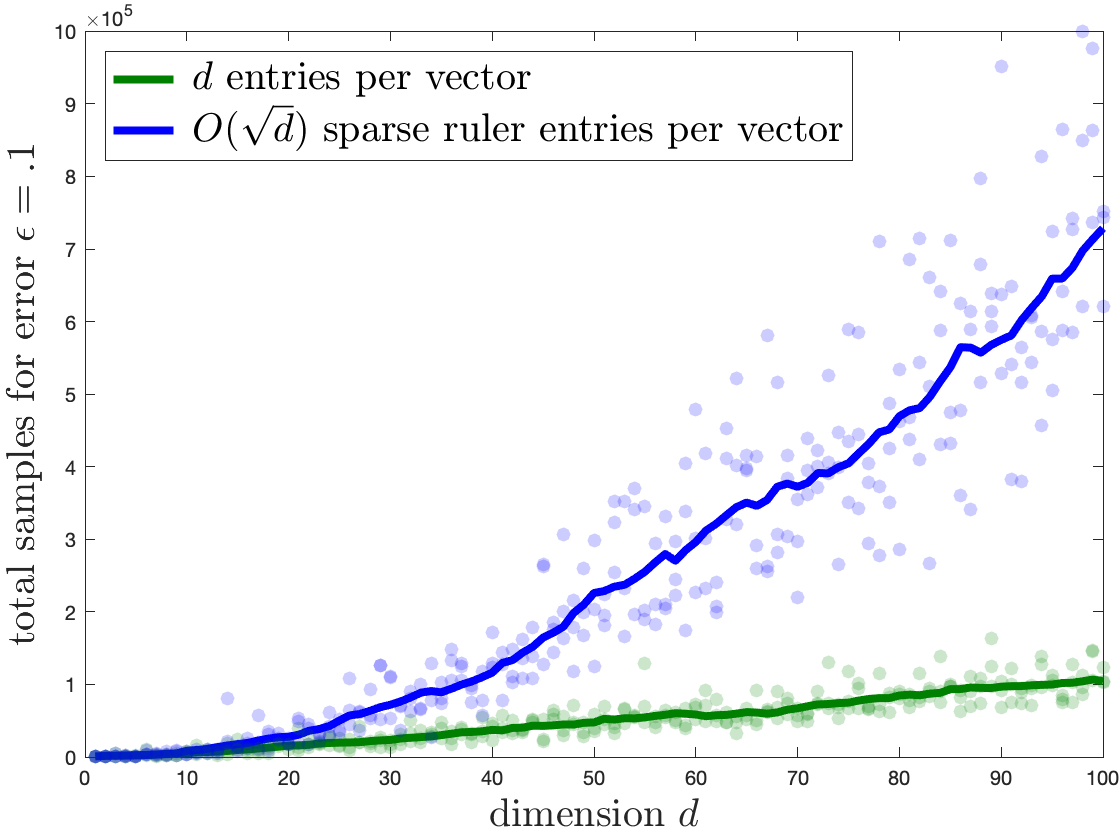}
		\caption{\small Full rank covariance.}
		\label{fig:full_rank_case}
	\end{subfigure}
	~
	\begin{subfigure}[t]{0.38\textwidth}
		\centering
		\includegraphics[width=1\textwidth]{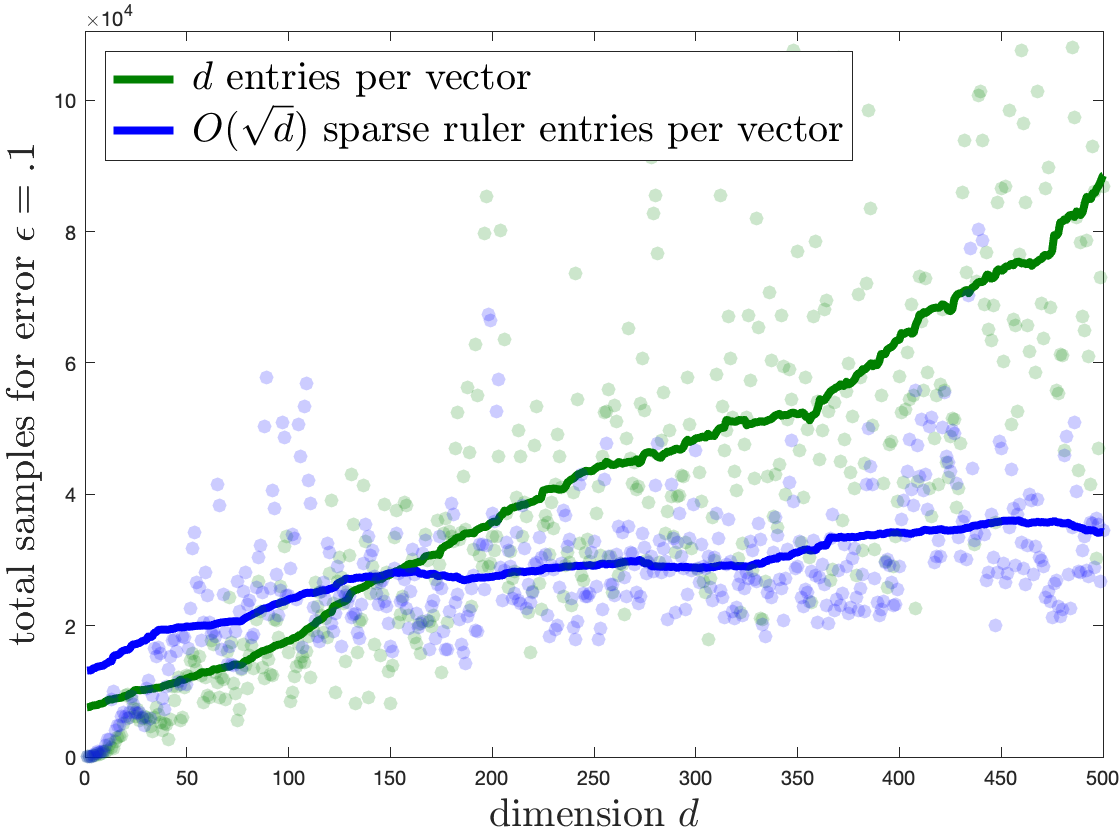}
		\caption{\small Fixed rank-$k$ covariance.}
		\label{fig:low_rank_case}
	\end{subfigure}
~
	\begin{subfigure}[t]{0.19\textwidth}
	\includegraphics[width=.89\textwidth]{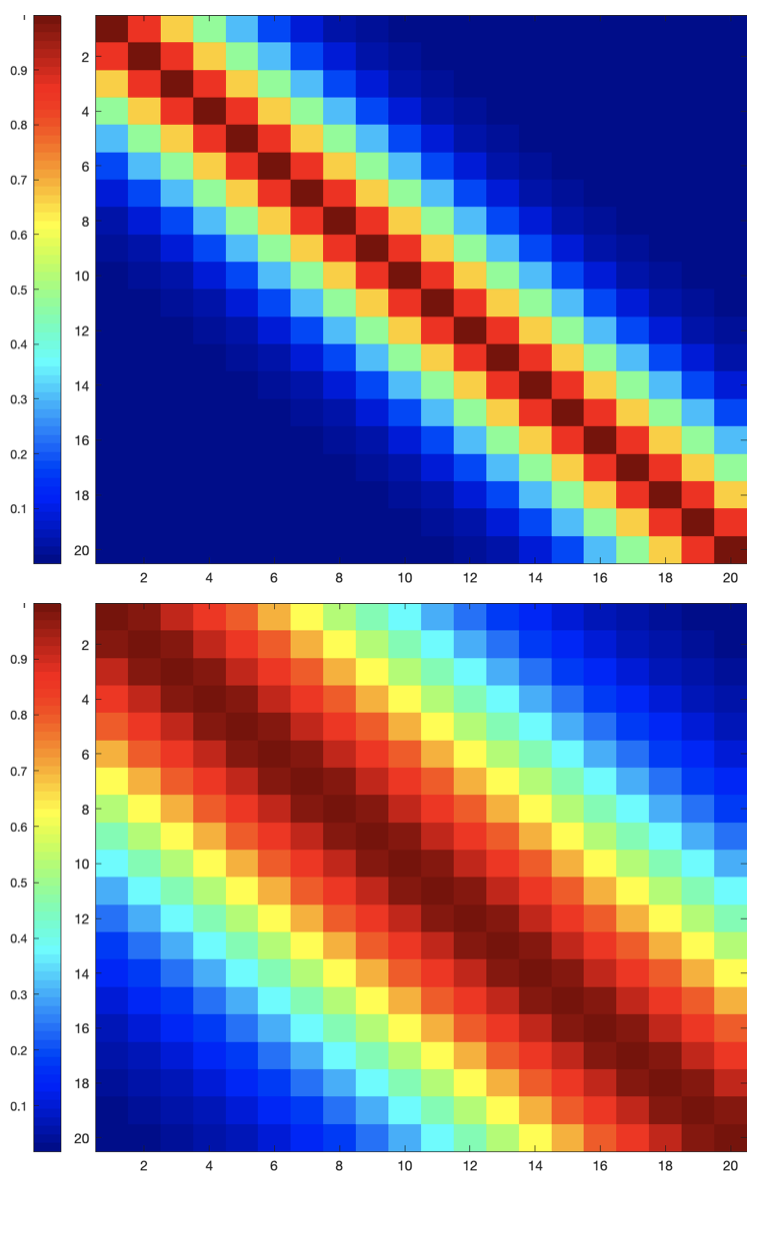}
	\caption{Toeplitz rank.}
	\label{fig:low_rank_case}
	\end{subfigure}
	\vspace{-.5em}
	\caption{\footnotesize These plots depict the total number of samples required to approximate a Toeplitz covariance matrix $T \in \R^{d\times d}$ to $\eps$ accuracy in the spectral norm, as a function of $d$. The dots represent individual trials and the solid lines are moving averages of the dots. In (a), each matrix tested was full rank (i.e., rank $d$) while in (b) we held rank constant at $30$ for all $d$. As predicted by Theorems \ref{thm:linear} and \ref{thm:32}, the averaging algorithm that looks at all $d$ entries per vector outperforms the sparse ruler method by approximately a $\sqrt{d}$ factor when $T$ is full rank. However, when rank is fixed, the sparse ruler method is actually more sample efficient for large $d$, a finding justified in Theorem \ref{thm:algSparse}. The images in (c) visualize a high rank (top) and low-rank (bottom) Toeplitz matrix.}
 	\label{fig:sparse_vs_full}
 	\vspace{-1em}
\end{figure}
 
\subsection{What is new: low-rank matrices}
We confirm these experimental findings theoretically by proving a tighter bound for the sparse ruler estimation algorithm when $T$ is low-rank (or nearly low-rank). In terms of total sample complexity, our bound {surpasses} the algorithm that looks at all $d$ entries in each $x^{(\ell)}$. In fact, for $k \lesssim d^{1/4}$ we show that \emph{sublinear} total sample complexity is possible. This isn't even possible for diagonal covariance matrices that are close to low-rank.\footnote{To see that this is the case, consider diagonal $T$  with just a single entry placed in a random position equal to $1$ and the rest equal to $0$. Then $\Omega(d)$ samples are necessary to identify  this entry, even though $T$ is rank $1$.} Overall we obtain:

%


\begin{theorem}[Sublinear sparse ruler sample complexity]\label{thm:ruler}
		For any positive semidefinite Toeplitz matrix $T \in \R^{d \times d}$ with rank $k$, Problem \ref{prob:main} can be solved by Algorithm \ref{alg:sparse} using the $\Theta(\sqrt{d})$ sparse ruler from Claim \ref{clm:sparse}
		with vector sample complexity $O \left (\frac{\min (k^2,d)\log \left ({d}/{\eps \delta} \right )}{\eps^2} \right )$, entry sample complexity $\Theta(\sqrt{d})$, and total sample complexity $O \left (\frac{\sqrt{d}\min (k^2,d)\log \left ({d}/{\eps \delta} \right )}{\eps^2} \right )$.
\end{theorem}
Importantly, Theorem \ref{thm:ruler} also holds when $T$ is \emph{close} to rank-$k$, which is the typical case in practice. For example, if $T_k$ is $T$'s optimal $k$-rank approximation given by the singular value decomposition, the bound stated here also holds when $\frac{\sqrt{d}}{k^2} \cdot \norm{T-T_k}_F^2 \le \norm{T}_2^2$ or $\frac{d}{k^2} \cdot \norm{T - T_k}_2^2 \leq \norm{T}_2^2$. See Section \ref{sec:ruler} for a formal statement and extensions to rulers with sparsity between $\Theta(\sqrt{d})$ and $d$. 
By achieving total sample complexity \emph{sublinear} in $d$, Theorem \ref{thm:ruler} raises an intriguing question.

\begin{center}
	\textbf{Can we learn a low-rank Toeplitz matrix with sample complexity independent of $d$?}
\end{center}

In particular, the $O(\sqrt{d})$ dependence on the total sample complexity in Theorem \ref{thm:ruler} arises from the sparsity of the ruler. While sparse rulers with less than $\sqrt{d}$ elements do not exist, we might hope to design alternative approaches for selecting entries from $x^{(1)}, \ldots, x^{(n)}$ which are specifically tailored to the low-rank setting. Doing so is the focus of our remaining results and requires more explicitly leveraging the Fourier structure in Toeplitz covariance matrices. 

\begin{wrapfigure}{r}{0.45\textwidth}
	\vspace{-1em}
	\captionsetup{width=.95\linewidth}
	\centering
	\includegraphics[width=.45\textwidth]{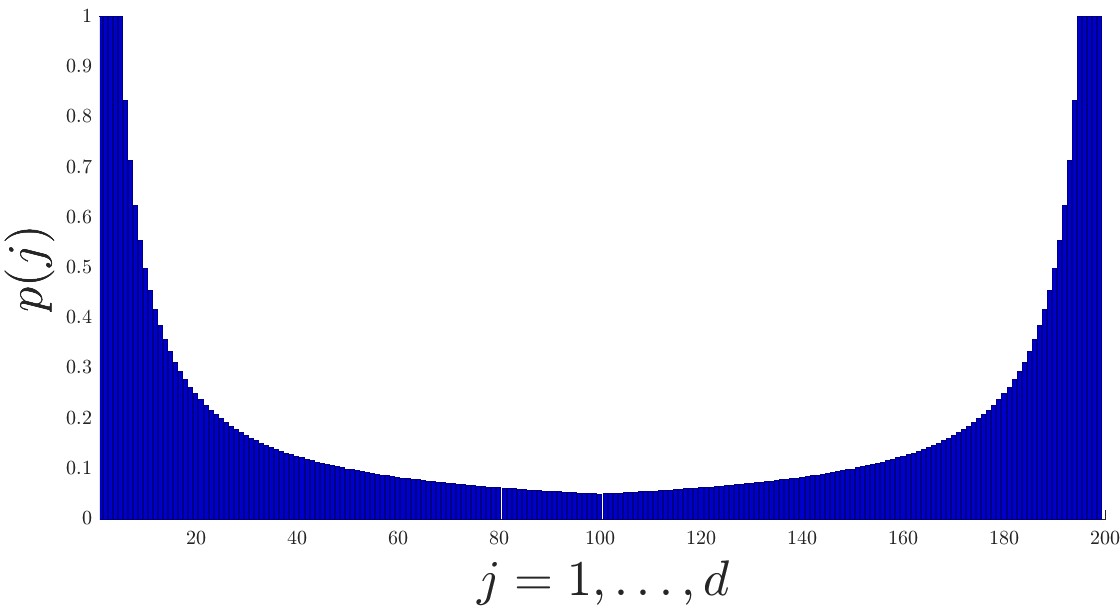}
	\caption{\small In Section \ref{sec:fourier}, we introduce a strategy that chooses a fixed sampling pattern by independently choosing each $j \in 1, \ldots, d$ with probability $p(j)$, for an explicit function $p$ (see Algorithm \ref{alg:sft}). Roughly, $p(j)\sim 1/\min(d-j, j)$, so samples closer to the edges of $1,\ldots, d$ are selected with higher probability. This produces a pattern that intuitively looks like many sparse rulers, but can be much sparser (see Figure \ref{fig:sampling_patterns}).}
	\label{fig:lev}
	\vspace{-1em}
\end{wrapfigure}

Specifically, the powerful Vandermonde decomposition theorem \cite{caratheodory1911zusammenhang,pisarenko1973retrieval} (see Fact \ref{lem:fourier}) implies a connection between low-rank Toeplitz matrices and Fourier sparse functions: $T$ is rank-$k$ if and only if its columns are spanned by $k$ vectors obtained from evaluating $k$ complex sinusoids on the integer grid $0, \ldots, d-1$. I.e., every column of $T$ has a $k$-sparse Fourier representation (with arbitrary frequencies in the range $[0,1]$) and it follows that any sample $x^{(\ell)} \sim \mathcal{N}(0,T)$ does as well.

This observation allows for a resolution of the question above. 
In particular, Prony's method \cite{de1795essai} implies that any sum of $k$ complex sinusoids can be recovered from exactly $2k$ values. Accordingly, when $T$ is exactly rank-$k$, we can select $2k$ arbitrary entries in $x^{(\ell)}$, use Prony's method to recover the remaining entries, and run our algorithm from Theorem \ref{thm:linear} to approximate $T$. This approach achieves sample complexity linear in $k$ and logarithmic in $d$. In Section \ref{sec:fourier} we improve the logarithmic dependence to obtain a result that is fully independent of the dimension:
\begin{theorem}[Dimension indep. complexity via Prony's method]
\label{thm:prony-exact-roots}
	For any positive semidefinite Toeplitz matrix $T \in \R^{d\times d}$ with rank $k$, Problem \ref{prob:main} can be solved by Algorithm \ref{alg:exactSFT-1} with vector sample complexity $O \left (\frac{\log (k / \delta)}{\eps^2} \right )$, entry sample complexity $\Theta(k)$, and total sample complexity $O \left (\frac{k\log (k / \delta)}{\eps^2} \right )$.
\end{theorem}

Unfortunately, Prony's method is notoriously sensitive to noise \cite{SarkarPereira:1995}: it fails to recover a sum of $k$ complex sinusoids from $2k$ values when those values are perturbed by even a very small amount of noise. The effect in our setting is that Theorem \ref{thm:prony-exact-roots} relies heavily on $T$ being \emph{exactly} rank-$k$.

To handle the more typical case when $T$ is only \emph{close} to low-rank, we first prove that in this setting, $T$ is still closely approximated by a $k$-sparse Fourier representation. We then build on recent work on provably interpolating Fourier sparse functions with noise \cite{ChenKanePrice:2016,ChenPrice:2018} to show how to recover this representation. Our proofs leverage several tools from randomized numerical linear algebra, including column subset selection bounds and the projection-cost preserving sketches of \cite{cohen2015dimensionality}.

While this approach is somewhat involved (see Section \ref{sec:fourier} for the full development) a key takeaway is \emph{how it collects samples} from each $x^{(\ell)}$. In contrast to ruler based methods, we generate a fixed pattern based on randomly sampling a subset of entries in $1, \ldots, d$ (see Figure \ref{fig:lev}). The randomized strategy concentrates samples near the edges of $1,\ldots, d$, producing a pattern that is intuitively similar to many sparse ruler designs, but much sparser. 
Overall, we obtain:
\begin{theorem}[Sparse Fourier trans. sample complexity]\label{thm:sftT}
	For any positive semidefinite Toeplitz matrix $T \in \R^{d\times d}$ with rank $k$, Problem \ref{prob:main} can be solved with failure probability $\delta \le \frac{1}{10}$ by Algorithm \ref{alg:sft} with vector sample complexity $O \left (\frac{\log d \log(d/\eps)}{\eps^2} + k \right )$ and entry sample complexity $\Theta(k^2 \log k \log(d/\eps))$. The same bounds hold when $T$ is close to low-rank, with $\left (\norm{T - T_k}_2 + \frac{\tr(T - T_k)}{k} \right ) \cdot  \tr(T)  = O(\eps^2 \Norm{T}_2^2)$.
\end{theorem}
See Section~\ref{sec:fourier} for a full statement of Theorem \ref{thm:sftT}. As a corollary,
in Theorem \ref{thm:sft}, we also give a bound with sample complexity depending polynomially on the \emph{stable rank} of $T$, $s = \frac{\tr(T)}{\norm{T}_2}$.

\begin{figure}[h]
	\centering
	\captionsetup{width=1\linewidth}
	\begin{subfigure}[t]{0.45\textwidth}
		\centering
		\includegraphics[width=1\textwidth]{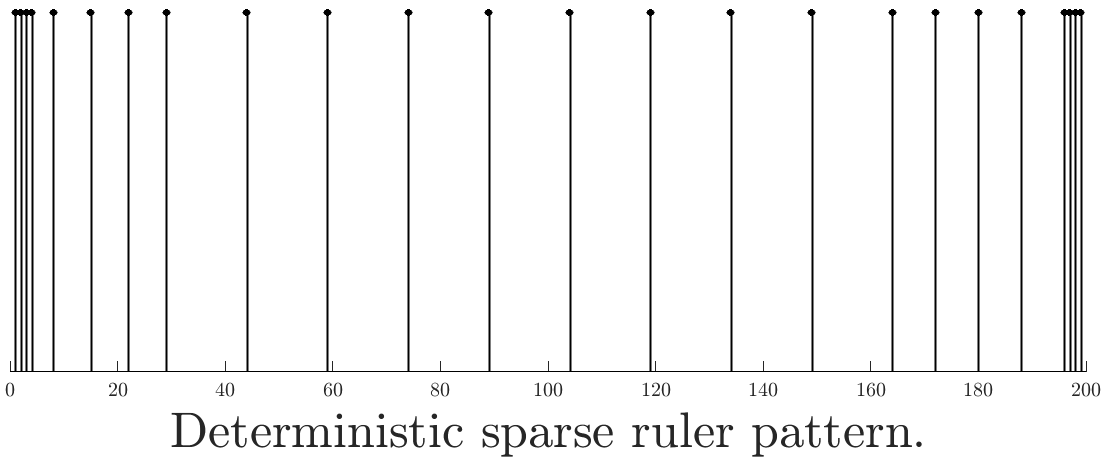}
	\end{subfigure}
\hspace{3em}
	\begin{subfigure}[t]{0.45\textwidth}
		\includegraphics[width=1\textwidth]{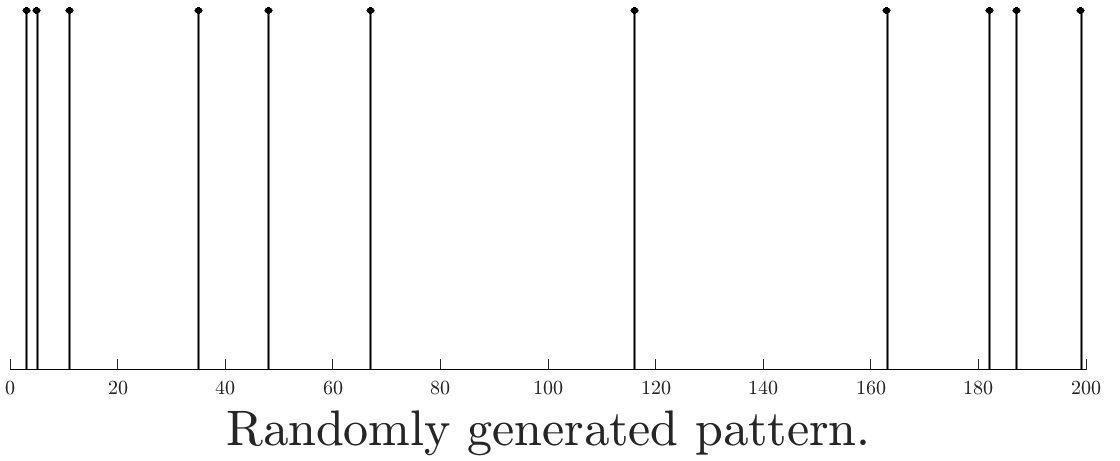}
	\end{subfigure}
	\vspace{-.5em}
	\caption{\footnotesize The left image depicts the minimum size ($s = 24$)  sparse ruler for $d = 199$ \cite{Wichmann:1963}. The right image depicts a randomly sampled set of $s = 12$ nodes that can be used in the algorithms from Section \ref{sec:fourier}. After random generation, the same pattern is used for selecting entries in each vector $x^{(1)}, \ldots, x^{(n)} \sim \mathcal{N}(0,T)$. Like the sparse ruler, the random pattern concentrates nodes towards the end points of $1, \ldots, d$. See Figure \ref{fig:lev} for more details.}
		\label{fig:sampling_patterns}
\end{figure}

We believe that improving Theorem \ref{thm:sftT} is an interesting and potentially impactful direction for future work. For example, it may be possible to improve the dependence on $k$ in the total sample complexity. We do give two incomparable lower bounds which argue that at least a linear dependence on $k$ is unavoidable.
The first is against any non-adaptive algorithm:
\begin{theorem}[Informal, see Theorem~\ref{thm:nonadaptive-lb}]
Any algorithm that solves Problem~\ref{prob:main} for any rank-$k$ Toeplitz matrix $T\in \R^{d\times d}$ by inspecting the entries of $n$ samples $x^{(1)}, \ldots, x^{(n)}$ non-adaptively (i.e. the algorithm chooses, possibly randomly, a fixed set of positions to be read before seeing the samples), has total sample complexity $\Omega(k/\log k)$.
\end{theorem}
In particular, all of our upper bounds use non-adaptive sampling strategies, and so this theorem shows that they cannot be improved beyond potentially a $\tilde O(k^2)$ factor.
We further demonstrate that even adaptive algorithms cannot avoid this lower bound, if a small amount of noise is added:
\begin{theorem}[Informal, see Theorem~\ref{thm:adaptive-lb}]
Any algorithm that solves Problem~\ref{prob:main} for any approximately rank-$k$ Toeplitz matrix $T\in \R^{d\times d}$, by inspecting the entries of $n$ samples in a potentially adaptive manner, has total sample complexity $\Omega(k)$.
\end{theorem}
For the formal definition of what it means to be ``approximately rank-$k$'' in this context, see Theorem~\ref{thm:adaptive-lb}.
The proof of this is inspired by lower bounds on sparse Fourier transforms and goes via information theory.
We highlight one particularly interesting consequence of these bounds. 
By taking $k = \Omega (d)$ in the above lower bounds, one immediately obtains that no algorithm (even potentially adaptive), can avoid total sample complexity $\Omega (d)$ for Problem~\ref{prob:main}, without any rank constraints.
Thus, in the metric of TSC, with no rank constraints, the simple algorithm of Theorem \ref{thm:linear} is in nearly optimal for Problem~\ref{prob:main}.
See Section~\ref{sec:lower} for a more detailed discussion.

\subsection{Related work}
\label{sec:prior_work}
The problem of estimating Toeplitz covariance matrices using all $d$ entries of each sample $x^{(\ell)}$ is a classical statistical problem. Methods for finding maximum likelihood or minimax estimators have been studied in a variety of settings \cite{Burg:1967, KayMarple:1981, burg1982estimation} and we refer the reader to \cite{BartonSmith:1997} for an overview. 
These methods seek to extract a Toeplitz matrix from the sample covariance to improve estimation performance. Existing approaches range from very simple (e.g., diagonal averaging as in our Theorem \ref{thm:linear} \cite{cai2013optimal}) to more complex (e.g., EM-like methods \cite{MillerSnyder:1987}). We are not aware of any work that establishes a generic non-asymptotic bound like Theorem \ref{thm:linear}.

Work on Toeplitz covariance estimation with the goal of minimizing entry sample complexity has 
%
 focused on sparse ruler based sampling.
  This approach has been known for decades \cite{Moffet:1968,PillaiBar-NessHaber:1985} and is widely applied in signal processing applications: we refer the reader to \cite{romero2016compressive} for an excellent overview. 
 There has been quite a bit of interest in finding rulers of minimal size -- asymptotically, $\Theta(\sqrt{d})$ is optimal, but in many applications minimizing the leading constant is important \cite{Leech:1956,Wichmann:1963,RufSwiftTanner:1988}.
There has been less progress on understanding rates of convergence for sparse ruler methods. \cite{WuZhuYan:2017} provides some initial results on bounding $\|T - \tilde{T}\|_2$, but we are not aware of any strong relative error bounds like those of Theorem \ref{thm:32} or Theorem \ref{thm:ruler}. \cite{qiao2017gridless} gives the closest analog, proving non-asymptotic vector sample complexity bounds for sparse ruler sampling. They measure error in terms of the $\ell_2$ norm of the first column of $T$, which can be translated to a Frobenius norm bound of the form $\norm{T - \tilde T}_F \le \epsilon \norm{T}_F$, but not directly to a spectral norm bound.

More related to the algorithmic contributions of Section \ref{sec:fourier}, are methods that make the additional assumption that $T$ is low-rank, or close to low-rank. The most well-known algorithms in this setting are the MUSIC \cite{Schmidt:1981, Schmidt:1986,pisarenko1973retrieval} and ESPRIT methods \cite{RoyKailath:1989}, although other approaches have been studied \cite{BreslerMacovski:1986a,FengBresler:1996,MishaliEldar:2009}. As discussed, Toeplitz $T$ is low-rank exactly when its columns are Fourier sparse functions. Accordingly, such methods are closely related to sparse recovery algorithms for ``off-grid'' frequencies. 
Most theoretical bounds for ``off-grid'' sparse recovery require strong \emph{separation} assumptions that no frequencies are too close together \cite{TangBhaskarShah:2013,CandesFernandez-Granda:2014,TangBhaskarRecht:2015,BoufounosCevherGilbert:2015}. 
A recent exception is \cite{ChenKanePrice:2016}, which we build on in this work since a lack of separation assumptions is essential to obtaining bounds for general low-rank Toeplitz matrices. 
Under separation assumptions on the frequencies in $T$'s Vandermonde decomposition, we note that \cite{qiao2017gridless} gives an algorithm with $O(\sqrt{k})$ ESC and $\poly(d,k,\epsilon)$ VSC when $T$ is close to rank-$k$, using a sparse recovery based approach.


Finally, we note that covariance estimation where $<d$ entries of $x^{(l)}$ are available has been studied without a Toeplitz assumption \cite{GonenRosenbaumEldar:2016}. Without Toeplitz and thus Fourier structure, this work is more closely related to the matrix completion problem \cite{CandesRecht:2009} than our setting.

\subsection{Open questions}
Our work initiates a comprehensive study of non-asymptotic complexity bounds for the recovery of Toeplitz covariance matrices, and leaves a number of open questions. 

\smallskip\noindent\textbf{Efficiency.} All of our estimators run in polynomial time except for Algorithm \ref{alg:sft}, which is used to obtain the bound of Theorem \ref{thm:sftT}.
An interesting open question is if one can use techniques from e.g.,~\cite{ChenKanePrice:2016} to obtain polynomial time estimators for this setting.

\smallskip\noindent\textbf{Tighter bounds.} While we present non-asymptotic sample complexity upper and lower bounds for many of the settings considered in this paper, and they are within polynomial factors of each other, it is an interesting open question to exactly characterize the sample complexities of these problems and to fully understand the tradeoffs between entry and vector sample complexity.
Moreover, it is unclear if the form of the tail error we obtain for approximately rank-$k$ Toeplitz matrices is optimal, and it would be interesting to understand this better. A particularly interesting question is if the $O(\sqrt{k})$ entry sample complexity bound of \cite{qiao2017gridless} can be obtained for general low-rank $T$, without the additional  assumptions made in that work.

\smallskip\noindent\textbf{Other settings / measures of error.} It would be interesting to consider continuous analogs of this problem, for instance, on continuous or infinite length stationary signals.
There has been substantial recent work on this subject \cite{CohenTsiperEldar:2018,ArianandaLeus:2012,LexaDaviesThompson:2011}, and we believe that similar insights can be brought to bear here.
There are also other interesting measures of error, for instance, relative Loewner ordering, or total variation distance the sampled and estimated distributions, which corresponds to recovering $T$ in the so-called Mahalanobis distance.
Many of our results can be directly applied to these settings when $T$ is well-conditioned. However, it would be interesting to understand the sample complexities of these problems in more detail.

\section{Notation and Preliminaries}
\label{sec:notation}
\smallskip\noindent\textbf{Indexing.} For a postive integer $z$ let $[z] \eqdef \{1,\ldots , z\}$. For a set $S$, let $S^z$ denote the set of subsets of $S$ with $z$ elements. For a vector $v \in \C^d$, let $v_S \in \R^{|S|}$ denote the restriction of $v$ to the coordinates in $S$. For $T \in \R^{d \times d}$ and set $R \subseteq [d]$ let $T_R \in \R^{|R| \times |R|}$ denote the principle submatrix of $T$ with rows and columns indexed by $R$. Note that even when $T$ is Toeplitz, $T_R$ may not be.

\smallskip\noindent\textbf{Toeplitz matrix operations.} For any vector $a = [{a}_0, \ldots, {a}_{d - 1}] \in \R^d$, let $\toep(a) \in \R^{d \times d}$ denote the symmetric Toeplitz matrix whose entries are $\toep(a)_{j,k} = a_{|j - k|}$. 
For any $M \in \R^{d \times d}$, let $\diag(M)$ be the matrix given by setting all of $M$'s off diagonal entries to $0$. Let $\avg(M)$ be the symmetric Toeplitz matrix given by averaging the diagonals of $M$. Namely, $\avg(M) = \toep(a)$ where for $s \in \{0,\ldots d-1\}$, $a_s = \frac{1}{\left |j,k \in [d]: |j-k| = s\right |} \sum_{j,k \in [d]: |j -k| = s} M_{j,k}$.

\smallskip\noindent\textbf{Linear algebra notation.} For a matrix $A$, let $A^T$ and $A^*$ denote the transpose and Hermitian transpose, respectively. When all of $A$'s entries are real, $A^* = A^T$. For a vector $y\in \C^d$ let $\|y\|_2 = \sqrt{y^*y}$ denote the $\ell_2$ norm. For a matrix $A$ with $d$ columns, let $\|A\|_2 = \sup_{x\in C^d} \|Ax\|_2/\|x\|_2$ denote spectral norm. Let $\|A\|_F$ denote the Frobenius norm.

A Hermitian matrix $A \in \C^{d \times d}$ is positive semidefinite (PSD) if for all $x \in \C^{d}$, $x^* A x \ge 0$. 
Let $\preceq$ denote the Loewner ordering: $A \preceq B$ indicates that $B - A$ is PSD.
Let $A^+$ denote the Moore-Penrose pseudoinverse of $A$. $A^+ = V \Sigma^{-1} U^*$ where $U\Sigma V^* = A$ is the compact singular value decomposition of $A$. When $A$ is PSD, let $A^{1/2} = U \Sigma^{1/2}$ where $\Sigma^{1/2}$ is obtained by taking the entrywise square root of $\Sigma.$ Let $A_k$ denote the projection of $A$ onto its top $k$-singular vectors: $A_k = U_k \Sigma_k V_k^*$ where $U_k, V_k \in \C^{d \times k}$ contain the first $k$ columns of $U,V$ respectively and $\Sigma_k \in \R^{k \times k}$ is diagonal matrix containing $A$'s largest $k$ singular values. Note that $A_k$ is an optimal low-rank approximation in that: $A_k= \argmin_{\rank\, k\ M} \norm{A-M}_F = \argmin_{\rank\, k\ M} \norm{A-M}_2$.

\smallskip\noindent\textbf{Gaussians, TV distance, KL divergence.}
We let $x\sim \mathcal{N}(m,T)$ indicates that $x$ is drawn from a Gaussian distribution with mean $m$ and covariance $T$. 
When $B \in \C^{d \times d}$, it is an elementary fact that $x \sim \normal(0, BB^*)$, can be written as $x = B U g$ where $g \sim \normal (0, I)$ and $U \in \C^{d \times d}$ is a fixed unitary matrix. When $B$ is real, $U$ can be any orthogonal matrix in $\R^{d \times d}$.
For any two distributions $F, G$, let $\dtv(F,G) = \frac{1}{2} \int |dF - dG|$ denote their total variation distance, and let $\dkl (F, G) = \int \log \frac{dF}{dG} dF$ denote the KL divergence~\cite{cover2012elements}.

\smallskip\noindent\textbf{Probability notation.}
We let $\Pr[g]$ denote the probability that event $g$ occurs. Sometimes we add an additional subscript to remind the reader of any random events that effect $g$, and thus its probability. For example, we might write $\Pr_x[g]$ if $x$ is a random variable and $g$ depends on $x$. We use $:$ to denote ``such that''. For example, $\Pr[\exists x \in \mathcal{X}: x > 10]$ denotes the ``probability that there exists some element $x$ in the set $\mathcal{X}$ such that $x$ is greater than $10$.''

\smallskip\noindent\textbf{Miscellaneous.}
 Let $|z| = \sqrt{z^*z}$ denote the magnitude of a complex number $z$. Let $i$ denote $\sqrt{-1}$. 
 When stating asymptotic bounds we use $\tilde{O}(f)$ as shorthand for $O(f \log^c f)$ for some constant $c$, and similarly we let $\Omegatilde(f) = \Omega (f / \log^c f)$ for some fixed $c$.

\subsection{Fourier analytic tools}
We repeatedly make use of the Fourier structure of Toeplitz matrices. To that end, we often work with asymmetric Fourier transform matrices whose rows correspond to ``on-grid'' integer points (e.g., samples in time domain) and whose columns correspond to ``off-grid'' frequencies. 

\begin{definition}[Fourier matrix]\label{def:fourier_matrix}
	For a set $S = \{f_1,\ldots f_s\} \subset  \C$, we define the Fourier matrix $F_S \in \C^{d\times s}$ as the matrix with $j^{\text{th}}$ column given by $F(f_j)$, where: 
	\begin{align*}
	F(f_j)^\top = \left [1, e^{-2 \pi i f_j}, e^{-2 \pi i 2 f_j}, \ldots, e^{-2 \pi i (d-1) f_j} \right ].
	\end{align*}
\end{definition}

While Toeplitz matrices cannot in general be diagonalized by the discrete Fourier transform like circulant matrices, the Carath\'{e}odory-Fej\'{e}r-Pisarenko decomposition (also known as the Vandermonde decomposition) implies that they \emph{can} be diagonalized by a Fourier matrix with off-grid frequencies \cite{caratheodory1911zusammenhang,pisarenko1973retrieval}.
A general version of the classical result holding for all ranks $r \leq d$ is:
\begin{fact}[Vandermonde Decomposition,  Cor.  1 of \cite{cybenko1982moment}]\label{lem:fourier}
	Any PSD Toeplitz matrix $T \in \R^{d \times d}$ with rank $r \le d$ can be written as $T = F_S D F_S^*$ where $D \in \R^{r\times r}$ is positive and diagonal and $F_S \in \C^{d \times r}$ is a Fourier matrix (as in Definition \ref{def:fourier_matrix}) with frequencies $S = \{f_1,\ldots f_r\} \subset [0,1]$. Furthermore, $f_1,\ldots f_r$ come in conjugate pairs with equal weights: i.e., for every $j \in [r]$, there is a $j'$ (possibly  with $j' = j$) with $e^{-2\pi i f_j} = e^{2\pi i f_{j'}}$ and $D_{j,j} = D_{j',j'}$.
\end{fact}
One useful consequence of the Vandermonde decomposition is that, as $d \rightarrow \infty$, $T$ is diagonalized by the \emph{discrete time Fourier transform}. We can take advantage of this property by noting that any Toeplitz matrix $\toep(a)$ can be extended to be arbitrarily large by padding $a$ with $0$'s. The spectral norm of the extension upper bounds that of $\toep(a)$. Accordingly, we can bound $\|\toep(a)\|_2$ by the supremum of $a$'s DTFT. We refer the reader to \cite{Meckes:2007} for a formal argument, ultimately giving:
\begin{fact}[See e.g., \cite{Meckes:2007}]
\label{lem:meckes}
For any $a = [a_0, \ldots, a_{d-1}] \in \R^d$,
\begin{align*}
\| \toep({a}) \|_2 &\leq \sup_{x \in [0,1]} L_{a}(x) &  &\text{where} & L_{{a}} (x) &\eqdef { a}_0 + 2\sum_{s = 1}^{d - 1}{ a}_s \cos (2 \pi s x).
\end{align*}
\end{fact}

\subsection{Linear algebra tools}
Our methods with sample complexity logarithmic in $d$ for low-rank (or close to low-rank) $T$ do not select entries from each $x \sim \mathcal{N}(0,T)$ according to a fixed pattern like a sparse ruler. Instead, they use a randomly chosen pattern. We view this approach as randomly selecting a subset of rows from the $d\times n$ matrix of samples $[x^{(1)}, \ldots, x^{(n)}]$. To find a subset that preserves properties of this matrix, we turn to tools from randomized numerical linear algebra and matrix sketching \cite{Mahoney:2011,drineas2016randnla}. Specifically, we adapt ideas from random selection algorithms based on \emph{importance sampling} -- i.e., when rows are sampled independently at random with non-uniform probabilities.

In particular, work on matrix sketching and graph sparsification has repeatedly used importance sampling probabilities proportional to the \emph{statistical leverage scores} \cite{SpielmanSrivastava:2011,DrineasMahoneyMuthukrishnan:2006}, which allow for random matrix compressions that preserve significant information \cite{cohen2015dimensionality,MuscoMusco:2017}. Formally, the leverage score is a measure defined for any row in a matrix $A \in \C^{d\times s}$:
\begin{definition}\label{def:lev_score}
	The leverage score, $\tau_j(A)$, of the $j^\text{th}$ row $a_j\in \C^{d\times 1}$ in $A \in \C^{d\times s}$ is defined as:
	\begin{align*}
	\tau_j(A) \eqdef  a_j\left(A^*A\right)^+ a_j^*.
	\end{align*} 
	Note that that $\sum_{j=1}^d \tau_j(A) = \tr\left(A\left(A^*A\right)^+A^*\right) = \rank(A) \leq \min(d,s)$. 
\end{definition}
The leverage score can be seen as measuring the ``uniqueness'' of any row. This intuition is captured by the following lemmas, which are used extensively in prior work on leverage scores:
\begin{fact}[Minimization Characterization]\label{fact:min_char} 
	\begin{align*}
		\tau_j(A) = \min_{y\in \C^d \text{ s.t. } y^TA = a_j} \|y\|_2^2.
	\end{align*}
\end{fact}
\begin{fact}[Maximization Characterization]\label{fact:max_char} 
	Let $\left(Ay\right)_j$ denote the $j^\text{th}$ entry of the vector $Ay$.
	\begin{align*}
	\tau_j(A)= \max_{y \in \C^s} \frac{\left|\left(Ay\right)_j\right|^2}{\norm{Ay}_2^2}.
	\end{align*}
\end{fact}
One consequence of Fact \ref{fact:min_char} is that $\tau_j(A) \leq 1$ for all $j$ since we can choose $y$ to be the $j^\text{th}$ standard basis vector. We defer further discussion and proofs of these bounds to Appendix \ref{app:leverage_scores}, where we also state a standard matrix concentration bound for sampling by leverage scores. 

In our applications, we cannot compute leverage scores explicitly, since we use them to sample rows from $[x^{(1)}, \ldots, x^{(n)}]$ without ever looking at this whole matrix. To cope with this challenge we take advantage of the Fourier structure of Toeplitz matrices (Claim \ref{lem:fourier}), which ensures that each $x^{(\ell)}$ can be written as a Fourier sparse or nearly Fourier  sparse function when $T$ is low-rank or close to low-rank. This allows us to extend recent work on nearly tight \emph{a priori} leverage scores estimates for Fourier sparse  matrices \cite{ChenKanePrice:2016,ChenPrice:2018,ChenPrice:2019, AvronKapralovMusco:2019} to our setting. See Appendix \ref{app:leverage_scores}.

\subsection{Distances between Gaussians}
In our lower bound proofs,
 we repeatedly use the following bounds on the TV distance and KL divergence between Gaussians.
\begin{fact}
Let $M_1, M_2 \in \R^{d \times d}$ be PSD.
Then, 
\[
\dkl(\normal (0, M_1), \normal (0, M_2)) = \Theta \Paren{ \Norm{I - M_2^{-1/2} M_1 M_2^{-1/2}}_F^2} \; .
\]
By Pinsker's inequality~\cite{cover2012elements}, this implies that
\[
\dtv(\normal (0, M_1), \normal (0, M_2)) = O \Paren{ \Norm{I - M_1^{-1/2} M_2 M_1^{-1/2}}_F} \; .
\]
\end{fact}

\subsection{Concentration bounds}

We make use of the following well known concentration inequality for quadratic forms:

\begin{claim}[Hanson-Wright Inequality -- see e.g., \cite{rudelson2013hanson}]
\label{thm:hw}
Let $x^{(1)}, \ldots, x^{(n)} \sim \normal (0, I)$ be independent.
There exists a universal constant $c > 0$ so that for all $A \in \R^{d \times d}$ and all $t > 0$,
\[
\Pr \Brac{\left| \frac{1}{n} \sum_{\ell = 1}^n {x^{(\ell)}}^T A x^{(\ell)} - \tr (A) \right| > t} \leq 2 \exp \Paren{- c\cdot n \min \Paren{ \frac{t^2}{\| A \|_F^2}, \frac{t}{\| A \|_2} }} \; .
\]
\end{claim}

\section{Covariance Estimation by Ruler}\label{sec:ruler}

We begin by establishing Theorems \ref{thm:linear} and \ref{thm:32}, which can be analyzed in a unified way.   In particular, both bounds are achieved by Algorithm \ref{alg:sparse} (see below) with the ruler $R$ chosen to either equal $[d]$ for Theorem \ref{thm:linear}, or to be a suitable sparse ruler for Theorem \ref{thm:32} (see Definition \ref{def:ruler}).
Algorithm \ref{alg:sparse} approximates $T = \toep(a)$ by individually approximating  each entry ${a}_s$ in $a$ by an average over all entries in each $x^{(j)}$ with distance $s$ from each other. Letting ${\tilde a}$ denote this approximation, the algorithm returns the Toeplitz matrix $\tilde{T} = \toep({\tilde a})$. 

 \begin{algorithm}[H]
\caption{\algoname{Toeplitz covariance estimation by ruler}}
{\bf input}: Independent samples $x^{(1)},\ldots,x^{(n)} \sim \mathcal{N}(0,T)$. \\
{\bf parameters}: Ruler $R \subseteq [d]$.\\
{\bf output}: $\tilde T \in \R^{d \times d}$ approximating $T$.
\begin{algorithmic}[1]
\For{$s=0,\ldots, d-1$}
\State{${\tilde a}_s := \frac{1}{n |R_s|}\sum_{\ell=1}^n \sum_{(j,k) \in R_s} x^{(\ell)}_j x^{(\ell)}_k$,   \hspace{2em}(where $R_s$ is defined in Definition \ref{def:ruler})}
\EndFor
\\\Return{$\tilde{T} \eqdef \toep({\tilde a})$.}
\end{algorithmic}
\label{alg:sparse}
\end{algorithm}

When $R = [d]$, Algorithm \ref{alg:sparse} reduces to approximating $T$ by the empirical covariance with its diagonals averaged to make it symmetric and Toeplitz. In the other extreme, we can approximate $T$ using a sparse ruler with $|R| = O(\sqrt{d})$ (see Claim \ref{clm:sparse}).
It is easy to see the following:
\begin{claim}\label{clm:basic}
Algorithm \ref{alg:sparse} has vector sample complexity $n$, entry sample complexity $|R|$, and total sample complexity $n|R|$. It runs in time $O\left(n|R|^2\right)$ and outputs a symmetric Toeplitz matrix $\tilde{T}$.
\end{claim}

We prove a generic bound on the accuracy of Algorithm \ref{alg:sparse} for any $n$ and ruler $R$ (the accuracy will depend on properties of $R$). We then instantiate this bound for specific choices of these parameters to give Theorem \ref{thm:linear}, Theorem \ref{thm:32}, and also Theorem \ref{thm:ruler} when $T$ is approximately low-rank.
We first define a quantity characterizing how extensively a ruler covers distances $0,\ldots,d-1$:
\begin{definition}[Coverage coefficient, $\Delta(R)$]\label{def:coverage}
For any ruler $R \subseteq [d]$ let
$
\Delta(R) \eqdef \sum_{s=0}^{d-1} \frac{1}{|R_s|}.
$
\end{definition}
The coverage coefficient of a ruler is smaller when distances are represented more frequently by the ruler, so intuitively, the sample complexity will scale with it. It is not hard to check that when $R = [d]$, $\Delta(R) = \Theta(\log d)$ and for any $R$ with $|R| = \Theta(\sqrt{d})$, $\Delta(R) = \Theta(d)$.
It is also possible to construct rulers which interpolate between $\Theta(\sqrt{d})$ and $d$ sparsity:
\begin{definition}
	\label{def:general_ruler}
For $\alpha \in [1/2, 1]$, let $R_\alpha$ be the ruler given by $R_\alpha = R^{(1)}_\alpha \cup R^{(2)}_\alpha$, 
where\footnote{We assume for simplicity of exposition that $d^{\alpha}$ and $d^{1-\alpha}$ are integers.}
\begin{equation}
R^{(1)}_\alpha = \{1, 2, \ldots, d^\alpha \} \; , ~\mbox{and} ~ \; R^{(2)}_\alpha = \{d, d - d^{1-\alpha}, d- 2 d^{1 - \alpha} , \ldots, d - (d^{\alpha}-1) d^{1 - \alpha} \} \; .
\end{equation}
\end{definition}
\noindent
For such rulers we prove the following general bound in Appendix \ref{app:ruler}:
\begin{lemma}
\label{lem:general-ruler}
For any $\alpha \in [1/2, 1]$, $R_\alpha$ has size $|R_\alpha| \leq 2 d^\alpha $, and moreover 
\[
\Delta (R_\alpha) \le 2d^{2 - 2\alpha} + d^{1-\alpha} ( 1 + \log(\lceil d^{2\alpha -1} \rceil) ) = d^{2 - 2 \alpha} + O(d^{1 - \alpha} \cdot \log d) \; .
\]
\end{lemma}
Note that for $\alpha = 1$ we recover a bound of $O(\log d)$ and for  $\alpha = 1/2$ a bound of $O(d)$.
With the coverage coefficient defined, we prove our general bound for any ruler $R$:
\begin{theorem}[Accuracy of Algorithm \ref{alg:sparse}]\label{thm:algSparse}
Let $\tilde T$ be the output of Algorithm  \ref{alg:sparse} run with ruler $R$. For any $\eps \in (0,1]$, let $\kappa \eqdef \min \left (1,\frac{\eps^2 \cdot \norm{T}_2^2}{\Delta(R) \cdot \norm{T_R}_2^2 }\right )$.
There exist universal constants $C, c > 0$ such that
\begin{align*}
\Pr \left[ \left\| T - \tilde T \right\|_2 > \eps \norm{T}_2 \right] \leq \frac{C d^2}{\eps} \exp \left( - cn \kappa \right) \; . 
\end{align*}
In particular, if $n = {\Theta} \left (\frac{\log \left ({d/\eps \delta} \right )}{\kappa} \right )$, then $\left\| T - \tilde T \right\|_2 \leq \eps \norm{T}_2 $ with probability at least $1 - \delta$.
\end{theorem}
\begin{proof}
As noted in Claim \ref{clm:basic}, $\tilde T$ is a symmetric Toeplitz matrix and so $T -  \tilde T$ is a symmetric Toeplitz matrix.
Let $e \in \R^{d}$ be its associated vector, so that $e$ is a random variable. 
By Fact~\ref{lem:meckes}, to prove the theorem it suffices to show that there exist $C,c > 0$ with:
\[
\Pr_{e} \Brac{ \exists x \in [0, 1]: | L_{ e} (x) | \geq \eps \norm{T}_2  } \leq \frac{Cd^2}{\eps} \exp \left( - c n \kappa \right) \; .
\]
Let $N = \{0, \frac{\eps}{Q d^2}, \frac{2 \eps}{Q d^2}, \ldots, 1 \}$, for some universal constant $Q$ sufficiently large ($Q = 80 \pi$ suffices).
We will prove that there exists a universal constant $c > 0$ so that the following two bounds hold:
\begin{align}
\Pr_{{e}} \Brac{\exists s \in \{0,\ldots, d-1\}: |{e}_s| > 10 \| T \|_2 } &\leq  \exp (-c n) \;, \; \mbox{ and } \label{eq:a0-upper-bound} \\
\Pr_{e} \Brac{ \exists x \in N: |L_{e} (x)| > \frac{\eps}{2} \| T \|_2 } &\leq \frac{3Q d^2}{\eps} \exp \Paren{- cn \kappa } \; . \label{eq:a0-conc-net}
\end{align}
We first show how~\eqref{eq:a0-upper-bound} and~\eqref{eq:a0-conc-net} together imply the desired bound.
Condition on the event that $|{e}_s| \leq 10 \| T \|_2$ for all $s$ and the event that $|L_{e} (x)| \leq \frac{\eps}{2} \| T \|_2$ for all $x \in N$.
By~\eqref{eq:a0-upper-bound} and~\eqref{eq:a0-conc-net} and a union bound, we know that this happens with probability at least 
\[
1 - \frac{4Q d^2}{\eps^2} \exp \Paren{- cn \kappa } \; .
\]
Observe that for all $x \in [0, 1]$, we have
\begin{align*}
|L_{{e}}'(x)| &= \left|  4 \pi \cdot \sum_{s = 1}^{d-1}   s\cdot{e}_s \sin(2 \pi s x)  \right| \\
&\leq 4 \pi d^2 \| {e} \|_\infty \leq 40 \pi d^2 \| T \|_2,
\end{align*}
where in the last line we use that we are conditioning on $|{e}_s| \le 10 \norm{T}_2$ for all $s$. 
Now, for any $x \in [0, 1]$, let $x' \in N$ be so that $|x - x'| \leq \eps / (Q d^2)$.
Then, as long as $Q \geq 80 \pi$, we have
\begin{align*}
|L(x) - L(x')| &\leq L(x') + \sup_{y \in [x, x']} |L'(y)| \cdot |x - x'|  \\
&\leq \frac{\eps}{2} \| T \|_2 + 40 \pi d^2 \| T \|_2 \cdot |x - x'| \leq \eps \| T \|_2 ,
\end{align*}
which proves the theorem, letting and $C = 4Q$.

It thus suffices to prove~\eqref{eq:a0-upper-bound} and~\eqref{eq:a0-conc-net}.
We first prove~\eqref{eq:a0-upper-bound}.
We can write for any $s \in 0, \ldots, d-1$:
$${e}_s = {a}_s - {\tilde a}_s = \frac{1}{n \cdot |R_s|}\sum_{\ell=1}^n \sum_{(j,k) \in R_s} \left [T_{j,k}  - x^{(\ell)}_j \cdot x^{(\ell)}_k \right ].$$
Observe that for each $\ell = 1, \ldots, n$, we have that $x^{(\ell)}_j \cdot x^{(\ell)}_k$ is a subexponential random variable with mean $T_{j,k}$ and second moment which can be written using Isserlis's theorem as
\[
\E \Brac{(x^{(\ell)}_j)^2 \cdot (x^{(\ell)}_k)^2} = T_{j,j} T_{k,k} + 2 T_{j,k}^2 \leq 3 \| T \|_2^2.
\]
The desired bound then follows from standard bounds on subexponential random variables \cite{Wainwright:2019}.

We now turn our attention to~\eqref{eq:a0-conc-net}.
Fix $x \in [0,1]$, and associate to it the Toeplitz matrix $M \in \R^{d \times d}$, where for all $s \in 0,\ldots, d-1$ and $(j,k) \in [d] \times [d]$ with $|j-k| = s$, we let 
\[
M_{j,k} = \frac{\cos (2 \pi s x)}{|R_s|}.
\]
Letting $\bar T$ be the empirical covariance matrix $\bar T \eqdef \frac{1}{n}\sum_{\ell=1}^n x^{(\ell)} {x^{(\ell)} }^T$,
we can see that 
\begin{align}\label{eq:asTrace}
L_{e} (x) = { e}_0 + 2\sum_{s = 1}^{d - 1}{ e}_s \cos (2 \pi s x) = \tr(T_R-\tilde T_R, M_R) = \tr(T_R - \bar T_R, M_R).
\end{align}
The last equality follows from the fact that $M$ is symmetric Toeplitz and $\tilde T_R$ is just obtained by averaging the diagonal entries of $\bar T_R$. That is, for any $s \in 0,\ldots d-1$, 
\begin{align*}
\sum_{(j,k) \in R \times R: |j-k| = s} \tilde T_{j,k} = \sum_{(j,k) \in R \times R: |j-k| = s} \bar T_{j,k},
\end{align*}
and thus
\begin{align*}
\tr(\tilde T_R, M_R) &= \sum_{(j,k) \in R \times R} \tilde T_{j,k} \cdot M_{j,k}\\
&= \sum_{s=0}^{d-1} \left [\sum_{(j,k) \in R \times R: |j-k| = s} \tilde T_{j,k} \cdot M_{j,k} \right]\\
&= \sum_{s=0}^{d-1} \left [\sum_{(j,k) \in R \times R: |j-k| = s} \bar T_{j,k} \cdot M_{j,k} \right ] = \tr(\bar T_R, M).
\end{align*}
We will show that for any $\eps \in [0, 1]$, there exists $c > 0$ such that:
\begin{align}
\Pr \Brac{| \tr(T - \bar T_R, M)| > \frac{\eps}{2} \| T \|_2 } &\leq 2 \exp \Paren{- cn \kappa} \; . \label{eq:quad-form-conc}
\end{align}

Notice that if we let $x^{(\ell)}_R \in \R^{|R|}$ denote the $\ell^{th}$ sample restricted to the indices in $R$ we have  $x^{(\ell)}_R \sim \normal (0, T_R)$.  Therefore, if we let $y^{(\ell)} = T_R^{-1/2} x^{(\ell)}_R$, the $y^{(\ell)}$ are i.i.d., with $y^{(\ell)} \sim \normal (0, I_{R \times R})$. Moreover, we can rewrite the quantity in \eqref{eq:quad-form-conc} as:
\begin{align*}
\tr(T_R - \bar T_R,M) &= \tr(T,M)  - \frac{1}{n} \sum_{i = 1}^n {x^{(\ell)}_R}^T M x^{(\ell)}_R\\
&=  \tr(M') - \frac{1}{n} \sum_{\ell = 1}^n {y^{(\ell)}}^T M' y^{(\ell)} \;,
\end{align*}
where $M' =T_R^{1/2} M T_R^{1/2}$. Observe that $\| M \|_F^2 = \sum_{s = 0}^{d-1} |R_s| \cdot  \frac{\cos (2 \pi s x)^2}{|R_s|^2} \le \Delta(R)$. Additionally, $\| M' \|_F \leq \| T_R \|_2 \| M \|_F$, and $\| M' \|_2 \leq \| T_R \|_2 \| M \|_2$.
Thus, applying the Hanson-Wright inequality (Claim~\ref{thm:hw}):
\begin{align*}
\Pr \Brac{| \tr(T_R - \bar T_R, M)| > \frac{\eps}{2} \| T \|_2 } &= \Pr \Brac{ \left| \tr(M') - \frac{1}{n} \sum_{\ell = 1}^n {y^{(\ell)}}^T M' y^{(\ell)} \right| > \frac{\eps}{2} \| T \|_2 } \\
&\leq 2 \exp \Paren{ -c n \min\Paren{\frac{\eps^2 \| T \|_2^2}{4 \| M' \|_F^2}, \frac{\eps \|T \|_2}{2 \| M' \|_2}} } \\
&\leq 2 \exp \Paren{ -c n \min\Paren{\frac{\norm{T}_2^2}{\norm{T_R}_2^2} \cdot \frac{\eps^2}{4 \| M \|_F^2}, \frac{\norm{T}_2}{\norm{T_R}_2} \cdot \frac{\eps}{2 \| M \|_2}} } \\
&\leq 2 \exp \Paren{- c n \kappa} \; ,
\end{align*}
recalling that  $\kappa \eqdef \min \left (1,\frac{\eps^2 \cdot \norm{T}_2^2}{\Delta(R) \cdot \norm{T_R}_2^2 }\right )$.
This gives \eqref{eq:quad-form-conc}. Combining with \eqref{eq:asTrace}, for any fixed $x \in [0, 1]$:
\begin{equation}
\Pr \Brac{ |L_{e} (x)| > \eps \| T \|_2 } \leq 2  \exp \Paren{- cn \kappa } \; ,
\end{equation}
and so~\eqref{eq:a0-conc-net} follows by a union bound over all $x \in N$. This completes the proof of Theorem \ref{thm:algSparse}.
\end{proof}

\subsection{Applications of the general bound, full rank matrices}
We describe a few specific instantiations of Theorem~\ref{thm:algSparse}. When $R = [d]$ is the full ruler, $\Delta(R) = O(\log d)$ by Lemma \ref{lem:general-ruler} and $\|T_R\|_2 = \|T\|_2$.
This immediately yields:
\begin{reptheorem}{thm:linear}[Near linear sample complexity -- full theorem]
	Let $\tilde{T}$ be the output of Algorithm \ref{alg:sparse} run with the full ruler $R = [d]$. There exist universal constants $C, c > 0$ such that, for any $\eps \in (0,1]$:
	\begin{align*}
		\Pr \left[ \left\| T - \tilde T \right\|_2 > \eps \norm{T}_2 \right] \leq \frac{C d^2}{\eps} \exp \left( - \frac{cn \eps^2}{\log d} \right) \; . 
	\end{align*}
In particular, if $n = {\Omega} \left (\frac{\log \left ({d}/{\eps \delta} \right ) \log d}{\eps^2} \right )$, then $\left\| T - \tilde T \right\|_2 \leq \eps \norm{T}_2 $ with probability at least $1 - \delta$.
\end{reptheorem}

We can also apply Theorem \ref{thm:algSparse} to any ruler with $|R| = \Theta(\sqrt{d})$ (e.g., the ruler in Claim \ref{clm:sparse}, or more optimal constructions). In this case, $\Delta(R) = O(d)$ and can bound $\|T_R\|_2 \leq \|T\|_2$ to give:
\begin{reptheorem}{thm:32}[Sparse ruler sample complexity -- full theorem] Let $\tilde{T}$  be the output of Algorithm \ref{alg:sparse} run with any $\Theta(\sqrt{d})$-sparse ruler $R$. There exist universal constants $C, c > 0$ such that, for any $\eps \in (0,1]$:
\begin{align*}
\Pr \left[ \left\| T - \tilde T \right\|_2 > \eps \norm{T}_2 \right] \leq \frac{C d^2}{\eps} \exp \left( - \frac{cn \eps^2}{d} \right) \; . 
\end{align*}
In particular, for $n = {O} \left (\frac{d \log \left ({d/\eps \delta} \right )}{\eps^2} \right )$, $\left\| T - \tilde T \right\|_2 \leq \eps \norm{T}_2 $ with probability at least $1 - \delta$. 
\end{reptheorem}

Finally, we record a bound for rulers $R_\alpha$ with any $\alpha \in [1/2,1]$. This result smoothly interpolates between Theorems \ref{thm:linear} and \ref{thm:32}, giving a natural way of trading between entry sample complexity and vector sample complexity, which might be valuable in practice. For example, by setting $\alpha = 3/4$, we obtain an algorithm with ESC $O(d^{3/4})$ and VSC $\tilde{O}(d^{1/2})$.
\begin{theorem}
	\label{thm:general-ruler-bound}
	Let $\alpha \in [1/2, 1]$ be fixed.
	Let $\tilde T$ the output of Algorithm \ref{alg:sparse} run with the ruler $R_\alpha$ of Def. \ref{def:general_ruler}. There exist universal constants $C, c > 0$ such that, for any $\eps \in (0,1]$:
	\begin{align*}
	\Pr \left[ \left\| T - \tilde T \right\|_2 > \eps \norm{T}_2 \right] \leq \frac{C d^2}{\eps} \exp \left( - \frac{cn \eps^2}{\max(d^{2 - 2\alpha}, d^{1 - \alpha} \log d)} \right) \; . 
	\end{align*}
	In particular, for any $\delta > 0$, if 
	\[
	n = {\Omega} \left (\frac{\max(d^{2 - 2\alpha}, d^{1 - \alpha} \log d) \cdot \log \left ({d}/{\eps \delta} \right )}{\eps^2} \right ) \; ,
	\] then $\left\| T - \tilde T \right\|_2 \leq \eps \norm{T}_2 $ with probability at least $1 - \delta$. 
\end{theorem}

As discussed in Section \ref{sec:our_results}, while obtaining optimal entry sample complexity for any method that reads a fixed subset of entries of each sample, Theorem \ref{thm:32} gives total sample complexity that is worse than Theorem \ref{thm:linear} by roughly a $\sqrt{d}$ factor.
In Appendix~\ref{sec:ruler-lb} we prove a lower bound demonstrating that  this is tradeoff is inherent:
\begin{theorem}
\label{thm:ruler-lower-bound}
Let $R \subseteq [d]$ be a sparse ruler with $|R| \leq d^{\alpha}$, for some $\alpha \in [1/2, 1]$, and let $\eps > 0$ be sufficiently small.
Let $\mathscr{A}$ be a (possibly randomized) algorithm that takes $x^{(1)}, \ldots, x^{(n)} \sim \normal (0, T)$ for some unknown, positive semidefinite and Toeplitz $T$, and for all $\ell$, looks only at the coordinates of $x^{(\ell)}$ in $R$ before outputting $\tilde T$. If, for any $T$, $\| \tilde T - T\|_2 \leq \eps \norm{T}_2$ with probability $\geq 1 / 10$, then we must have
$n = \Omega (d^{3 - 4\alpha} / \eps^2)$. The algorithm thus requires total sample complexity $\Omega (d^{3 - 3\alpha} / \eps^2)$.
\end{theorem}
Note that when $\alpha = 1/2$, Theorem \ref{thm:ruler-lower-bound} implies that we require $\Omega (d / \eps^2)$ vector samples, which nearly matches the upper bound of Theorem \ref{thm:32} and confirms the simulation results in Figure \ref{fig:sparse_vs_full}.
Interestingly in Theorem~\ref{thm:ruler-lower-bound} we get some tradeoff for all $\alpha \in [1/2, 1]$, but it does not match the upper bound of Theorem \ref{thm:general-ruler-bound}. We leave it as an open question to resolve this gap.


\subsection{Applications of the general bound, low-rank matrices}
As we saw in Figure~\ref{fig:sparse_vs_full}, sparse ruler estimation can actually outperform full ESC methods in terms of total sample complexity when $T$ is low-rank. We support this observation theoretically, giving a much tighter bound than Theorem \ref{thm:32} when $T$ is close to low-rank and $R$ is the sparse ruler of Claim \ref{clm:sparse} or one of the family of rulers of Definition \ref{def:general_ruler}.

To do so, we critically use that a low-rank matrix cannot concentrate significant mass on more than a few small principal submatrices. 
 A version of this fact, showing that a low-rank matrix cannot concentrate on many diagonal entries is shown in \cite{lrpd}. We use as similar argument. 
\begin{lemma}\label{lem:offDiagWeight}
Consider any partition $R_1 \cup \ldots \cup R_t = [d]$. For any $T \in \R^{d \times d}$, let $T_{R_j}$ be the principal sub-matrix corresponding to $R_j$. If $T$  is rank-$k$, then for some $S \subseteq [t]$ with $|S| \le \frac{k}{\eps}$, for all $\ell \in [t]\setminus S$:
\begin{align*}
\norm{T_{R_\ell}}_F^2 \le \eps \cdot \norm{T_{(R_\ell,[d])}}_F^2,
\end{align*}
where $T_{(R_\ell,[d])}$ is the submatrix of $T$ given by selecting the rows in $R_\ell$ and all columns.
\end{lemma}
\begin{proof}
Letting $\tau_j(T)$ be $T$'s $j^{\text{th}}$ leverage score (Definition \ref{def:lev_score}) 
we have $\sum_{j=1}^d \tau_j(T) = \rank(T) = k$. Thus, excluding a set $S \subseteq [t]$ with $|S| \le \frac{k}{\eps}$ indices, for all $\ell \in [t] \setminus S$, we have $\sum_{j \in R_\ell} \tau_j(T) \le \eps$. Applying the maximization characterization of Fact \ref{fact:max_char}, for any $k \in R_\ell$, letting $e_k$ be the $k^{th}$ standard basis vector and $T_{(k,[d])}$ denote the $k^{\text{th}}$ column of $T$:
\begin{align*}
\frac{1}{\norm{T_{(k,[d])}}_2^2} \sum_{j \in R_\ell} T_{j,k}^2 = \sum_{j \in R_\ell} \frac{(T{e}_k)_j^2}{\norm{T {e}_k}_2^2} \le \sum_{j \in R_\ell} \tau_j(T) \le \eps.
\end{align*}
Applying this bound to all $k \in R_\ell$ we have $\frac{\norm{T_{R_\ell}}_F^2}{\norm{T_{(R_\ell,[d])}}_F^2} \le \eps$, which completes the lemma.
\end{proof}

In Appendix \ref{app:ruler} we use Lemma \ref{lem:offDiagWeight} to show:
\begin{lemma}\label{lem:stableSubsample}
Let $\alpha \in [1/2, 1]$.
For any $k \le d$, any PSD Toeplitz matrix $T \in \R^{d \times d}$, and the sparse ruler 
$R_\alpha$ of Def. \ref{def:general_ruler},\footnote{For $\alpha=1/2$ the bound also applies to the construction of Claim \ref{clm:sparse}, which is essentially identical to that of Def. \ref{def:general_ruler}, but applies when $\sqrt{d}$ is not an integer.}
\begin{align*}
\norm{T_{R_\alpha}}_2^2 \le \frac{32 k^2}{d^{2 - 2 \alpha}} \cdot \norm{T}_2^2 + 8 \cdot \min \left( \norm{T - T_k}_2^2, \frac{2}{d^{1 - \alpha}} \cdot \norm{T-T_k}_F^2 \right),
\end{align*}
\end{lemma}
\noindent
where $\displaystyle T_k = \argmin_{\rank-k\ M} \norm{T-M}_F = \argmin_{\rank-k\ M} \norm{T-M}_2$. If $T$ is rank-$k$, $\norm{T-T_k}_F^2 = \norm{T - T_k}_2^2 = 0$.

\noindent Plugging Lemma \ref{lem:stableSubsample} in Theorem \ref{thm:algSparse} we obtain:
\begin{reptheorem}{thm:ruler}[Sublinear sparse ruler sample complexity -- full theorem]
	Let $\alpha \in [1/2, 1]$.
	Let $\tilde{T}$  be the output of Algorithm \ref{alg:sparse} run with the sparse ruler $R_\alpha$ of Def. \ref{def:general_ruler}. There exist universal constants $C, c > 0$ such that, for any $k \le d$ and $\eps \in (0,1]$:
{\small
\begin{align*}
\Pr \left[ \left\| T - \tilde T \right\|_2 > \eps \norm{T}_2  \right] &\leq \frac{C d^2}{\eps} \exp \left( - cn \eps^2 \cdot \min \left (\frac{1}{k^2\left (1+\frac{\log d}{d^{1-\alpha}} \right )}, \frac{\norm{T}_2^2}{(d^{2 - 2 \alpha} + d^{1 - \alpha} \log d) \cdot \norm{T - T_k}_2^2} \right ) \right) \; ,\\~\mbox{and} \\
\Pr \left[ \left\| T - \tilde T \right\|_2 > \eps \norm{T}_2  \right] &\leq \frac{C d^2}{\eps} \exp \left( - cn \eps^2 \cdot \min \left (\frac{1}{k^2\left (1+\frac{\log d}{d^{1-\alpha}} \right )}, \frac{\norm{T}_2^2}{(d^{1 - \alpha} + \log d) \cdot \norm{T-T_k}_F^2} \right ) \right).
\normalsize
\end{align*}
}
\end{reptheorem}
\noindent
This bound implies, for example, that if if we take $\alpha = 1/2$, and $\frac{\sqrt{d}}{k^2} \cdot \norm{T-T_k}_F^2 \le \norm{T}_2^2$ or $\frac{d}{k^2} \cdot \norm{T - T_k}_2^2 \leq \norm{T}_2^2$, and $n = {O} \left (\frac{k^2 \cdot \log \left ({d}/{\eps \delta} \right )}{\eps^2} \right )$, then $\left\| T - \tilde T \right\|_2 \leq \eps \norm{T}_2 $ with probability at least $1 - \delta$. 
Note that, unlike Theorem \ref{thm:32}, which applies to all sparse rulers, we only prove Theorem \ref{thm:32} for the specific family of sparse rulers $R_\alpha$ (it also applies to the essentially identical ruler of Claim \ref{clm:sparse} for $\alpha = 1/2$). 
We conjecture that it should hold more generally. 


\begin{proof}
By Lemma \ref{lem:stableSubsample} we have 
\[
\norm{T_{R_\alpha}}_2^2 \le \frac{32 k^2}{d^{2 - 2\alpha}} \cdot \norm{T}_2^2 + 8 \cdot \min\left( \norm{T - T_k}_2^2, \frac{2}{d^{1 - \alpha}} \norm{T-T_k}_F^2 \right) \; .
\]
By Lemma~\ref{lem:general-ruler}, we also have
$\Delta(R_\alpha) \le 2d^{2 - 2 \alpha} + O(d^{1-\alpha} \cdot \log d)$.
Applying Theorem \ref{thm:algSparse} with 
$$\kappa = \min \left (1, \frac{\eps^2 \cdot \norm{T}_2^2}{\Delta(R_\alpha) \cdot \norm{T_R}_2^2}\right )$$
and plugging in our bound on $\Delta (R_\alpha)$ gives the theorem.
\end{proof}


\section{Covariance Estimation by Sparse Fourier Transform}\label{sec:fourier}

In this section we describe an estimation algorithm that further improves on sparse ruler methods for low-rank or nearly low-rank $T$. When $k$ is fixed, we almost eliminate dependence on $d$ in our total sample complexity, reducing the $\sqrt{d}$ in Theorem \ref{thm:ruler} to a logarithmic dependence.

Our approach heavily uses the Fourier structure of Toeplitz matrices, in particular that a rank-$k$ Toeplitz matrix can be written in the span of $k$ off-grid frequencies. Using sparse Fourier transform techniques, we can recover (approximations to) these frequencies using very few samples.

\subsection{Recovering exactly low-rank matrices via Prony's method}\label{sec:exactLow}

We first consider the case when $T$ is exactly rank-$k$ for some $k \le d$. In this case, we give a short argument that Problem \ref{prob:main} can be solved with $O(k \log k / \eps^2)$ total sample complexity, which is completely independent of the ambient dimension $d$. Our main contribution will be to extend this claim to the practical setting where $T$ is not precisely low-rank. However, we discuss the exactly low-rank case first to exposit the main idea.

Consider $x^{(1)},\ldots,x^{(n)}$ drawn independently from $\mathcal{N}(0,T)$ and let
$X \in \R^{d \times n}$ be the matrix with the samples as its columns. 
Via the Vandermonde decomposition of Lemma \ref{lem:fourier} we can write $T$ as $T = F_S D F_S^*$ where $D \in \R^{k \times k}$ is positive and diagonal and $F_S \in \C^{d \times k}$ is a Fourier matrix with columns corresponding to some set of $k$ frequencies $S = \{f_1,\ldots f_k\}$.  
Note that $X$ is distributed as $T^{1/2} G$ where each entry of $G$ is drawn i.i.d. from $\mathcal{N}(0,1)$. 
Equivalently, there exists some unitary matrix $U \in \C^{d \times d}$ so that $X$ is distributed as $F_S D^{1/2} U G = F_S Z$ where we let $Z \in \C^{k \times n}$ denote $D^{1/2} U G$. Note that each column of $F_S Z$ is Fourier $k$-sparse. This structure lets us recover any column of $Z$ from $2k$ entrywise measurements of the corresponding column in $X$, using Prony's method \cite{de1795essai}.
Moreover, the method lets us simultaneously recover the frequencies $S$, giving us a representation of $X = F_SZ$.

For any vector $x \in \R^d$, let $P_k (x) \in \R^{k \times k}$ denote the matrix whose entries are given by $(P_k (x))_{i,j} = x_{i + j - 1}$, and let $b_k (x) \in \R^k$ be given by $(b_k(x))_i = x_{k + i}$.
Observe that $P_k(x)$ and $b_k (x)$ can be simultaneously formed from $2k$ entrywise observations from $x$.
We require the following classical lemma, underlying Prony's method, which states that a signal with a $k$-sparse Fourier transform can be recovered from $2k$ measurements.
\begin{lemma}[\cite{de1795essai}]
\label{lem:prony}
Let $S$ be a set of frequencies with $|S| = k$, and let $x = F_S y$ for some unknown $y \in \R^k$.
Let $c \in \R^k$ be the solution to the linear equation $P_k (x) c = -b_k (x)$, and define the polynomial $p (t) = \sum_{s = 1}^k c_s t^s$.
Let $R = R(x) = \{r_1, \ldots, r_{k'}\}$ be the roots of $p$.
Then (1) we have $R \subseteq S$, and (2) we have that $x$ is in the column space of $F_R$.
\end{lemma}
\noindent
For simplicity, in this section we show how Lemma \ref{lem:prony}, combined with an exact polynomial root finding oracle, yields an algorithm that recovers $T$ using only $O(k \log k /\eps^2)$ samples and time. In reality, such an oracle does not exist, so in Appendix \ref{sec:badRoots} we show how to use approximate root finding algorithms to actually instantiate this algorithm.
We give two algorithms: the first pays logarithmically in $d$ in terms of sample complexity, and polynomially in $d$ in terms of runtime. The second recovers the $O(k \log k/\eps^2)$ sample complexity and pays just a $\log \log d$ factor and a $\log\log \kappa$ factor in runtime, where $\kappa$ is the condition number of $T$ (the ratio between its largest and smallest non-zero eigenvalues). It is stronger except in the extreme case when $\kappa = \Omega(2^{2^d})$.

\begin{reptheorem}{thm:prony-exact-roots}[Dimension independent sample complexity -- full theorem]
Let $\That$ be the output of Algorithm~\ref{alg:exactSFT-1} run on samples $x^{(1)}, \ldots, x^{(n)} \sim \normal (0, T)$ for a rank-$k$ Toeplitz covariance $T = F_S D F_S^*$.
Suppose there is an oracle $\mathscr{O}$ which takes a degree $k$ polynomial $p$ with roots $S = \{z_1, \ldots, z_k\}$ in the unit complex disc and exactly returns the set $S$.
Then there is a universal constant $C > 0$ so that for any $\eps \in (0,1)$:
\[
\Pr \Brac{\Norm{T - \That}_2 > \eps \Norm{T}_2} \leq 2k \exp \Paren{-C n \eps^2} \; .
\]
Moreover, the algorithm requires $O(k n)$ total samples, $O(n)$ queries to $\mathscr{O}$, and $\poly (k, n)$ additional runtime.
In particular, for any $\delta > 0$, if we let $n = \Theta (\log (k / \delta)/ \eps^2)$, then the algorithm outputs $\hat{T}$ so that $\Pr \Brac{\norm{T - \That}_2 > \eps \norm{T}_2} < \delta$, using $\Theta (k \log (k / \delta) / \eps^2)$ entrywise samples and $\poly (k, 1/\eps, \log 1 / \delta)$ runtime.
\end{reptheorem}
\noindent
\begin{proof}
The sample and runtime bounds are immediate by inspection of Algorithm~\ref{alg:exactSFT-1}.
Note that the algorithm does not explicitly output $\That$, as that would require $d^2$ time.
Instead, it outputs a set of $k$ frequencies $R$ and a $k \times k$ diagonal matrix $\hat D$.
As in the pseudocode, we let $\That = F_R \hat{D} F_R^*$.
Thus the remainder of this section is dedicated to a proof of correctness.

As discussed, for all $j \in [n]$, we can write $x^{(j)} = F_S D^{1/2} U g^{(j)}$, where $g^{(j)} \sim \normal(0, I)$ and $U$ is unitary.
Letting $y^{(j)} = D^{1/2} U g^{(j)}$, we see that with probability $1$, $y^{(j)}$ has full support. Since $F_S$ has full column rank (or $T$ would have rank $<k$), we thus see that with probability $y^{(j)}$ will not fall in the span of $F_R$ for any $R \subset S$.
Therefore, by Lemma~\ref{lem:prony}, the oracle $\mathscr{O}$ outputs the full set $\{e^{2 \pi i f_1}, \ldots, e^{2 \pi i f_k}\}$, where $S = \{f_1, \ldots, f_k\}$.
In particular $R_1 = R_j = S$ for all $j \in [n]$ with probability $1$. 
As a result, with probability $1$, we have that $\hat{y}^{(j)} = y^{(j)}$ for all $j = 1, \ldots, n$.
 The random variable $Z_{j \ell} \eqdef |y^{(j)}_{\ell}|^2 = D_{\ell \ell} \left |\Paren{ Ug^{(j)}}_{\ell}\right |^2$ is a sub-exponential random variable with expecation $D_{\ell \ell}$.
Hence, if we let $Z_\ell = \frac{1}{n} \sum_{j \in [n]} Z_{j \ell}$, we have
\[
\Pr \Brac{|Z_\ell - D_{\ell \ell}| > \eps |D_{\ell \ell}|} \leq 2 \exp \Paren{-C n \eps^2} \; ,\]
for some universal constant $C$.
In Line 4 of Algorithm \ref{alg:exactSFT-1} we set $\hat{D} = \mathrm{diag} (Z_1, \ldots, Z_k)$.
By a union bound we thus have
\begin{align}\label{eq:dspec}
\Pr \Brac{-\eps D\preceq\hat{D} - D\preceq \eps D} \leq 2k \exp \Paren{-C n \eps^2} \; .
\end{align}
In other words, if the event of \eqref{eq:dspec} holds, for any vector $w$, $|w^T (\hat D - D) w| \le \eps \cdot w^T D w.$ This gives that for any $w$, $|w^T F_S (\hat D - D)F_S^* w| \le \eps \cdot w^T F_S DF_S^* w$ and thus $\Norm{F_S \hat{D} F_S^* - T}_2 \le \eps \Norm{T}_2$. Thus by \eqref{eq:dspec} we have as claimed:
\[
\Pr \Brac{\Norm{F_S \hat{D} F_S^* - T}_2 > \eps \Norm{T}_2} \leq 2k \exp \Paren{-C n \eps^2} \; .\]
\end{proof}

  \begin{algorithm}[h]
\caption{\algoname{toeplitz covariance estimation by Prony's method (exact root finding)}}
{\bf input}:  $X \in \R^{d \times n}$ with columns $x^{(1)},...,x^{(n)}$ drawn independently from $\mathcal{N}(0,T)$, and an oracle $\mathscr{O}$ for exact root finding \\
{\bf parameters}: rank $k$.\\
{\bf output}: A set of $\le k$ frequencies $R$ and $\hat D \in \R^{k \times k}$ so that $\bar T \eqdef F_{R} \hat D F_{R}^* \in \R^{d \times d}$ approximates $T$.
\begin{algorithmic}[1]
\State{Let $X_{2k} \in \R^{n \times 2k}$ have $i^{th}$ column equal to the first $2k$ entries of $x^{(i)}$.}
\State{Let $[\hat{Y}, R] \eqdef \texttt{Prony}(X_{2k})$, and let $\hat{y}^{(j)}$ be the $j$th column of $\hat Y$.}
\State{Let $D_\ell = \frac{1}{n} \sum_{j \in [n]} |\hat{y}^{(j)}_{\ell}|^2$, for $\ell \in [k]$.}
\State{Let $\hat D = \mathrm{diag} (D_1, \ldots, D_k)$.}\\
\Return{$R$ and $\hat D$.}

\end{algorithmic}
{\bf subroutine}:  $\texttt{Prony}(W)$, for input $W \in \R^{2k \times n}$.
\begin{algorithmic}[1]
\For{$j=1,\ldots n$}
\State{Let $w^{(j)} \in \R^{2k}$ denote the $j^{th}$ column of $W$.}
\State{Solve $P_k (w^{(j)}) c = b_k(w^{(j)})$.}
\State{Let $p(t) = \sum_{s=1}^k c_s t^s$ and find the roots $R_j = \{r_1,\ldots, r_k \}$ of $p(t)$ via $\mathscr{O}$.}
\State{Let $F_{R_j} \in \C^{k \times k}$ be the Fourier matrix (Def \ref{def:fourier_matrix}) with $k$ rows and  frequencies $R_j$.}
\State{Solve $F_{R_j} \hat{y}^{(j)} = (w^{(j)})_{[k]}$.} \\
\Return{$R_1$ and $[\hat{y}^{(1)}, \ldots, \hat{y}^{(n)}]$}.
\EndFor
\end{algorithmic}
\label{alg:exactSFT-1}
\end{algorithm}

\subsection{Approximately low-rank matrices}
\label{sec:approx-low-rank}

Unlike in the exact low-rank case, when $T$ is just close to low-rank, the samples $x^{(1)},\ldots, x^{(n)}$ do not exactly have a $k$-sparse Fourier representation, and we cannot recover them exactly using Prony's method. However, we can prove that these samples are still \emph{approximately} Fourier $k$-sparse. Thus, we will be able to recover an approximation to the samples and the empirical covariance matrix. We start with a short lemma, which guarantees that a good approximation to the empirical covariance matrix $XX^T$ suffices to obtain a good approximation to $T$.
\begin{lemma}\label{lem:frobReduction}
Consider PSD Toeplitz $T \in \R^{d \times d}$ and let $X \in \R^{d \times n}$ have columns drawn i.i.d. from $\frac{1}{\sqrt{n}} \cdot \mathcal{N}(0,T)$. If $n = {\Omega} \left (\frac{\log \left ({d}/{\eps \delta} \right ) \log d}{\eps^2} \right )$ then with probability at least $1-\delta$, for any  $B \in \R^{d \times d}$:
\begin{align*}
\norm{T - \avg(B)}_2 \le \eps \norm{T}_2 + \norm{B - XX^T}_F,
\end{align*}
where $\avg(B)$ is the symmetric Toeplitz matrix given by averaging the diagonals of $B$ (see Sec. \ref{sec:notation}).
\end{lemma}
\begin{proof}
We have by triangle inequality:
\begin{align*}
\norm{T - \avg(B)}_2 &\le \norm{T - \avg(XX^T)}_2 + \norm{\avg(B) - \avg(XX^T)}_2\nonumber\\
&\le \norm{T - \avg(XX^T)}_2 + \norm{\avg(B - XX^T)}_F\nonumber\\
&\le \norm{T - \avg(XX^T)}_2+ \norm{B - XX^T}_F,
\end{align*}
where the last inequality  follows from the fact that $\avg(\cdot)$ can only decrease Frobenius norm. The lemma follows by  applying Theorem \ref{thm:linear}, which shows that, for $n = {\Omega} \left (\frac{\log \left ({d}/{\eps \delta} \right ) \log d}{\eps^2} \right )$, with probability  at least $1-\delta$, $\norm{T - \avg(XX^T)}_2 \le \eps \norm{T}_2$.
\end{proof}

\subsubsection{Existence of frequency based low-rank approximation}

With Lemma \ref{lem:frobReduction} in place, our goal is to show how to obtain $B$ with small $\norm{B - XX^T}_F$ using a sublinear number of reads from each sample of $\mathcal{N}(0,T)$ (i.e., by loooking at a sublinear number of rows of $X$). For general $X$, this would be impossible. However, since the columns of $X$ are distributed as $\mathcal{N}(0,T)$, if $T$ is close to low-rank we can argue that there is a $B$ that well approximates $XX^T$ and further \emph{is spanned by a small number of frequencies}. Using sparse Fourier transform techniques similar at a high level to the MUSIC algorithm \cite{Schmidt:1981}, we  can sample-efficiently recover such a $B$, without reading all of $X$.

We again start with the Vandermonde decomposition of Lemma \ref{lem:fourier}, writing $T$ as $T = F_S D F_S^*$ where $D \in \R^{d \times d}$ is nonnegative and diagonal and $F_S$ is a Fourier matrix with columns corresponding to some set of frequencies $S = \{f_1,\ldots f_d\}$.
If $T$ were circulant, $f_1,\ldots, f_d$ would be `on-grid' orthogonal frequencies and the columns of $F_S$ would be eigenvectors of $T$. Thus, the columns corresponding to the largest $k$ entries of $D$ (the top $k$ eigenvalues) would span an optimal rank-$k$ approximation of $T$. While this is not the case for general Toeplitz matrices, we can still apply the well known technique of \emph{column subset selection} to prove  the existence of a set of $O(k)$ frequencies spanning a near optimal low-rank approximation to $T$. Specifically in Appendix \ref{app:additional} we show:
\begin{lemma}[Frequency-based low-rank approximation]\label{lem:fss2}
For any PSD Toeplitz matrix $T \in \R^{d\times d}$, rank $k$, and $m \ge ck$ for some fixed constant $c$, there exists $M = \{f_1, \ldots, f_m\} \subset [0,1]$ such that, letting $F_M \in \C^{d \times m}$ be the Fourier matrix  with frequencies $M$ (Def. \ref{def:fourier_matrix}) and $Z = F_M^+ T^{1/2}$, we have 1): $\norm{F_M^+ }_2^2 \le \frac{2}{\beta}$ and 2):
\begin{align}
\norm{F_M Z - T^{1/2}}_F^2  &\le 3 \norm{T^{1/2} - T^{1/2}_k}_F^2 + 6\beta\norm{T}_2 \; ,~\mbox{and} \label{fss:frob1}\\
\norm{F_M Z - T^{1/2}}_2^2  &\le 3\norm{T^{1/2} - T^{1/2}_k}_2^2 + \frac{3}{k} \norm{T^{1/2} - T^{1/2}_k}_F^2 + 6 \beta \norm{T}_2.\label{fss:spectral1}
\end{align}
\end{lemma}
Note that a bound on $\norm{F_M^+ }_2$ will eventually be important in arguing that we can discretize $[0,1]$ to search for $M$ with frequencies on some finite grid, without incurring too much error. If we did not require this bound, we could obtain \eqref{fss:frob1} and \eqref{fss:spectral1} without the additive error depending on $\beta$. 

Lemma \ref{lem:fss2} shows the existence of $m = O(k)$ frequencies that span an $m$-rank approximation to $T^{1/2}$ that is nearly as good as the best rank-$k$ approximation. However,
to apply Lemma \ref{lem:frobReduction} we must identify a subset of frequencies $M$ such that $F_M Z$ approximates not $T^{1/2}$ itself but $X$ with columns drawn from $\frac{1}{\sqrt{n}} \cdot\mathcal{N}(0,T)$.
$X$ is distributed as $T^{1/2} G$ where each entry of $G$ is distributed as $\mathcal{N}(0,1/\sqrt{n})$. Thus $X$ can be viewed as a random `sketch' of $T^{1/2}$, approximating many of its properties \cite{drineas2016randnla,Woodruff:2014}. In particular, we employ a  projection-cost-preserving sketch property, following from Theorems 12 and 27 of \cite{cohen2015dimensionality}\footnote{\cite{cohen2015dimensionality} considers only  real valued matrices, however their proofs can be extended to complex matrices. Alternatively, the bound can be shown by rotating $F_S$ to a real matrix. See Appendix \ref{app:additional}.}, which in particular implies that if $T^{1/2}$ can be well approximated by projecting onto a set of $m$ columns of $F$ (i.e., $m$ frequencies), then as long as $n \approx O(m)$, so can $X$:
\begin{lemma}[Projection-cost-preserving sketch]\label{lem:pcp}
Consider PSD $T \in \R^{d\times d}$ and $X \in \R^{d \times n}$ with columns drawn i.i.d. from $\frac{1}{\sqrt{n}} \cdot \mathcal{N}(0,T)$. For any rank $m$ and $\gamma,\delta \in (0,1]$, if $n \ge \frac{c (m + \log(1/\delta))}{\gamma^2}$ for sufficiently large $c$, then with probability $\ge 1- \delta$, for all $F_M \in \C^{d \times m}$:
\begin{align*}
 \norm{F_M (F_M^+X) - X}_F^2  &\in (1 \pm \gamma) \norm{F_M (F_M^+T^{1/2}) - T^{1/2}}_F^2
\text{\hspace{1em} and,}\\
 \norm{F_M (F_M^+X) - X}_2^2  &\in (1 \pm \gamma) \norm{F_M (F_M^+T^{1/2})   - T^{1/2}}_2^2 \pm \frac{\gamma}{k} \norm{F_M (F_M^+T^{1/2})   - T^{1/2}}_F^2.
\end{align*}
\end{lemma}

Applying Lemma \ref{lem:pcp} along with the frequency subset selection result of Lemma \ref{lem:fss2} we conclude (see Appendix \ref{app:additional} for full proof):
\begin{lemma}\label{lem:existenceT}
Consider PSD Toeplitz $T \in \R^{d\times d}$ and $X \in \R^{d \times n}$ with columns drawn i.i.d. from $\frac{1}{\sqrt{n}} \cdot \mathcal{N}(0,T)$. For any rank $k$, $\eps,\delta \in (0,1]$, $m \ge c_1 k$, and $n \ge c_2 \left (m + \log(1/\delta) \right )$ for sufficiently large $c_1,c_2$, with probability $\ge 1- \delta$,  there exists $M = \{f_1,\ldots,f_m\} \subset [0,1]$ such that, letting $Z = F_M^+ X$, $\norm{Z}_2^2 \le \frac{c_3 \cdot d^2 \norm{T}_2}{\eps^2}$ for some fixed $c_3$ and:
\begin{align*}
 \norm{F_M Z Z^* F_M^* -XX^*}_F &\le 10 \sqrt{\norm{T-T_k}_2 \cdot \tr(T) + \frac{\tr(T-T_k) \cdot \tr(T)}{k}} + \frac{\eps}{2} \norm{T}_2.
\end{align*}
\end{lemma}


\subsubsection{Computing a frequency based low-rank approximation}

With Lemma \ref{lem:existenceT} in place, it remains to show that, with few samples from $X$, we can actually find an $M \subset [0,1]$ and $W = ZZ^*$ (approximately) satisfying the bound of the lemma. To do this, we will consider a large but finite number of possible sets obtained by discretizing $[0,1]$ and  will brute force search over these sets. The critical step required to perform this search is to compute a nearly optimal $W$ for each $M$ while only looking at a subset of rows in $X$. We can argue that it is possible to do so if we could sample $X$ using the row leverage scores of $F_M$ (Def. \ref{def:lev_score}), which are well known to be useful in approximately solving regression problems, such as the one required to compute $W$ \cite{DrineasMahoneyMuthukrishnan:2006,Woodruff:2014}. 
 Since leverage score sampling succeeds with high probability,  by a union bound this would allow us to find a near optimal $W$ for every subset $M$ in our brute-force search, and thus identify a near optimal subset. 
 
 Unfortunately, the leverage scores of $F_M$ depend on $M$ and will be different  for different frequency sets. Naively, we will have to sample from a different distribution for each $M$. Over the course of the search, this could require reading all $d$ rows of the $X$ (i.e., all entries in our samples). To deal with this challenge, we extend recent work on \emph{a priori leverage score bounds} for Fourier  sparse functions. As we prove in Appendix~\ref{app:leverage_scores}, there is a fixed function that well approximates the leverage scores of \emph{any} $F_M$ with $|M| = m$. See Figure \ref{fig:lev} for an illustration of the sampling distribution corresponding to this function. This fact lets us sample a \emph{single set} of $\tilde O(m)$ rows using this distribution and, by a union bound, find a near optimal $W$ for all $M$ in our search.

We start by claiming that discretizing frequencies to a net does not introduce too much additional error. Formally, we prove the following in Appendix \ref{app:additional}:

\begin{lemma}\label{lem:mnetT}
Consider PSD Toeplitz $T \in \R^{d\times d}$ and $X \in \R^{d \times n}$ with columns drawn i.i.d. from $\frac{1}{\sqrt{n}} \cdot \mathcal{N}(0,T)$.
For any rank $k$ and $\delta,\eps \in (0,1]$, consider $m \ge c_1 k$, $n \ge c_2 \left (m + \log(1/\delta) \right )$,
and $N = \{0,\alpha,2\alpha,\ldots 1\}$ for $\alpha = \frac{\eps^2}{c_3 d^{3.5}}$
 for sufficiently large constants $c_1,c_2,c_3$.  With probability $\ge 1- \delta$, there exists $M = \{f_1,\ldots,f_m\} \subset N$ with:
\begin{align*}
\min_{W \in \C^{m \times m}} \norm{F_M W F_M^* - XX^T}_F \le 10 \sqrt{\norm{T-T_k}_2 \cdot \tr(T) + \frac{\tr(T-T_k) \cdot \tr(T)}{k}} + \eps \norm{T}_2.
\end{align*}
\end{lemma}

We next show how to sample efficiently find $M \subset N$ and $W \in \C^{d \times d}$ that nearly  minimize $\norm{F_M W F_M^* - XX^T}_F$ and thus (approximately) satisfy the guarantee of Lemma \ref{lem:mnetT}.
\begin{lemma}\label{lem:msample}
Consider any $N \subset [0,1]$, rank $m$, and $B \in \R^{d \times d}$. Let $S_1\in \R^{s_1 \times d}$ and  $S_2\in \R^{s_2 \times d}$  be independent sampling matrices sampled using the scheme of Claim \ref{claim:lev_sampling} with parameters $\eps = \frac{1}{c_1}$ and $\delta = \frac{1}{c_2N^m}$ and the leverage score distribution of Cor. \ref{cor:sum_bound} with frequency set size $2m$ for sufficiently large constants $c_1,c_2$. With probability at least $97/100$, $\max(s_1,s_2) \le  cm^2 \log m \log(|N|)$ for sufficiently  large $c$ and
letting $\displaystyle \tilde M,  \tilde W = \argmin_{W \in \C^{m \times m}, M \in N^m} \norm{S_1F_M W F_M^*S_2^T - S_1B S_2^T}_F$,
\begin{align*}
\norm{F_{\tilde M} \tilde W F_{\tilde M}^* - B}_F \le 216 \min_{W \in \C^{m \times m}, M \in N^m}\norm{F_M W F_M^* - B}_F.
\end{align*}
\end{lemma}
\begin{proof}
Let $\displaystyle \hat M, \hat W =\argmin_{W \in \C^{m \times m}, M \in N^m} \norm{F_M \tilde W F_M^* - B}_F.$
We will show that with probability $\ge 97/100$, \emph{for all} $M \in N^m$ and $W \in \C^{m \times m}$,
\begin{align}\label{eq:constEmb}
 \norm{S_1F_M W F_M^*S_2^T - S_1BS_2^T}_F \in \left (1 \pm \frac{1}{8} \right ) \norm{F_M W F_M^* - B}_F \pm 100 \norm{F_{\hat M} \hat W F_{\hat  M}^* - B}_F.
\end{align}
Conditioning on \eqref{eq:constEmb} holding, the lemma follows  by  applying it twice to bound:
\begin{align*}
\norm{F_{\tilde M} \tilde W F_{\tilde M}^* - B}_F &\le \frac{8}{7}  \norm{S_1F_{\tilde M} \tilde W F_{\tilde M}^*S_2^T - S_1BS_2^T}_F + \frac{8}{7}\cdot 100 \norm{F_{\hat M} \hat W F_{\hat M}^* - B}_F\\
&\le \frac{8}{7}   \norm{S_1F_{\hat M} \hat W F_{\hat M}^*S_2^T - S_1BS_2^T}_F + \frac{800}{7}  \norm{F_{\hat M} \hat W F_{\hat M}^* - B}_F\\
&\le \frac{9}{8} \cdot \frac{8}{7}  \norm{F_{\hat M} \hat W F_{\hat M}^* - B}_F + \left (\frac{800}{7} + 100 \right ) \norm{F_{\hat M} \hat W F_{\hat M}^* - B}_F\\
&\le 216 \cdot \norm{F_{\hat M} \hat W F_{\hat M}^* - B}_F.
\end{align*}
We now prove that \eqref{eq:constEmb} holds with probability  $\ge 97/100$.
By triangle inequality, for any $M,W$:
\begin{align}\label{eq:break}
\large \|S_1F_M W F_M^*S_2^T - S_1BS_2^T \large \|_F \in \large ||S_1F_M W F_M^*S_2^T &- S_1 F_{\hat M} \hat W F_{\hat M}^* S_2^T \large \|_F\nonumber\\
&\pm \norm{S_1F_{\hat M} \hat W F_{\hat M}^*S_2^T - S_1 B S_2^T}_F.
\end{align}
Note that by design of the sampling scheme in Claim \ref{claim:lev_sampling}, both $S_1$ and $S_2$ preserve the squared Frobenius norm in expectation. I.e., for any matrix $C \in \C^{d \times p}$, $\E [\norm{S_1 C}_F^2] = \E [\norm{S_2 C}_F^2]  = \norm{C}_F^2$. Thus, by Markov's inequality, with probability $\ge 99/100$, 
$$\norm{S_1F_{\hat M} \hat W F_{\hat M}^* - S_1 B}_F^2 \le 100  \E \left [\norm{S_1F_{\hat M} \hat W F_{\hat M}^* - S_1 B}_F^2\right] = 100 \norm{S_1F_{\hat M} \hat W F_{\hat M}^* - S_1 B}_F^2.$$
 In turn, fixing $S_1$, with probability $99/100$, 
 $$\norm{S_1F_{\hat M} \hat W F_{\hat M}^*S_2^T - S_1 BS_2^T}_F^2 \le 100  \E \left [\norm{S_1F_{\hat M} \hat W F_{\hat M}^*S_2^T - S_1 BS_2^T}_F^2\right] \le 100 \norm{S_1F_{\hat M} \hat W F_{\hat M}^* - S_1 B}_F^2.$$
Overall, with probability $\ge 98/100$,
\begin{align}\label{eq:addBound}
\norm{S_1F_{\hat M} \hat W F_{\hat M}^*S_2^T - S_1 B S_2^T}_F^2 \le 100^2 \norm{F_{\hat M} \hat W F_{\hat M}^* - B}_F^2.
\end{align}
Additionally, by Corollary \ref{cor:sum_bound}, \emph{for any} $M \in N^m$, $S_1$ is sampled using leverage scores upper bounds of $[F_M, F_{\hat M}]$ which sum to $O(m \log m)$.
Thus if we apply Claim \ref{claim:lev_sampling} with error parameter $\eps = \frac{1}{32}$ and some failure probability $\delta$, with probability $\ge 1-\delta$, $s_1 \le cm \log m \cdot \log \left (\frac{m}{\delta}\right)$ for sufficiently large $c$, and for all $y \in \C^{2m}$, $\norm{S_1 [F_M, F_{\hat M}] y}_2^2 \in (1 \pm 1/32) \norm{[F_M, F_{\hat M}] y}_2^2$. In particular this yields that for all
$W \in \C^{m \times m}$:
\begin{align*}
\norm{S_1F_M W F_M^* - S_1F_{\hat M} \hat W F_{\hat M}^*}_F \in \left (1 \pm \frac{1}{32} \right ) \norm{F_M W F_M^* - F_{\hat M} \hat W F_{\hat M}^*}_F.
\end{align*}
Similarly, fixing $S_1$, with probability $\ge 1-\delta$ over the random choice of $S_2$:
\begin{align*}
\norm{S_1F_M W F_M^*S_2 - S_1F_{\hat M}\hat  W F_{\hat M}^*S_2}_F \in \left (1 \pm \frac{1}{32} \right ) \norm{S_1 F_M W F_M^* - S_1 F_{\hat M} \hat W F_{\hat M}^*}_F.
\end{align*}
Combining these two bounds, for any $M$, with probability $\ge 1-2\delta$, for all $W\in \C^{m \times m}$:
\begin{align}\label{eq:multBound}
\norm{S_1F_M W F_M^*S_2 - S_1F_{\hat M} \hat W F_{\hat M}^*S_2}_F \in \left (1 \pm \frac{1}{8} \right ) \norm{F_M W F_M^* - F_{\hat M} \hat W F_{\hat M}^*}_F.
\end{align}
Setting $\delta = \frac{1}{200\cdot |N|^m}$ gives that, by a union bound, with probability $\ge 99/100$: \eqref{eq:multBound} holds for all $M \in N^m$ and $\max(s_1,s_2) \le c m \log m \log \left (m \cdot |N|^m \right ) = 2c m^2 \log m \log(|N|)$ for sufficiently large $c$.
Plugging \eqref{eq:multBound} and \eqref{eq:addBound} back into \eqref{eq:break} and applying a union bound gives that with probability $\ge 97/100$, \eqref{eq:constEmb} holds  for all $M \in N^m$ and $W \in \C^{m \times m}$ simultaneously, completing  the lemma.
\end{proof}

With Lemmas \ref{lem:mnetT}, and \ref{lem:msample} in place, we are ready to give our method, detailed in Algorithm \ref{alg:sft}, and its analysis. 

 \begin{algorithm}[H]
\caption{\algoname{Toeplitz Covariance Estimation by Sparse Fourier Transform}}
{\bf input}: $X \in \R^{d \times n}$ with columns $x^{(1)},...,x^{(n)}$ drawn independently from $\mathcal{N}(0,T)$.\\
{\bf parameters}: rank $m$, net discretization level $\alpha$, constants $c_1,c_2$. \\
{\bf output}: $\bar T \in \R^{d \times d}$ approximating $T$.
\begin{algorithmic}[1]
\State{Sample $S_1 \in \R^{d \times s_1}$ and $S_2 \in \R^{d \times s_2}$ using the sampling scheme of Clm. \ref{claim:lev_sampling} with parameters $\eps = \frac{1}{c_1}$, $\delta = \frac{\alpha^m}{c_2}$, and the leverage score distribution of Cor. \ref{cor:sum_bound} with frequency set size $2m$.}
\State{$c_{best} \eqdef \infty$, $M_{best} = \emptyset$, $W_{best} = 0^{m \times m}$.}
\For{$M \in \{0,\alpha,2\alpha,\ldots 1\}^m$}
\State{$W \eqdef (S_1^T F_M)^+ ( S_1 XX^T S_2^T) (F_M^* S_2)^+.$}
\If{$\norm{S_1^T F_M W F_M^*S_2 - S_1^T XX^T S_2^T}_F \le c_{best}$}
\State{$M_{best} \eqdef M$, $W_{best} \eqdef W$, and $c_{best} \eqdef \norm{S_1^T F_M W F_M^*S_2 - S_1^T XX^T S_2^T}_F$.}
\EndIf
\EndFor
\\\Return{$\bar T \eqdef \avg \left (F_{M_{best}} W_{best} F_{M_{best}}^* \right )$.}
\end{algorithmic}
\label{alg:sft}
\end{algorithm}

\begin{reptheorem}{thm:sftT}[Sparse Fourier transform sample complexity -- full theorem]
Consider PSD Toeplitz $T \in \R^{d \times d}$. For any rank $k$, Algorithm \ref{alg:sft} run on inputs drawn from $\mathcal{N}(0,T)$ with $c_1,c_2$ sufficiently large, rank $m \ge c_3 k$, $n \ge c_4 \max \left ( \frac{\log \left (\frac{d}{\eps}\right  ) \log d}{\eps^2}, k \right )$, and discretization $\alpha = \frac{\eps^2}{c_5 d^{3.5}}$ for sufficiently large constants $c_3,c_4,c_5$ with probability $\ge 19/20$, (1): returns $\bar T$ that satisfies
$$\left\| T - \bar T \right\|_2 = O \left ( \sqrt{\norm{T-T_k}_2 \cdot \tr(T) + \frac{\tr(T-T_k) \cdot \tr(T)}{k}} + \eps \norm{T}_2 \right )$$
and (2): has total sample complexity $O \left (k^3 \log k\cdot \log \left (\frac{d}{\eps}\right ) + \frac{k^2 \log k \cdot \log d\cdot \log^2 \left (\frac{d}{\eps} \right )}{\eps^2} \right )$.
\end{reptheorem}
\noindent
While this above error guarantee appears somewhat unusual, its can provide fairly strong bounds. See Theorem \ref{thm:sft} for an example giving error $\epsilon \norm{T}_2$ when $T$ has low stable rank $\tr(T) / \Norm{T}_2$.
\begin{proof}
We first note that in Line 4, we have (see, e.g. \cite{friedland2007generalized}):
\begin{align*}
W \eqdef (S_1^T F_M)^+ ( S_1 XX^T S_2^T) (F_M^* S_2)^+ = \argmin_{W \in \C^{m \times m}} \norm{S_1 F_M W F_M^* S_2^T - S_1 XX^T S_2^T}_F.
\end{align*}
If $c_1,c_2$ are sufficiently large,
for our setting of $\eps,\delta$, by Lemma \ref{lem:msample} we thus have with probability $97/100$:  $\max(s_1,s_2) \le c m^2 \log m \log(1/\alpha)$ for some constant $c$ and: 
\begin{align*}
\norm{F_{M_{best}} W_{best} F_{M_{best}}^* - XX^T }_F \le 216 \min_{W \in \C^{m \times m}, M \in \{0,\alpha,\ldots,1\}^m}\norm{F_M W F_M^* - XX^T}_F.
\end{align*}
Further, by Lemma \ref{lem:mnetT} for our setting of $m$, $n$, and $\alpha$, with probability $\ge 99/100$ we have:
\begin{align*}
\min_{W \in \C^{m \times m}, M \in \{0,\alpha,2\alpha,\ldots,1\}^m}\norm{F_M W F_M^* - XX^T}_F \le 10 \sqrt{\norm{T-T_k}_2 \cdot \tr(T) + \frac{\tr(T-T_k) \cdot \tr(T)}{k}} + \eps \norm{T}_2.
\end{align*}
Thus by  a union bound, with probability  $\ge 96/100$,
\begin{align*}
\norm{F_{M_{best}} W_{best} F_{M_{best}}^* - XX^T }_F \le 2160 \sqrt{\norm{T-T_k}_2 \cdot \tr(T) + \frac{\tr(T-T_k) \cdot \tr(T)}{k}} + 216\eps \norm{T}_2.
\end{align*}
Finally, for our setting of $n \ge \frac{c_4\log \left (\frac{d}{\eps}\right  ) \log d}{\eps^2}$ by Lemma \ref{lem:frobReduction} and another union bound, with probability $\ge 95/100$:
\begin{align*}
\norm{T-\avg \left (F_{M_{best}} W_{best} F_{M_{best}}^* \right )}_2 \le 2160 \sqrt{\norm{T-T_k}_2 \cdot \tr(T) + \frac{\tr(T-T_k) \cdot \tr(T)}{k}} + 217\eps \norm{T}_2.
\end{align*}
Finally, we note that by our bound on $\max(s_1,s_2)$ the entry sample complexity is bounded by
$$n \cdot \max(s_1,s_2) = O \left (n \cdot m^2 \log m \log(1/\alpha)\right)  = O \left (k^3 \log k\cdot \log \left (\frac{d}{\eps}\right ) + k^2 \log k \cdot \log d\cdot \log^2 \left (\frac{d}{\eps} \right ) \right ),$$
which completes the theorem.
\end{proof}

Finally, we can consider the case when $T$ has low stable rank: $\frac{\tr(T)}{\norm{T}_2} \le s$. In this case $\norm{T-T_k}_2 \le \frac{\tr(T)}{k} \le \frac{s \norm{T}_2}{k}$. Setting $k = \frac{cs^2}{\eps^2}$ for large enough $c$ we can apply Theorem \ref{thm:sftT} to give:
\begin{theorem}[Toeplitz Covariance Estimation via Sparse Fourier Transform -- Low Stable Rank]\label{thm:sft}
Consider PSD $T \in \R^{d \times d}$ with stable rank $\frac{\tr(T)}{\norm{T}_2} \le s$. Algorithm \ref{alg:sft} run on inputs drawn from $\mathcal{N}(0,T)$ with $c_1,c_2$ sufficiently large, rank $m \ge \frac{c_3 s^2}{\eps^2}$, $n \ge c_4 \max \left ( \frac{\log \left (\frac{d}{\eps}\right  ) \log d}{\eps^2}, m \right )$, and discretization $\alpha = \frac{\eps^2}{c_5 d^{3.5}}$ for sufficiently large constants $c_3,c_4,c_5$ probability $\ge 19/20$, (1): returns $\bar T$ that satisfies
$$\left\| T - \bar T \right\|_2 \le \eps \norm{T}_2$$
and (2): has  entrywise sample complexity $O \left (\frac{s^6}{\eps^6} \log^2 \left (\frac{d}{\eps}\right ) + \frac{s^4}{\eps^4} \log^4 \left (\frac{d}{\eps} \right ) \right )$.
\end{theorem}

\section{Lower Bounds}\label{sec:lower}
In this section we give lower bounds for Toeplitz  covariance estimation, demonstrating that polynomial dependencies on the rank  are unavoidable in the TSC of the problem.
This validates our low-rank bounds in Section~\ref{sec:fourier}, and also, when we take the rank to be linear in the dimension, demonstrates that a linear dependence on the dimension is unavoidable.

All lower bounds in this section will apply to the case when the matrix is in fact circulant.
When $T$ is a symmetric circulant matrix, it is not hard to show that (with the possible exception of the all-ones eigenvector), all of its eigenvectors come in conjugate pairs, and the corresponding eigenvalues for these eigenvectors are the same.
More concretely, one can show that for all $j = 1, \ldots, \lfloor d / 2 \rfloor$, the following two vectors are orthonormal eigenvectors of $T$ with the same eigenvalue:
\begin{align*}
u_j = \sqrt{\frac{1}{d}} \cdot \left[ \begin{array}{c}
1 \\
\cos \Paren{2 \pi j / d}\\
\cos \Paren{4 \pi j / d}\\
\vdots\\
\cos \Paren{(d - 1) \pi j / d} \end{array} \right] \; ,~~
v_j = \sqrt{\frac{1}{d}} \cdot \left[ \begin{array}{c}
1 \\
\sin \Paren{2 \pi j / d}\\
\sin \Paren{4 \pi j / d}\\
\vdots\\
\sin \Paren{(d - 1) \pi j / d} \end{array} \right] \; ,
\end{align*}
and the matrix 
\[
\Phi = \left\{ \begin{array}{ll}
\left[\frac{1}{\sqrt{d}} \one, u_1, v_1, u_2, v_2, \ldots, u_{\lfloor d / 2 \rfloor}, v_{\lfloor d / 2 \rfloor}\right] & \mbox{if $d$ is odd}\\
\left[\frac{1}{\sqrt{d}} \one, u_1, v_1, u_2, v_2, \ldots, u_{d / 2} \right] & \mbox{if $d$ is even} \end{array} \right. \; ,
\] 
is an orthonormal matrix that diagonalizes $T$.
For simplicity of exposition in the remainder of this section we will assume that $d$ is odd, however, the ideas immediately transfer over to the case where $d$ is even.
For any subset $S \subseteq \{1, \ldots,  \lfloor d / 2 \rfloor\}$, let $D_S \in \R^{d \times d}$ denote the diagonal matrix which has diagonal entries $(D_S)_{2j, 2j} = (D_S)_{2j + 1, 2j +1} = 1$ for all $j \in S$, and $0$ elsewhere, and let $\Gamma_S = \Phi D_S \Phi$.
This matrix is PSD, circulant, and has eigenvalues which are either $0$ and $1$, and the nonzero eigenvalues have corresponding eigenvectors exactly $u_j, v_j$ for $j \in S$.

\subsection{Non-adaptive lower bound in the noiseless case}
In this section we consider the case where the algorithm is non-adaptive, and $T$ is exactly rank-$k$.
Note that in this setting, the Prony's method based algorithm of Theorem \ref{thm:prony-exact-roots} gives an upper bound of $O(k \log k)$ total samples using a non-adaptive algorithm (with ESC $O(k)$).
We show that this TSC bound is tight up to the factor of $\log^2 k$.
\begin{theorem}
\label{thm:nonadaptive-lb}
Let $k = \omega (1)$.
Any (potentially randomized) algorithm that reads a total of $m$ entries of  $x^{(1)},\ldots,x^{(n)} \sim \normal (0, T)$ non-adaptively where $T$ is an (unknown) rank-$k$ circulant matrix and outputs $\That$ so that $\norm{T - \That}_2 \leq \frac{1}{10} \norm{T}_2$ with probability $\geq 2/3$ must satisfy $m = \Omega (k/\log k)$.
\end{theorem}
\noindent We observe that the constant $9 / 10$ is arbitrary here, and can be replaced with any constant larger than $1/2$, by changing the constant in the dependence between $m$ and $k$.
Moreover, observe that by letting $k = \Theta (d)$, we get a nearly linear total sample complexity lower bound in the full-rank case, nearly matching Theorem \ref{thm:linear}.
It is an interesting question for future work if this can be made tight up to constants (i.e. without log factors). 

By Yao's minimax principle, it suffices to demonstrate a distribution over $\That$ so that any deterministic algorithm that succeeds with probability $\geq 9 / 10$ (over the choice of $\That$ and the random samples $x^{(1)},\ldots,x^{(n)}$) requires $m= \Omega(k/\log k)$.
Without loss of generality, assume that $n = m$: since the algorithm inspects at most $m$ entries, chosen ahead of time, clearly we may assume it takes at most $m$ full samples from the distribution.
If it takes fewer samples, it can simply ignore the remaining samples.
Observe that any non-adaptive algorithm for learning $T$ to error $\eps$ is fully characterized by the following:
\begin{itemize}
\item subsets $S_1, \ldots, S_m \subseteq [d]$, where $S_i$ is the set of positions that the algorithm inspects for sample $i$, and
\item a function $f: \R^{S_1} \times \ldots \R^{S_m} \to \R^{d \times d}$ so that if $x^{(1)}, \ldots, x^{(m)} \sim \normal (0, T)$, with probability $\ge 2/3$
\[
\norm{T - f((x^{(1)})_{S_1}, \ldots, (x^{(m)})_{S_m})}_2 \leq \eps \norm{T}_2 \; .
\]
\end{itemize}
\noindent
We will show that the following random distinguishing problem requires $m = \Omega(k/\log k)$:
\begin{problem}
\label{prob:nonadaptive}
We consider the following two player game with parameter $k'$.
\begin{itemize}
\item
First, Alice chooses subsets $S_1, \ldots, S_m \subseteq [d]$, where $\sum_{i = 1}^m |S_i| \leq m$.
\item
Then, Bob chooses two random subsets $R_0, R_1$ of $[(d - 1) / 2]$ as follows: to form $R_0$, he includes every $i \in [(d - 1) / 2]$ with probability $O(\frac{k' \log k'}{d})$.
Then, to form $R_1$, if $R_0 = \emptyset$, he lets $R_1 = \emptyset$. 
Otherwise, he chooses a random $i \in R_0$, and removes it.
\item
If $R_0 = \emptyset$ or $R_1 = \emptyset$, we say Alice \emph{succeeds}.
\item
Otherwise, Bob draws $\ell \sim \mathrm{Ber}(1/2)$, and generates $x^{(1)}, \ldots, x^{(m)} \sim \normal (0, \Gamma_\ell)$, where $\Gamma_\ell = \Gamma_{R_\ell}$, for $\ell = \{0, 1\}$, and gives Alice the vectors $(x^{(1)})_{S_1}, \ldots, (x^{(m)})_{S_m}$, as well as $R_0$ and $R_1$.
\item
Alice can then do any arbitrary deterministic postprocessing on $(x^{(1)})_{S_1}, \ldots, (x^{(m)})_{S_m}$ and $R_0$ and $R_1$ and outputs $\ell' \in \{0, 1\}$, and we say she \emph{succeeds} if $\ell = \ell'$.
Otherwise, we say she \emph{fails}.
\end{itemize}
\end{problem}
\noindent
Theorem~\ref{thm:nonadaptive-lb} will be a simple consequence of the following two lemmata:
\begin{lemma}
\label{lem:learning-to-distinguishing}
Let $k = \omega(1)$, and suppose there exists an (potentially randomized) algorithm, that reads $m$ entries non-adaptively from $x^{(1)},\ldots, x^{(m)} \sim \normal (0, T)$ where $T$ is an (unknown) rank-$k$ circulant matrix and outputs $\That$ so that $\norm{T - \That}_2 \leq \frac{1}{10} \norm{T}_2$ with probability $\geq 2/3$.
Then, Alice can succeed at Problem~\ref{prob:nonadaptive} where $k' = O(k / \log k)$ with probability $\geq 1/2 + 1/5$ with $m$ queries.
\end{lemma}

\begin{lemma}
\label{lem:distinguishing-lower-bound}
Alice cannot succeed at Problem~\ref{prob:nonadaptive} with probability greater than $1/2 + 1/5$ unless $m = \Omega(k')$.
\end{lemma}
\noindent
Observe that by combining these two lemmata, we immediately obtain Theorem~\ref{thm:nonadaptive-lb}.
We now prove these two lemmata in turn.
\begin{proof}[Proof of Lemma~\ref{lem:learning-to-distinguishing}]
Our reduction will be the trivial one.
Any non-adaptive learning algorithm immediately gives a routine for the distinguishing problem: simply run the learning algorithm, obtain output $\That$, and output $\argmin_{\ell \in \{0, 1\}} \norm{\That - \Gamma_\ell}_2$.
This algorithm clearly also only pays for $m$ samples.
We now show that any algorithm with the guarantees as in Theorem~\ref{thm:nonadaptive-lb}, when transferred in this way, immediately yields a solution to Problem~\ref{prob:nonadaptive}.

In Problem~\ref{prob:nonadaptive}, condition on the event that $|R_0| \neq \emptyset, |R_1| \neq \emptyset,$ and $|R_1| < |R_0| \leq k / 2$.
By basic Chernoff bounds, we know this happens with probability $\geq 99 / 100$ for $k = \omega(1)$.
Conditioned on this event, we know that $\norm{\Gamma_\ell}_2 = 1$, $\norm{\Gamma_0 - \Gamma_1}_2 = 1$, and $\rank(\Gamma_\ell) \leq k$ for $\ell \in \{0, 1\}$.
Thus, if this event occurs, if the samples are drawn from $\Gamma_\ell$, then the guarantees of the supposed algorithm imply that with probability $\geq 2/3$, it outputs $\That$ so that $\norm{\That - \Gamma_\ell}_2 \leq 1 / 10$.
By triangle inequality, this implies that $\norm{\That - \Gamma_{1 - \ell}}_2 \geq 9 / 10$.
Hence in this case, the distinguishing algorithm succeeds so long as the learning algorithm succeeds, so overall the distinguishing algorithm succeeds with probability at least $2/3 - 99 / 100 \geq 1/2 + 1/5$.
\end{proof}

\begin{proof}[Proof of Lemma~\ref{lem:distinguishing-lower-bound}]
Let $Y$ denote the distribution of the entries of the samples that the algorithm looks at, if we stack them to form a single, $m$-dimensional column vector.
It is not hard to see that if the samples are drawn from $\normal(0, \Gamma_\ell)$, for $\ell \in \{0, 1\}$, then $Y \sim \normal (0, \Xi_\ell)$, where 
\[
\Xi_\ell = \left[ 
\begin{array}{cccc}
\Paren{\Gamma_\ell}_{S_1} & 0 & \cdots & 0 \\
0 & \Paren{\Gamma_\ell}_{S_2} & \ldots & 0 \\
0 & 0 & \ddots & 0 \\
0 & 0 & \cdots & \Paren{\Gamma_\ell}_{S_m} 
\end{array}
\right] \; .
\]
We will show that if $m = O(k')$, with probability $99 / 100$ over the choice of $R_0$ and $R_1$, $$\dtv (\normal (0, \Xi_0), \normal (0, \Xi_1)) \leq 1 / 10.$$ We can then condition on the event that both $R_0, R_1 \neq \emptyset$, as in this case the distributions are not well-formed (and Alice succeeds by default).
By a Chernoff bound, it is not hard to see that this occurs with overwhelming probability as long as $k = \omega (1)$.
This implies that conditioned on this choice of $R_0$ and $R_1$, with probability $9 / 10$, one cannot distinguish between $Y \sim \normal (0, \Xi_0)$ and $Y \sim \normal (0, \Xi_1)$.
Since the algorithm is only allowed to take $Y \sim \normal (0, \Xi_\ell)$ and then do post-processing on it to distinguish between $R_0$ and $R_1$, by data processing inequalities (see e.g.~\cite{cover2012elements}), this implies that the algorithm cannot succeed with probability more than $1/2 + 1/10 + 1 / 100 < 1/2 + 1 / 5$ over the choice of both $R_0, R_1$, and $Y$.

We proceed to prove the above total variation distance bound.
It is well-known (e.g. via Pinsker's inequality) that for any two PSD matrices $M_1, M_2$, we have $\dtv (\normal (0, M_1), \normal (0, M_2)) \leq O \Paren{\Norm{I - M_1^{-1/2} M_2 M_1^{-1/2}}_F}$.
This can further be upper bounded by 
\[
\Norm{I - M_1^{-1/2} M_2 M_1^{-1/2}}_F = \Norm{M_1^{-1/2} \Paren{M_1 - M_2} M_1^{-1/2}}_F \leq \Norm{M_1^{-1}}_2 \Norm{M_1 - M_2}_F \; .
\]
Instantiating this bound for our case, and using the shared block-diagonal structure of our matrices $\Xi_0$ and $\Xi_1$, we get that 
\begin{equation}
\label{eq:TV-block-bound}
\dtv (\normal (0, \Xi_0), \normal (0, \Xi_1)) \leq C \cdot \Paren{\max_{i \in \{1, \ldots m\}} \norm{\Paren{\Gamma_0}_{S_i}^{-1}}_2 } \sqrt{\sum_{i = 1}^m \Norm{\Paren{\Gamma_0 - \Gamma_1}_{S_i}}_F^2} \; .
\end{equation}
We will bound both terms on the RHS of~\eqref{eq:TV-block-bound} separately.
We first bound the spectral norm.
Let $\Phi_{S_i}$ denote the restriction of $\Phi$ to the columns in $S_i$.
Observe that $\Phi_{S_i}$ has orthonormal columns and hence its row leverage scores (Def. \ref{def:lev_score}) are all identically $|S_i| / d \leq m / d = O(k'/d)$.
Then, $(\Gamma_0)_{S_i}$ is formed by subsampling the rows of $\Phi_{S_i}$ according to $R_0$, and taking the outer product of this matrix with itself. Since each entry is included in $R_0$ with probability $\frac{ck' \log k'}{d}$ for some constant $c$, by Claim~\ref{claim:lev_sampling} 
we have that with probability $\geq 1 - 1/(k')^2$, $\norm{\frac{d}{ck' \log k'} \Paren{\Gamma_0}_{S_i} -  \Phi_{S_i} \Phi_{S_i}^\top}_2 \leq 1/10$.
In particular, since the smallest (nonzero) singular value of $\Phi_{S_i}$ is at least the smallest singular value of $\Phi$, which is 1, we obtain that for every $i$, we have $(\Gamma_0)_{S_i} \succeq \frac{9ck' \log k'}{10d} \cdot I$ with probability at least $1 - 1/(k')^2$.
Thus, by a union bound, if $m \leq k'$, we obtain that with probability at least $1 - 1/k'$,
\begin{equation}
\label{eq:nonadaptive-spectral-bound}
\max_{i \in \{1, \ldots m\}} \norm{\Paren{\Gamma_0}_{S_i}^{-1}}_2 = O\left  (\frac{d}{k' \log k'} \right ) \; .
\end{equation}

We now turn our attention to the second term on the RHS of~\eqref{eq:TV-block-bound}.
Let $j \in [(d - 1) / 2]$ be the unique index in which $R_0, R_1$ differ.
Then, observe that 
\begin{align*}
\Norm{\Paren{\Gamma_0 - \Gamma_1}_{S_i}}_F &= \Norm{(u_j)_{S_i} (u_j)_{S_i}^\top + (v_j)_{S_i} (v_j)_{S_i}^\top}_F \\
&\leq \Norm{(u_j)_{S_i} (u_j)_{S_i}^\top}_F + \Norm{(v_j)_{S_i} (v_j)_{S_i}^\top}_F \\
&= \Norm{(u_j)_{S_i}}_2^2  + \Norm{(v_j)_{S_i}}_2^2 \leq \frac{2 |S_i|}{d} \; .
\end{align*}
Hence 
\begin{equation}
\label{eq:nonadaptive-frob-bound}
\sqrt{\sum_{i = 1}^m \Norm{\Paren{\Gamma_0 - \Gamma_1}_{S_i}}_F^2} \leq \frac{2}{d} \sqrt{\sum_{i = 1}^m |S_i|^2} \leq \frac{2}{d} \sum_{i = 1}^m |S_i| \leq \frac{2 m}{d} \; .
\end{equation}
Plugging~\eqref{eq:nonadaptive-spectral-bound} and~\eqref{eq:nonadaptive-frob-bound} into~\eqref{eq:TV-block-bound} yields that for $k'$ sufficiently large, with probability $\geq 99 / 100$, if $m = O( k')$,
\[
\dtv (\normal (0, \Xi_0), \normal (0, \Xi_1)) = O \left (\frac{1}{\log k'} \right ) \le \frac{1}{10} \; ,
\]
which completes the proof.
\end{proof}

\subsection{Adaptive lower bound in the noisy case}
In this section, we consider a somewhat different setting.
We allow the algorithm now to be adaptive; that is, the choice of entry it can inspect is no longer specified ahead of time, as it was in the previous section, but rather can now depend on the answers it has seen so far.
However, we also ask the algorithm to solve a slightly harder question, where the covariance is still circulant, but is now only approximately low-rank.
By using information theoretic techniques inspired by~\cite{HassaniehIndykKatabi:2012}, we demonstrate that in this setting, a linear dependence on the rank is still necessary.
Note that by taking the rank to be linear in $d$, this implies that the near linear sample complexity result of Theorem \ref{thm:linear} is tight for full-rank matrices, up to logarithmic factors. This is in spite of that fact that the algorithm used, and in fact all or algorithms, make non-adaptive queries.

\begin{problem}
\label{prob:adaptive-noisy}
Let $\alpha \in [0, 1]$ be a parameter.
Let $T \in \R^{d \times d}$ be an (unknown) PSD Toeplitz matrix, and let $k \leq d$ be a known integer.
Given samples $X_1, \ldots, X_m \sim \normal (0, T)$, and $m$ adaptively chosen entrywise lookups into $X_1, \ldots, X_m$, output $\That$ so that with probability $\ge 1/10$, we have $\Norm{T - \That}_2 \leq \frac{1}{10} \Norm{T}_2 + \alpha d \Norm{T - T_k}_2$, where $T_k = \argmin_{\rank-k\ M} \norm{T-M}_2$.
\end{problem}
\noindent
As in the previous section, the constants in the statement of Problem~\ref{prob:adaptive-noisy} are more or less arbitrary, and the results hold for any (sufficiently small) constants.
Our main result in this section is:
\begin{theorem}
\label{thm:adaptive-lb}
Let $\alpha \geq 10 / k$.
Any algorithm for Problem~\ref{prob:adaptive-noisy} requires 
\[
m \geq \Omega \Paren{k \log (d / k) \cdot \min \Paren{\frac{1}{\log \alpha k}, \frac{1}{\log \log (d k)}}} \; .
\]
\end{theorem}
\noindent
In particular, we note a few interesting regimes.
\begin{itemize}
\item
When $\alpha$ is a small constant, we get a linear lower bound, on the sample complexity, even for any algorithm achieving a weak tail guarantee of $\norm{T - \That}_2 \leq \frac{1}{10} \norm{T}_2 + \frac{d}{10} \norm{T - T_k}_2$.
Notice that (up to constants in the tail), this is weaker than guarantee achieved in Theorem \ref{thm:sftT}.
\item
When $\alpha = 10 / k$, this gives a super-linear lower bound of $\Omegatilde(k \log (k / d))$ for algorithms that achieve a stronger tail guarantee.
In analogy to results on sparse Fourier transforms, we conjecture that this bound is in fact tight (up to log log factors).
\item
When $\alpha = O(1/k)$ and additionally $k = \Omega (d)$, this implies that \emph{any} algorithm for solving Problem~\ref{prob:main} requires total sample complexity $\Omega (d)$.
In particular, this shows that the simple algorithm of Theorem \ref{thm:linear} based on taking full samples is in fact optimal for in terms of total sample complexity, up to logarithmic factors.
\end{itemize}
\noindent
Our lower bound goes through information theory, via the following game between Alice and Bob.
\begin{problem}
\label{prob:adaptive}
Let $\alpha$ be as in Problem~\ref{prob:adaptive-noisy}.
We consider the following two-player game.
\begin{itemize}
\item First, Alice selects a subset $S \subseteq [(d - 1) / 2]$ of size $|S| = k / 2$ uniformly at random, then generates samples $x^{(1)}, \ldots, x^{(m)} \sim \normal (0, \Gamma_S + \frac{1}{\alpha d} I)$.
\item Then, Bob makes $m$, sequentially adaptive queries $(i_1, j_1), \ldots, (i_m, j_m)$.
More precisely, for all times $t \leq m$, he sees $x^{(i_t)}_{j_t}$, then decides $(i_{t + 1}, j_{t + 1})$ as a deterministic function of the answers he has seen so far.
\item After $m$ queries, Bob must output $S'$ as a deterministic output of the answers he has seen so far. 
We say that Bob \emph{succeeds} if $S' = S$, otherwise we say he \emph{fails}.
\end{itemize}
\end{problem}
\noindent
As before, the theorem will be a simple consequence of the following two lemmata along with Yao's minimax principal:
\begin{lemma}
\label{lem:adaptive-learn-to-dist}
Let $\alpha \geq 10 / k$.
Suppose there exists an algorithm for Problem~\ref{prob:adaptive-noisy} that succeeds with probability $\geq 1 / 10$ with $m$ entrywise lookups.
Then, there exists a strategy for Bob to succeed at Problem~\ref{prob:adaptive} with probability $\geq 1 / 10$ with the same number of samples.
\end{lemma}
\begin{lemma}
\label{lem:adaptive-dist-lb}
Let $\alpha \geq 10 / k$.
Any deterministic strategy for Bob for Problem~\ref{prob:adaptive} that succeeds with probability $\geq 1 / 10$ requires 
\[
m = \Omega \Paren{k \log (d / k) \cdot \min \Paren{\frac{1}{\log \alpha k}, \frac{1}{\log \log (d k)}}} \; .
\]
\end{lemma}
\noindent
We now prove these two lemmata in turn.
\begin{proof}[Proof of Lemma~\ref{lem:adaptive-learn-to-dist}]
The reduction is again straightforward: given an algorithm for Problem~\ref{prob:adaptive-noisy} and an instance of Problem~\ref{prob:adaptive}, we run the algorithm to obtain some matrix $\That$, and we output $\argmin_{|S| = k / 2} \norm{\That - \Gamma_S}_2$.
Let $S^*$ be the true subset, and suppose the algorithm succeeds.
For all $S$ with $|S| = k / 2$, we have $\Gamma_S$ is rank-$k/2$, and $\norm{\Gamma_S + \frac{1}{\alpha d} I}_2 = 1 + \frac{1}{\alpha d} \leq 1 + 1/10$.
Thus, if the algorithm succeeds, we have $\norm{\That - (\Gamma_{S^*} + \alpha I)}_2 \leq \frac{1}{5} + \frac{1}{\alpha d} \ll 1/2$ by our choice of $\alpha$.
Since $\norm{\Gamma_S - \Gamma_{S'}}_2 = 1$ for all $S \neq S'$, we conclude that $\argmin_{|S| = k / 2} \norm{\That - \Gamma_S}_2 = S^*$, in this case, so the algorithm successfully solves the distinguishing problem.

\end{proof}
\noindent
Before we prove Lemma~\ref{lem:adaptive-dist-lb}, we first need the following moment bound.
\begin{lemma}
\label{lem:max-moment-bound}
Let $M \in \R^{d \times d}$ have rows $M_1, \ldots, M_d$, and let $s = \max_{j =1, \ldots, d} \Norm{M_j}_2^2$. 
Let $X_1, \ldots, X_m \sim \normal (0, M M^\top)$.
Then 
\[
\E \Brac{\max_{i \in [m], j \in [d]} X_{i,j}^2} = O(s \log (m d)) \; .
\]
\end{lemma}
\begin{proof}
For all $i \in [m]$ and $j \in [d]$, we have that $X_{i,j} \sim \normal (0, \norm{M_j}_2^2)$.
Thus, by standard Gaussian concentration, there exists some universal constant $C$ so that for all $t > 0$,
\[
\Pr \Brac{X_{i,j}^2 > t} \leq \exp \Paren{-\frac{C t}{\norm{M_j}_2^2}} \leq \exp \Paren{- \frac{C t}{s}} \; .
\]
Thus, by a union bound, and since probability is bounded by $1$, we know that for all $t > 0$, we have 
\[
\Pr \Brac{\max_{i \in [m], j \in [d]} X_{i,j}^2 > t} \leq \min\Paren{1, md \cdot \exp \Paren{- \frac{C t}{s}}} \; .
\]
Hence
\begin{align*}
\E \Brac{\max_{i \in [m], j \in [d]} X_{i,j}^2 } &= \int_{0}^\infty \Pr \Brac{\max_{i \in [m], j \in [d]} X_{i,j}^2  \geq t} dt \\
&\leq \int_0^\infty \min\Paren{1, md \cdot \exp \Paren{- \frac{C t}{s}}} dt \\
&\leq \int_0^{2 s \log(md) / C} 1 dt + \int_{2 s \log(md) / C}^\infty \exp \Paren{- \frac{C t}{2 s}} dt \\
&= \frac{2s \log (md)}{C} + \int_{2 s \log(md) / C}^\infty \exp \Paren{- \frac{C t}{2 s}} dt \\
&\stackrel{(a)}{\leq} \frac{2s \log (md)}{C} + \int_0^\infty \exp \Paren{- \frac{C t}{2 s}} dt \\
&= \frac{2s \log (md)}{C} + \frac{s}{C'} \; ,
\end{align*}
for some other constant $C'$, where (a) follows since by explicit calculation, we have that for all $t \geq 2 s \log(md) / C$, we have $md \leq \exp \Paren{\frac{C t}{2 s}}$.
\end{proof}

\begin{proof}[Proof of Lemma~\ref{lem:adaptive-dist-lb}]
We will assume familiarity with the basics of information theory (see e.g.~\cite{cover2012elements}).
We will use $I(X; Y)$ to denote mutual information, $H$ to denote discrete entropy, and $h$ to denote differential entropy.

Suppose that there exists a protocol for Bob that succeeds with probability $\geq 9 / 10$.
Since there exist $\binom{(d - 1) / 2}{k}$ possibilities for $S,$ by Fano's inequality, we have $H(S | S') \leq 1 + \frac{9}{10} \log \binom{(d - 1) / 2}{k / 2}$.
Thus, 
\begin{equation}
\label{eq:info-lb}
I(S; S') = H(S) - H(S | S') \geq -1 + \frac{1}{10} \log \binom{(d - 1) / 2}{k / 2} = \Omega (k \log (d / k)) \; .
\end{equation}
The main work will be to show a corresponding upper bound, namely,
\begin{equation}
\label{eq:info-ub}
I(S; S') \leq O(m (\log (\alpha k) + \log \log (md)) ) \; .
\end{equation}
Combining~\eqref{eq:info-lb} and~\eqref{eq:info-ub} and solving for $m$ immediately yields Lemma~\ref{lem:adaptive-dist-lb}.
Thus to complete the proof it suffices to verify Equation~\ref{eq:info-ub}.
WLOG we may assume that $(i_t, j_t) \neq (i_{t'}, j_{t'})$ for all $t \neq t'$, as otherwise that look-up give exactly no additional information.
Observe that since $x^{(\ell)} \sim \normal (0, \Gamma_S + \frac{1}{\alpha d} I)$, we can write $x^{(\ell)} = y^{(\ell)} + \eta^{(\ell)}$, where $y^{(\ell)} \sim \normal (0, \Gamma_S)$, and $\eta^{(\ell)} \sim \normal (0, \alpha I)$ are independent.
Let $X^{\Paren{t}} = x^{(i_t)}_{j_t}$, $Y^{\Paren{t}} = y^{(i_t)}_{j_t}$ and $\eta^{\Paren{t}} = \eta^{(i_t)}_{j_t}$, so that $X^{\Paren{t}} = Y^{\Paren{t}} + \eta^{\Paren{t}}$.
Observe that since all the $(i, j)$ pairs are distinct, we have that each $\eta^{\Paren{t}}$ is a completely independent Gaussian.
By the data processing inequality, we have that $I(S; S') \leq I(S; X^{\Paren{1}}, \ldots, X^{\Paren{m}})$.
Then, for all $t = 1, \ldots, m$, we have
\begin{align*}
I(S; X^{\Paren{t}} | X^{\Paren{1}}, \ldots, X^{\Paren{t - 1}}) &= h(X^{\Paren{t}} | X^{\Paren{1}}, \ldots, X^{\Paren{t - 1}}) - h(X^{\Paren{t}} | S, X^{\Paren{1}}, \ldots, X^{\Paren{t - 1}}) \\
&= h(X^{\Paren{t}} | X^{\Paren{1}}, \ldots, X^{\Paren{t - 1}}) - h(Y^{\Paren{t}} + \eta^{\Paren{t}} | S, X^{\Paren{1}}, \ldots, X^{\Paren{t - 1}}) \\
&\stackrel{(a)}{\leq} h(X^{\Paren{t}} | X^{\Paren{1}}, \ldots, X^{\Paren{t - 1}}) - h(\eta^{\Paren{t}}) \\
&\stackrel{(b)}{\leq} h(X^{\Paren{t}}) - h(\eta^{\Paren{t}}) \; , \numberthis \label{eq:info-bound}
\end{align*}
where (a) follows since $\eta^{\Paren{t}}$ is independent of all other quantities in that expression, and (b) follows by monotonicity of information.

We now seek to upper bound the variance of the random variable $X^{\Paren{t}} = x^{(i_t)}_{j_t}$.
The main difficulty for doing so is that $(i_t, j_t)$ is allowed to depend arbitrarily on $S$, and the previous $X^{\Paren{\ell}}$  for $\ell < t$.
We will circumvent this by simply noting that the second moment of this random variable is certainly upper bounded by the second moment of the maximum over all possible $(i, j)$ of $\Paren{x^{(i)}_{j}}^2$.
Observe that for any $S$, we can write $\Gamma_S = \Phi D_S (\Phi D_S)^\top$, and moreover, the maximum squared $\ell_2$ norm of any row in $\Phi D_S$ is $k / d$.
Hence, we obtain
\begin{align*}
\E \Brac{\Paren{x^{(i)}_{j}}^2 | S} \leq \E \Brac{\max_{i \in [m], j \in [d]} \Paren{x^{(i)}_{j}}^2 \middle| S} &\leq 2 \E \Brac{\max_{i \in [m], j \in [d]} \Paren{y^{(i)}_{j}}^2 | S} + 2 \E \Brac{\max_{i \in [m], j \in [d]} \Paren{\eta^{(i)}_{j}}^2 | S} \\
&\stackrel{(a)}{=} 2 \E \Brac{\max_{i \in [m], j \in [d]} \Paren{y^{(i)}_{j}}^2 | S} + 2 \E \Brac{\max_{i \in [m], j \in [d]} \Paren{\eta^{(i)}_{j}}^2} \\
&\stackrel{(b)}{=} O \Paren{\frac{k}{d} \log (m d)} \; , 
\end{align*}
where (a) follows since $\eta$ is independent of $S$, and (b) follows from two applications of Lemma~\ref{lem:max-moment-bound}, and our choice of $\alpha$.
Thus, by taking an expectation over $S$, we obtain that 
\[
\Var \Brac{(X^{(t)})^2} \leq \E \Brac{(X^{(t)})^2} = O \Paren{\frac{k}{d} \log (m d)} \; .
\]
Since Gaussians maximize the entropy of a distribution with fixed variance, we can upper bound~\eqref{eq:info-bound} by 
\begin{align*}
I(S; X^{\Paren{t}} | X^{\Paren{1}}, \ldots, X^{\Paren{t - 1}}) &\leq h \Paren{\normal \Paren{0, O \Paren{\frac{k}{d} \log (m d)}}} - h \Paren{\normal \Paren{0, \frac{1}{\alpha d}}} \\
&= O \Paren{\log (\alpha k) + \log \log (md) } \; .
\end{align*}
Equation~\eqref{eq:info-ub} immediately follows from this and the chain rule for information.
\end{proof}

\section*{Conclusion}
Our work provides some of the first non-asymptotic bounds on the sample complexity of recovering Toeplitz covariance matrices. We analyze several classical techniques, including those based on obtaining samples according to a sparse ruler, and Prony's method. Additionally, we are able to improve on these techniques with new algorithms. For estimating full-rank covariance matrices, we introduce a new class of rulers with sparsity between $\sqrt{d}$ and $d$, that let us smoothly trade between the optimal \emph{entry sample complexity} of methods based on $O(\sqrt{d})$-sparse rulers, and the optimal \emph{vector sample complexity} of methods based on a fully dense ruler. 

For estimating rank-$k$ covariance matrices for $k \ll d$, we introduce a \emph{randomly constructed} sampling set, which can be used with a new recovery algorithm to achieve entry sample complexity and total sample complexity that scales just logarithmically with $d$. This surpasses the $O(\sqrt{d})$ limitation for methods based on sparse rulers, while offering significantly improved robustness in comparison to Prony's method: the method works even when $T$ is only close to rank-$k$.

Technically, our work requires a combination of classical tools from harmonic analysis with techniques from theoretical computer science and randomized numerical linear algebra. We hope that it can help initiate increased collaboration and cross-fertilization of ideas between these research areas and the signal processing community.

\section*{Acknowledgements}
We  thank Eric Price for helpful discussion and clarification on \cite{ChenKanePrice:2016}. We also thank Haim Avron for helpful conversations on the techniques used in Section \ref{sec:fourier}.

\bibliographystyle{alpha}
\bibliography{toeplitz}	
\clearpage
\appendix

\section{Leverage Score Facts and Properties}\label{app:leverage_scores}
We begin by proving the alternative characterizations of leverage scores from Fact \ref{fact:min_char} and Fact \ref{fact:max_char}. 

\begin{proof}[Proof of Fact \ref{fact:max_char}]
The minimization problem is simply an underconstrained regression problem and accordingly, the least norm solution 
can be obtained by setting $y = a_j(A^*A)^+A^*$. This yields
\begin{align*}
\min \|y\|_2^2 = a_j(A^*A)^+A^*A(A^*A)^+{a_j^*} =  a_j(A^*A)^+{a_j^*}
\end{align*} 
proving equivalence to our original Definition \ref{def:lev_score}.
\end{proof}
Fact \ref{fact:min_char} implies that upper bounds on $\tau_j(A)$ can be obtained by exhibiting a linear combination of rows in $A$ that reconstructs $a_j$. The squared $\ell_2$ norm of the coefficients in this linear combination upper bounds the leverage score. More generally, when $a_j$ ``aligns well'' with other rows in $A$, its leverage score is smaller, indicating that it is less ``unique''.

\begin{proof}[Proof of Fact \ref{fact:max_char}] Our proof follows the one for continuous operators in \cite{AvronKapralovMusco:2017}.
	For $j \in [d]$, let $m = \max_{y \in \C^s} \frac{\left|\left(Ay\right)_j\right|^2}{\norm{Ay}_2^2}$. We want to establish that for row vector $a_j \in \C^{d \times 1}$,
	\begin{align*}
	m = \tau_j(A)= a_j(A^*A)^+a_j^*.
	\end{align*}
	To do so, we will argue separately that $m \geq  \tau_j(A)$ and $m \leq  \tau_j(A)$. To see the former, we simply plug in $y = (A^*A)^+a_j^*$ to the maximization problem to obtain:
	\begin{align*}
	m\geq \frac{\left|\left(A(A^*A)^+a_j^*\right)_j\right|^2}{\norm{A(A^*A)^+a_j^*}_2^2} = \frac{|a_j (A^*A)^+a_j^*|^2}{a_j(A^*A)^+A^*A(A^*A)^+a_j^*} = a_j (A^*A)^+a_j^* = a_j (A^*A)^+a_j^*) .
	\end{align*}	
	To show that $m \leq  \tau_j(A)$ we first note that we can parameterize our maximization problem:
	\begin{align*}
	\max_{y \in \C^s} \frac{\left|\left(Ay\right)_j\right|^2}{\norm{Ay}_2^2} = \max_{w \in \C^s} \frac{\left|\left(A(A^*A)^+w\right)_j\right|^2}{\norm{A(A^*A)^+w}_2^2} = \max_{w \in \C^s} \frac{\left|a_j(A^*A)^+w\right|^2}{w^*(A^*A)^+w}
	\end{align*}
	Since $(A^*A)^+$ is PSD and thus has a Hermitian square root $Z = Z^*$ with $(A^*A)^+ = Z^*Z$, we can write $a_j(A^*A)^+w = (a_jZ^*)(Zw)$ and apply Cauchy-Schwarz inequality.  Specifically, for any $w \in \C^s$, we have that:
	\begin{align*}
	\frac{|a_j(A^*A)^+w|^2}{w^*(A^*A)^+w} \leq \frac{|a_j(A^*A)^+a_j^*||w^*(A^*A)^+w|}{w^*(A^*A)^+w} =  a_j (A^*A)^+a_j^*.
	\end{align*}
	So we conclude that  $m \leq  \tau_j(A)$, which completes the proof.
\end{proof}

Fact \ref{fact:max_char} is ``dual'' to Fact \ref{fact:min_char}. It implies that the leverage score of a row $a_j$ is higher if we can find a vector $y$ whose inner product with $a_j$ has large squared magnitude in comparison to its inner product with other rows in the matrix (the sum of which is $\|Ay\|_2^2$).  

\subsection{Matrix concentration sampling bound}
We restate a matrix concentration result for sampling rows of a matrix by leverage scores, which is by now standard. Note that the result holds when rows are sampled according to any set of probabilities that \emph{upper bound} the leverage scores. However, the number of rows sampled depends on the sum of probabilities used. 

\begin{claim}[Leverage score sampling -- e.g., Lemma 4 in the arXiv version of \cite{CohenLeeMusco:2015}]
	\label{claim:lev_sampling}
	Given a matrix $A\in C^{d\times n}$, let $\tilde{\tau}_1, \ldots, \tilde{\tau}_d \in [0,1]$ be a set of values such that, for all $j \in [d]$,
	\begin{align*}
		\tilde{\tau}_j \geq {\tau}_j(A).
	\end{align*}
	There exists fixed constants $c,C$ such that if we:
	\begin{enumerate}
		\item set $p_i = \min\left(1, \tilde{\tau}_jc\log(d/\delta)/\eps^2 \right)$ for all $j \in [d]$, and
		\item construct a matrix $B\in C^{s\times n}$ by selecting each row $a_j$ independently with probability $p_i$ and adding $\frac{1}{\sqrt{p_j}}a_j$ as a row to $B$ if it is selected
	\end{enumerate}
then with probability $(1-\delta)$:
	\begin{enumerate}
	\item $(1-\eps) B^*B \preceq A^*A \preceq (1+\eps)B^*B$, and
	\item $s \leq \frac{C\log(d/\delta)}{\eps^2} \cdot \sum_{j=1}^d \tilde{\tau_j}$.
	\end{enumerate}
Recall that $\preceq$ denotes the PSD Loewner ordering, which equivalently implies that for all $x\in \C^d$,
\begin{align*}
(1-\eps)\|Bx\|_2^2 \leq \|Ax\|_2^2 \leq (1+\eps)\|Bx\|_2^2.
\end{align*} 
Finally, note that we could equivalently state this claim as sampling a matrix $S \in \R^{s\times d}$ where $S$ has a row with entry $j$ equal to $1/\sqrt{p_j}$, and zeros elsewhere, for every $a_j$ which was sampled. Then we have $B = SA$. This is how the claim is used in Algorithm \ref{alg:sft}.
\end{claim}
While Claim  \ref{claim:lev_sampling}  is stated for real value matrices in \ref{claim:lev_sampling}, it holds for general complex $A$ as well: it follows Corollary 5.2 in \cite{Tropp:2012} which is proven for complex matrices.

\section{Additional Proofs: Ruler Based Methods}
\label{app:ruler}
\begin{replemma}{lem:general-ruler}
	For any $\alpha \in [1/2, 1]$, letting $R_\alpha$ be defined as in Def. \ref{def:general_ruler}, we have $|R_\alpha| \leq 2 d^\alpha$, and moreover: 
	\[
	\Delta (R_\alpha) \le 2d^{2 - 2\alpha} + d^{1 - \alpha} ( 1 + \log(\lceil d^{2 \alpha - 1}\rceil) ) \le 2d^{2 - 2 \alpha} + O(d^{1 - \alpha} \cdot \log d) \; .
	\]
\end{replemma}
\begin{proof}
	The bound on the size of $R_\alpha$ is immediate.
	It suffices to bound the coverage coefficient.
	We break up the set of distances into two intervals $A_1 = \{1, d - d^{\alpha} \}$ and $A_2 = \{d - d^{\alpha} , \ldots, d - 1\}$.
	For every $s \in A_1$, there exist at least $\lfloor d^{2 \alpha - 1} \rfloor$ elements $r_2 \in R^{(2)}_\alpha$ so that $s < r_2 \leq s + d^{\alpha}$.
	For each such element, there is at least one element $r_1 \in \R^{(1)}_\alpha$ so that $r_2 - r_1 = s$.
	Hence, for every $s \in A_1$, we have that $|(R_\alpha)_s| \ge \lfloor d^{2 \alpha - 1} \rfloor$.
	This yields 
	\[
	\sum_{s \in A_1} \frac{1}{|(R_\alpha)_s|} \leq \frac{|A_1|}{\lfloor d^{2 \alpha - 1} \rfloor} \leq \frac{d -d^\alpha}{d^{2\alpha-1}/2} \le 2d^{2 - 2\alpha} \; .
	\]
	We now turn our attention to $A_2$.
	We further subdivide $A_2$ into $A_2 = \bigcup_{j = 1}^{\lceil d^{2 \alpha - 1}\rceil} B_j$, where $B_j = \{ d - (j - 1) d^{1 - \alpha}-1, \ldots, \max(d - j d^{1 - \alpha},d-d^\alpha)  \}$.
	For each $j \in [1, \ldots, \lceil 2^{2 \alpha - 1} \rceil]$, there exist $j$ elements $r \in R^{(2)}_\alpha$ so that $r > d - (j - 1) d^{1 - \alpha}-1$.
	This implies that $|(R_\alpha)_s| \geq j$ for all $s \in B_j$.
	Therefore,
	\begin{align*}
	\sum_{s \in A_2} \frac{1}{|(R_\alpha)_s|} &= \sum_{j = 1}^{\lceil d^{2 \alpha - 1}\rceil} \sum_{s \in B_j} \frac{1}{|(R_\alpha)_s|} \leq \sum_{j = 1}^{\lceil d^{2 \alpha - 1} \rceil} \frac{|B_j|}{j} \\
	&\leq \sum_{j = 1}^{\lceil d^{2 \alpha - 1} \rceil} \frac{d^{1 - \alpha}}{j} \leq d^{1 - \alpha} ( 1 + \log(\lceil d^{2 \alpha - 1} \rceil) ) \; .
	\end{align*}
	Combining the bounds on $A_1$ and $A_2$ yield the desired claim.
\end{proof}

\begin{replemma}{lem:stableSubsample}
Let $\alpha \in [1/2, 1]$.
For any $k \le d$, any PSD Toeplitz matrix $T \in \R^{d \times d}$, and the sparse ruler 
$R_\alpha$ defined in Definition \ref{def:general_ruler},
	\begin{align*}
	\norm{T_{R_\alpha}}_2^2 \le \frac{32 k^2}{d^{2 - 2 \alpha}} \cdot \norm{T}_2^2 + 8 \cdot \min \left( \norm{T - T_k}_2^2, \frac{2}{d^{1 - \alpha}} \cdot \norm{T-T_k}_F^2 \right),
	\end{align*}
\end{replemma}
\noindent
where $\displaystyle T_k = \argmin_{\rank-k\ M} \norm{T-M}_F = \argmin_{\rank-k\ M} \norm{T-M}_2$. If $T$ is rank-$k$, $\norm{T-T_k}_F^2 = \norm{T - T_k}_2^2 = 0$.
\begin{proof}
	Recall $R_\alpha^{(1)} \eqdef \{1, \ldots, d^\alpha \}$ and $ R^{(2)}_\alpha = \{d, d - d^{1-\alpha}, d- 2 d^{1 - \alpha} , \ldots, d - (d^{\alpha}-1) d^{1 - \alpha} \}$. We can bound:
	\begin{align}\label{eq:easySpectral}
	\norm{T_{R_\alpha}}_2^2 \le 4 \max \left (\norm{T_{R_\alpha^{(1)}}}_2^2, \norm{T_{R_\alpha^{(2)}}}_2^2 \right )
	\end{align}
	since for any $x = [x_1,x_2]$, letting $y = [x_1,-x_2]$ we have 
	\begin{align*}
	x^T T_{R_\alpha} x \le x^T T_{R_\alpha} x + y^T T_{R_\alpha} y &= 2 x_1^T T_{R_\alpha^{(1)}} x_1 + 2 x_2^T T_{R_\alpha^{(2)}} x_2\\
	&\le 2 \left (\norm{T_{R_\alpha^{(1)}}}_2 \norm{x_1}_2^2 + \norm{T_{R_\alpha^{(2)}}}_2 \norm{x_2}_2^2 \right )\\
	&\le 2 \max(\norm{T_{R_\alpha^{(1)}}}_2,\norm{T_{R_\alpha^{(2)}}}_2) \cdot \norm{x}_2^2.
	\end{align*}
	Squaring both sides gives \eqref{eq:easySpectral}.
	So to prove the lemma it suffices to show that:
	\begin{align}\label{eq:camSuffice}
	\max \left (\norm{T_{R_\alpha^{(1)}}}_2^2, \norm{T_{R_\alpha^{(2)}}}_2^2 \right ) \le \frac{8 k^2}{{d^{2 - 2\alpha}}} \cdot \norm{T}_2^2 + 2 \cdot \min \left( \norm{T - T_k}_2^2, \frac{2}{d^{1 - \alpha}}\cdot\norm{T-T_k}_F^2 \right) \; .
	\end{align}
	We first show the bound for $T_{R_\alpha^{(1)}}$. The bound for $T_{R_\alpha^{(2)}}$ follows analogously.
	We first claim:
	\begin{align}
	\label{eq:rulerFrobBound}
	\norm{T_{R_\alpha^{(1)}}}_2^2\le \frac{8k^2}{d^{2 - 2 \alpha}} \cdot \norm{T}_2^2 +  \frac{4}{d^{1 - \alpha}} \cdot \norm{T-T_k}_F^2 \; .
	\end{align}
	There are $d^{1 - \alpha}$ disjoint principal submatrices of $T$ that are identical to $T_{R_\alpha^{(1)}}$. They correspond to the sets $R_j = \{j d^{\alpha} + 1,j d^\alpha +2,\ldots, (j+1)d^\alpha\}$ for $j \in \{ 0,\ldots, d^{1-\alpha} - 1\}$. 
	We can apply Lemma \ref{lem:offDiagWeight} to the low-rank approximation $T_k$ with partition $R_0 \cup \ldots \cup R_{d^{1-\alpha}-1}$ and $\eps = \frac{2k}{d^{1 - \alpha}}$. Letting $S$ be the set with $|S| \le \frac{k}{\eps} \le \frac{d^{1 - \alpha}}{2}$ whose existence is guaranteed by the lemma:
	\begin{align*}
	\sum_{\ell \in \{0,\ldots, d^{1-\alpha}-1\} \setminus S} \norm{(T_k)_{R_\ell}}_F^2 &\le \sum_{\ell \in \{0,\ldots, d^{1-\alpha}-1\} \setminus S} \eps \norm{{T_k}_{(R_\ell,[d])}}_F^2 \le \frac{2k}{d^{1 - \alpha}}  \norm{T_k}_F^2 \\
	&\le \frac{2k^2}{d^{1 - \alpha}}   \norm{T_k}_2^2 \le \frac{2k^2}{d^{1 - \alpha}}  \norm{T}_2^2 \; . \numberthis \label{eq:blockFirst1} 
	\end{align*}
	Since each $T_{R_j}$ is identical, and since $|\{0,\ldots, d^{1-\alpha}-1\} \setminus S| \ge \frac{k}{\eps} = \frac{d^{1 - \alpha}}{2}$
	we have:
	\begin{align}\label{eq:blockFirst}
	\norm{T_{R_\alpha^{(1)}}}_2^2 = \norm{T_{R_0}}_2^2 \le \norm{T_{R_0}}_F^2 &= \frac{1}{|\{0,\ldots, d^{1 - \alpha}-1\} \setminus S |} \cdot \sum_{\ell \in \{0,\ldots, d^{1-\alpha}-1\} \setminus S}  \norm{T_{R_\ell}}_F^2\nonumber\\
	&\le \frac{2}{d^{1 - \alpha}} \cdot \sum_{\ell \in \{0,\ldots, d^{1 - \alpha}-1\} \setminus S}  2 \left (\norm{(T_k)_{R_\ell}}_F^2 + \norm{(T - T_k)_{R_\ell}}_F^2 \right )\nonumber\\
	&\le \frac{4}{d^{1 - \alpha}} \left (\norm{T-T_k}_F^2 + \sum_{\ell \in \{0,\ldots d^{1 - \alpha}-1\} \setminus S}\norm{(T_k)_{R_\ell}}_F^2 \right )\nonumber\\
	&\le \frac{8k^2}{d^{2 - 2 \alpha}} \cdot \norm{T}_2^2 +  \frac{4}{d^{1 - \alpha}} \cdot \norm{T-T_k}_F^2 \; ,
	\end{align}
	where the last bound follows from \eqref{eq:blockFirst1}. This gives the bound of \eqref{eq:rulerFrobBound} for $T_{R_\alpha^{(1)}}$.
	We now show
	\begin{equation}
	\label{eq:rulerSpectralBound}
	\norm{T_{R_\alpha^{(1)}}}_2^2\le \frac{8k^2}{d^{2 - 2 \alpha}} \cdot \norm{T}_2^2 +  2 \norm{T-T_k}_2^2,
	\end{equation}
	which will complete \eqref{eq:camSuffice}.
	We proceed similarly to the proof of~\eqref{eq:rulerFrobBound}.
	We have that
	\begin{align}\label{eq:blockFirst2}
	\norm{T_{R_\alpha^{(1)}}}_2^2 = \norm{T_{R_0}}_2^2 &= \frac{1}{|\{0,\ldots, d^{1 - \alpha} -1\} \setminus S |} \cdot \sum_{\ell \in \{0,\ldots d^{1 - \alpha}-1\} \setminus S}  \norm{T_{R_\ell}}_2^2\nonumber\\
	&\le \frac{1}{|\{0,\ldots d^{1 - \alpha}-1\} \setminus S |} \cdot \sum_{\ell \in \{0,\ldots, d^{1 - \alpha}-1\} \setminus S}  2 \left (\norm{(T_k)_{R_\ell}}_2^2 + \norm{(T - T_k)_{R_\ell}}_2^2 \right )\nonumber\\
	&\le 2 \norm{T-T_k}_2^2 + \frac{2}{d^{1 - \alpha}} \sum_{\ell \in \{0,\ldots, d^{1 - \alpha}-1\} \setminus S}2\norm{(T_k)_{R_\ell}}_F^2\nonumber\\
	&\le \frac{8k^2}{d^{2 - 2 \alpha}} \cdot \norm{T}_2^2 +  2 \norm{T-T_k}_2^2,
	\end{align}
	where the last bound again follows from \eqref{eq:blockFirst1}. 
	Putting this together with~\eqref{eq:rulerFrobBound} yields~\eqref{eq:camSuffice} for $T_{R_\alpha^{(1)}}$.
	An identical argument shows the same bound for $T_{R_\alpha^{(2)}}$, completing the lemma.
\end{proof}

\subsection{A lower bound for estimation with sparse rulers: Proof of Theorem~\ref{thm:ruler-lower-bound}}
\label{sec:ruler-lb}
In this section we give a lower bound for Toeplitz covariance estimation via sparse ruler based measurements.
We show that the bound for rulers with sparsity $\Theta(\sqrt{d})$ given in Theorem \ref{thm:32} cannot be improved by more than a logarithmic factor.
We also give a more general tradeoff for more dense rulers, however, it appears to be loose.
Closing this gap is an interesting open direction.

Recall that any two distributions $F, G$, we let $\dkl (F, G) \eqdef \int \log \frac{dF}{dG} dF$ denote the KL divergence between the two distributions~\cite{cover2012elements}.
Our main tool for demonstrating this lower bound will be the following classical lemma:
\begin{lemma}[Assouad's lemma \cite{assouad1983deux}]\label{lem:asso}
Let $\mathscr{C} = \{D_z\}_{z \in \{1, -1 \}^r}$ be a family of $2^r$ probability distributions.
Let $\eps > 0$, and suppose there exist $\alpha = \alpha(\eps), \beta = \beta(\eps) > 0$ satisfying:
\begin{itemize}
\item For all $x, y \in \{1, -1\}^r$, we have $\dtv (D_x, D_y) \geq \alpha \Norm{x - y}_1$.
\item For all $x, y \in \{1, -1\}^r$ with Hamming distance $1$, we have $\dkl (D_x, D_y) \leq \beta$.
\end{itemize}
Then, no algorithm that takes samples $x^{(1)}, \ldots, x^{(n)}$ from an unknown $D_z \in \mathscr{C}$ and outputs $\hat D$ can satisfy 
\[
\dtv (D_z, \hat D) = o (\alpha r \exp (- O(\beta n)))
\] 
with probability $\geq 1/10$, for all $D_z \in \mathscr{C}$.
\end{lemma}

\subsubsection{Lower bound construction}
We now begin to describe our construction.
For any set $S \subseteq [d / 2]$, any $\eta \in (0, 1/(2d))$, and any $\sigma \in \{1, -1\}^{S}$, let $a_{S, \eta} (\sigma) \in \R^d$ be the vector given by 
\[
a_{S, \eta} (\sigma)_j = \left\{ \begin{array}{ll}
1 & \mbox{if $j = 1$} \\
\eta \sigma_{j} & \mbox{if $j \in S$} \\
0 & \mbox{otherwise}
\end{array} \right. \; .
\]
Let $T_{S, \eta} (\sigma) = \toep(a_{S, \eta} (\sigma)) \in \R^{d \times d}$.
By our restriction on $\eta$, we know that $\frac{1}{2} I \preceq T_{S, \eta} (\sigma) \preceq 2 I$ for all $\sigma$.
Therefore all of these matrices are positive semidefinite Toeplitz covariance matrices.
Let
\[
\mathscr{F}_{S, \eta} = \left\{\normal (0, T_{S, \eta} (\sigma)): \sigma \in \{1, +1\}^{S} \right\} \; .
\]
For simplicity of notation, in settings where $S$ and $\eta$ understood, we will often drop them from the subscripts.
Their meanings will be clear from context.
We first observe:
\begin{lemma}
For all $\sigma, \sigma' \in \{1, -1\}^S$ with $\sigma \neq \sigma'$, we have $\Norm{\toep(\sigma) - \toep(\sigma')}_2 \geq 2 \eta$.
\end{lemma}
\begin{proof}
Let $S' = \{j \in S: \sigma(j) \neq \sigma'(j) \}$.
Then we have:
\[
\Norm{\toep(\sigma) - \toep(\sigma')}_F^2 = \sum_{s \in S'} 8 (d - s) \eta^2 \geq 4 |S'| \eta^2 d \geq 4 \eta^2 d \; ,
\]
since $|S'| \geq 1$, from which we conclude that $\Norm{\toep(\sigma) - \toep(\sigma')}_2 \geq 2\eta$, as claimed.
\end{proof}

For any ruler $R \subseteq [d]$, we also let
\[
\mathscr{F}_{S, \eta, R} = \left\{\normal (0, (T_{S, \eta} (\sigma))_{R}): \sigma \in \{1, +1\}^{S} \right\} \; 
\]
denote the distribution of samples from any element of $\mathscr{F}_{S, \eta}$ when restricted to the ruler $R$.
As above, in settings where $S$ and $\eta$ are understood, we drop them from the subscripts and denote this set simply by $\mathscr{F}_R$.
Note that as with  the original covariance matrices $\toep(\sigma)$ we also have that $\frac{1}{2} I \preceq (\toep(\sigma)) _R \preceq 2 I$ for all $R$.

We now show that recovering a covariance in $\mathscr{F}$ to $\eps$ spectral norm error implies we must recover the covariance to good Frobenius norm error.
Specifically, for any $M \in \R^{d \times d}$, let $\gamma(M) \in \{1, -1\}^{S}$ be the vector given by $\gamma(M)_s = \mathrm{sign} \Paren{\sum_{|i - j| = s} M_{i,j}}$.
That is, for all diagonals $s \in S$, it simply takes the average of the $s$-th diagonal of $M$ and outputs the sign.
We then have:
\begin{lemma}
\label{lem:spectral-to-frob-ruler}
Let $\sigma \in \{1, -1\}^S$, and let $M \in \R^{d \times d}$ be so that $\Norm{M - \toep(\sigma)}_2 \leq \xi$.
Moreover, assume that for all $s \in S$, we have $|R_s| \leq B$.
Then 
\[
\Norm{\toep(\sigma)_R - \toep(\eta \cdot \gamma(M))_R}_F \leq \xi \sqrt{2 B} \; .
\]
Since all the matrices are well-conditioned, this implies that
\begin{align*}
\dtv (\normal (0, \toep(\sigma)_R), \normal (0, \toep(\eta \cdot \gamma(M))_R)) &= O \Paren{\Norm{I -  \toep(\sigma)_R^{-1/2} \toep(\eta \cdot \gamma(M))_R \toep(\sigma)_R^{-1/2}}_F}\\
&= O(\xi \sqrt{B}) \; .
\end{align*}

\end{lemma}
\begin{proof}
For all $s \in S$, let $\alpha_s = \frac{1}{2(d-s)}\sum_{|i - j| = s} M_{i,j}$ be the average of $M$'s entries along the $s$-th diagonal.
Note that for all $s \in S$, we have that 
\begin{align*}
\sum_{|i - j| = s} \Paren{\toep(\sigma)_{i,j} - M_{i,j}}^2 &= \sum_{|i - j| = s} \Paren{\eta \sigma_s - \alpha_s}^2 + \sum_{|i - j| = s} \Paren{M_{i,j} - \alpha_s}^2 \\
&\geq \sum_{|i - j| = s} \Paren{\eta \sigma_s - \alpha_s}^2\\
&\geq \frac{1}{4} \sum_{|i - j| = s} \Paren{\eta \sigma_s - \eta \gamma(M)_s}^2 \\
&=  \frac{\eta^2 \cdot 2(d - s)}{4} \cdot \Paren{\sigma_s - \gamma(M)_s}^2 \geq  \frac{\eta^2 \cdot d}{4} \Paren{\sigma_s - \gamma(M)_s}^2 \; .
\end{align*}
Let $S' = \{j \in S: \sigma_s \neq \gamma(M)_s \}$.
By the above calculation, we know that
\begin{align*}
\Norm{\toep(\sigma) - M}_F^2 \geq \eta^2 \cdot d |S'|\; , 
\end{align*}
from which we deduce
\[
\xi^2 \geq \Norm{\toep(\sigma) - M}_2^2 \geq \eta^2 \cdot |S'|\; ,
\]
or $|S'| \leq \xi^2 / \eta^2$.
This implies that
\begin{align*}
\Norm{\toep(\sigma)_R - \toep(\eta \cdot \gamma(M))_R}_F^2 &= \sum_{s \in S} \sum_{\substack{i, j \in R\\|i - j| = s}} \eta^2 (\sigma_s - \gamma(M)_s)^2 \\
&= \sum_{s \in S'} \sum_{\substack{i, j \in R\\|i - j| = s}} \eta^2 (\sigma_s - \gamma(M)_s)^2 \\
&\leq 2B \eta^2 |S'| \leq 2 B\xi^2 \; ,
\end{align*}
from which the claim follows.
\end{proof}
\noindent
In particular, Lemma~\ref{lem:spectral-to-frob-ruler} implies that to prove Theorem \ref{thm:ruler-lower-bound}, it suffices to demonstrate that the learning problem on distributions restricted to $R$, $\mathscr{F}_R$ is hard.
With these tools in hand, we are now ready to prove the full lower bound:
\begin{proof}[Proof of Theorem~\ref{thm:ruler-lower-bound}]
Suppose that there was an algorithm violating the theorem.
We will say that, on samples from $\normal (0, T)$, the algorithm \emph{succeeds} if it outputs $\tilde T$ so that $\| \tilde T - T\|_2 < \eps \norm{T}_2$.

By assumption, we have that $|R \times R| \leq d^{2 \alpha}$.
This implies that there exist $\geq 3 d/ 4$  elements $s \in [d]$ so that $|R_s| \leq 2 d^{2 \alpha - 1}$.
In particular, this means there is a set $S \subseteq [d / 2]$ of size at least $|S| \geq d / 4$ so that for all $s \in S$, we have $|R_s| \leq 2 d^{2 \alpha - 1}$.

Let $\eta > 0$ be a parameter to be fixed later, and consider the family $\mathscr{F} = \mathscr{F}_{S, \eta}$.
Since $\mathscr{A}$ succeeds with probability $\ge 1/10$ on any Toeplitz covariance $T$, it  clearly succeeds with at least this probability when the Toeplitz matrix is guaranteed to come from $\mathscr{F}$.

Since the matrices in $\mathscr{F}$ have spectral norm at most $2$, by Lemma~\ref{lem:spectral-to-frob-ruler}, this implies that there is an algorithm with the following guarantee: given $n$ samples $x^{(1)}, \ldots, x^{(n)} \sim \normal (0, \toep(\sigma)_R)$ for some $\toep(\sigma)_R \in \mathscr{F}_R$, the algorithm outputs $\widetilde M \in \R^{R \times R}$ so that with probability $\geq 1 / 10$, we have $\dtv (\normal (0, \toep(\sigma)_R), \normal (0, \widetilde M)) \leq c d^{\alpha - 1/2} \eps$, for some small constant $c$. 

We now show that this contradicts Assouad's lemma (Lem. \ref{lem:asso}) unless the number of samples taken satisfies $n = \Omega (d^{3 - 4 \alpha} / \eps^2)$. 
Indeed, to do so it suffices to compute the constants $\alpha, \beta$ in the statement of the lemma.
In general, for any $\sigma, \sigma' \in V$ so that $\sigma \neq \sigma'$, we have
\[
\Norm{\toep(\sigma)_R - \toep(\sigma')_R}_F^2 = \sum_{s: \sigma_s \neq \sigma'_s} \eta^2 \cdot |R_s| \; .
\]
Since our matrices are well-conditioned, it follows that for all $\sigma, \sigma' \in V$, we have
\begin{align*}
\dkl \Paren{\normal (0, \toep(\sigma)_R), \normal (0, \toep(\sigma')_R)} &= \Theta\Paren{\dtv \Paren{\normal (0, \toep(\sigma)_R), \normal (0, \toep(\sigma')_R)}^2}\\
&= \Theta \left (\Norm{\toep(\sigma)_R - \toep(\sigma')_R}_F^2 \right) \; .
\end{align*}
Combining these bounds, this implies that for all $\sigma, \sigma' \in V$, we have
\begin{align*}
\dtv(\normal (0, \toep(\sigma)_R), \normal (0, \toep(\sigma')_R)) &= \Omega \Paren{\sqrt{\sum_{s: \sigma_s \neq \sigma'_s} \eta^2 \cdot |R_s|}} \\
&= \Omega \Paren{\eta \Norm{\sigma - \sigma'}_1^{1/2}} = \Omega \Paren{\eta / \sqrt{d}} \cdot \Norm{\sigma - \sigma'}_1 \; ,
\end{align*}
and if $\sigma, \sigma'$ are neighboring, then we have
\begin{align*}
\dkl(\normal (0, \toep(\sigma)_R), \normal (0, \toep(\sigma')_R)) &= O \Paren{\sum_{s: \sigma_s \neq \sigma'_s} \eta^2 \cdot |R_s|} \\
&= O \Paren{\eta^2 d^{2\alpha - 1}} \; .
\end{align*}
Thus, by Assouad's lemma (Lemma \ref{lem:asso}) we know that no algorithm can can take $n$ samples $x^{(1)}, \ldots, x^{(n)} \sim \normal (0, \toep(\sigma)_R)$ for some $\toep(\sigma)_R \in \mathscr{F}_R$, and output $\widetilde M \in \R^{R \times R}$ so that with probability $\geq 1 / 10$, 
\begin{align*}
\dtv (\normal(0, \widetilde M),\normal(0, \toep(\sigma)_R)) = o \Paren{\eta \sqrt{d} \exp (- O(-d^{2\alpha - 1} \eta^2 n))}.
\end{align*}
%
Setting $\eta = O(\eps d^{\alpha - 1})$ gives that no algorithm can achieve
\begin{align*}
\dtv (\normal(0, \widetilde M),\normal(0, \toep(\sigma)_R)) = o \Paren{\eps d^{\alpha-1/2} \exp (- O(-\frac{n \eps^2}{d^{3-4\alpha}}))}.
\end{align*}

However, this contraticts our assumption, which gives an algorithm achieving error $c  d^{\alpha - 1/2} \eps$ with probability $\ge  1/10$, unless
$n = \Omega (d^{3 - 4\alpha} / \eps^2)$. This completes the proof.
%

\end{proof}

\section{Additional Proofs: Fourier Methods}\label{app:additional}

\begin{proof}[Proof of Claim \ref{thm:diagonal}]
	For diagonal $T$, let $E$ denote the empirical covariance $E = \frac{1}{n}\sum_{j=1}^n (x^{(j)})(x^{(j)})^T$. The $\ell^\text{th}$ diagonal of each sample outerproduct $(x^{(j)})(x^{(j)})^T_{\ell,\ell}$ is the square of a Gaussian random variable with mean $0$ and variance $T_{\ell,\ell}$ -- i.e., a chi-squared random variable with mean $T_{\ell,\ell}$. $E_{\ell,\ell}$ is the average of $n$ chi-squared random variables. So by a standard chi-squared tail bound \cite{Wainwright:2019},
	\begin{align*}
	\Pr[|T_{\ell,\ell} - E_{\ell,\ell}| \geq \eps T_{\ell,\ell}] \leq \delta/d
	\end{align*}
	as long as $n \geq c \log(d/\delta)/\eps^2$ for a fixed constant $c$. From a union bound, it follows that, for all $\ell$, $E_{\ell,\ell} \in (1\pm\eps)T_{\ell,\ell}$  with probability $(1-\delta)$, and thus, $$\|T- \diag(E)\|_2 = \max_\ell |T_{\ell,\ell} - \diag(E_{\ell,\ell})| \leq \eps \max_\ell T_{\ell,\ell} = \eps \|T\|_2.$$
	
	The argument for circulant matrices is essentially the same, but requires a transformation. In particular, since any positive semidefinite circulant matrix $T$ can be written as $T = FDF^*$ where $F$ is the unitary discrete Fourier transform matrix and $D$ is a positive diagonal matrix, we see that:
	\begin{align*}
		\E \left[F^*(x^{(j)})(x^{(j)})^TF \right] = F^*FDF^*F = D
	\end{align*}
	So it is natural to approximate $D$ by the empirical average $A = \diag\left(\frac{1}{n}\sum_{j=1}^n F^*(x^{(j)})(x^{(j)})^TF\right)$. We just need to bound how well $A_{\ell,\ell}$ concentrates around $D_{\ell,\ell}$. To do so, we claim that for any sample $x$, $\left[F^*xx^T F\right]_{\ell,\ell}$ is still distributed as a chi-squared random variable. In particular,
	\begin{align*}
	\left[F^*xx^T F\right]_{\ell,\ell} &= \left(\frac{1}{\sqrt{d}}\sum_{w=1}^d x_we^{-2\pi i (w-1)(\ell-1)} \right)^*\left(\frac{1}{\sqrt{d}}\sum_{w=1}^d x_we^{-2\pi i (w-1)(\ell-1)} \right) \\
	&= \left(\frac{1}{\sqrt{d}}\sum_{w=1}^d x_w|e^{-2\pi i (w-1)(\ell-1)}| \right)\left(\frac{1}{\sqrt{d}}\sum_{w=1}^d x_w |e^{-2\pi i (w-1)(\ell-1)}| \right) 
	\end{align*}
	So $\left[F^*xx^T F\right]_{\ell,\ell}$ can be written as $g^2$ where $g$ is a weighted sum of entries in $x_w$. Since $x_w$ is a normal random vector, it follows that $g$ is normal (even though $x_w$ has correlated entries). Accordingly, $\left[F^*xx^T F\right]_{\ell,\ell}$ is a chi-squared random variable and, as  we already argued, has mean $T_{\ell,\ell}$. 
	
	Again from chi-squared concentration and a union bound, we have that $\|\diag(A) - D\|_2 \leq \eps \|D\|_2$ with probability $(1-\delta)$ as long as $n \geq c \log(d/\delta)/\eps^2$. Since $F$ is unitary, this implies that $\|F\diag(A)F^* - FDF^*\|_2 \leq \eps \|FDF^*\|_2$, which proves the claim.
\end{proof}

\begin{replemma}{lem:fss2}[Frequency-based low-rank approximation]
	For any PSD Toeplitz matrix $T \in \R^{d \times d}$, rank $k$, and $m \ge ck$ for some fixed constant $c$, there exists $M = \{f_1, \ldots, f_m\} \subset [0,1]$ such that, letting $F_M \in \C^{d \times m}$ be the Fourier matrix  with frequencies $M$ (Def. \ref{def:fourier_matrix}) and $Z = F_M^+ T^{1/2}$, we have 1): $\norm{F_M^+ }_2^2 \le \frac{2}{\beta}$ and 2):
	\begin{align}
	\norm{F_M Z - T^{1/2}}_F^2  &\le 3 \norm{T^{1/2} - T^{1/2}_k}_F^2 + 6\beta\norm{T}_2 \label{fss:frob}\\
	&\text{ and }\nonumber\\
	\norm{F_M Z - T^{1/2}}_2^2  &\le 3\norm{T^{1/2} - T^{1/2}_k}_2^2 + \frac{3}{k} \norm{T^{1/2} - T^{1/2}_k}_F^2 + 6 \beta \norm{T}_2.\label{fss:spectral}
	\end{align}
\end{replemma}
\begin{proof}
	Let $\bar T = T + \frac{\beta \norm{T}_2}{d} \cdot I$ and $\bar T = F_S D F_S^*$ be  the Vandermonde decomposition of $\bar T$ of Lemma \ref{lem:fourier} with $S = \{f_1,\ldots f_d\}$. We can bound the entries of $D$ since, letting $(F_S)_j$ denote the $j^{th}$ column of $F_S$:
	\begin{align*}
	\norm{\bar T}_2 \ge \norm{D_{j,j}\cdot (F_S)_j (F_S)_j^*}_2 = D_{j,j} \cdot d
	\end{align*}
	so $D_{j,j} \le \norm{\bar T}_2/d$. Note that $\bar T$ is full rank and thus $F_S$ is full-rank. For any $M \subseteq S$ we have $\displaystyle \sigma_{\min}(F_M) = \min_{x: \norm{x}_2 = 1} \norm{F_M x}_2^2 \ge \min_{x: \norm{x}_2 = 1} \norm{F_S x}_2^2 = \sigma_{\min}(F_S).$ Thus,
	\begin{align*}
	\norm{F_M^{+}}_2^2 \le \norm{F_S^{+}}_2^2 \le \norm{(F_SD^{1/2})^+}_2 \cdot \norm{D^{1/2}}_2^2 \le  \frac{\norm{\bar T^{-1/2}}_2 \cdot \norm{\bar T^{1/2}}^2_2}{d} \le\frac{1+\beta/d}{\beta} \le \frac{2}{\beta}.
	\end{align*}
	We now proceed to prove \eqref{fss:frob} and \eqref{fss:spectral}. Let $\bar Z = F_M^+ \bar T^{1/2}$.
	For any  $M$:
	\begin{align}\label{eq:beta1}
	\norm{F_M Z - T^{1/2}}_F \le \norm{F_M \bar Z - T^{1/2}}_F &\le \norm{F_M \bar Z - \bar T^{1/2}}_F + \norm{T^{1/2}-\bar T^{1/2}}_F\nonumber\\
	&\le \norm{F_M \bar Z - \bar T^{1/2}}_F + \sqrt{\beta}\norm{T^{1/2}}_2.
	\end{align}
	Similarly, 
	\begin{align}\label{eq:beta2}\norm{\bar T^{1/2} - \bar T^{1/2}_k}_F \le \norm{\bar T^{1/2} - T^{1/2}_k}_F &\le \norm{T^{1/2}-
		T^{1/2}_k}_F + \norm{T^{1/2}-
		\bar T^{1/2}}_F\nonumber\\
	&\le \norm{T^{1/2} - T^{1/2}_k}_F + \sqrt{\beta} \norm{T^{1/2}}_2.
	\end{align}
	To prove \eqref{fss:frob} we will  show that there exists $M$ with:
	\begin{align}
	\norm{F_M \bar Z - \bar T^{1/2}}_F^2  &\le \frac{3}{2} \norm{\bar T^{1/2} - \bar T^{1/2}_k}_F^2\label{fss:frob2}.
	\end{align}
	Combined with \eqref{eq:beta1} and \eqref{eq:beta2}, \eqref{fss:frob2} gives $\norm{F_M Z - T^{1/2}}_F \le \sqrt{\frac{3}{2} \norm{T^{1/2} - T^{1/2}_k}_F} +3\sqrt{\beta} \norm{T^{1/2}}_2$. \eqref{fss:frob} then follows from the AMGM inequality. Following a similar argument, we note that for any $M$:
	\begin{align}\label{eq:beta3}
	\norm{F_MZ - T^{1/2}}_2 \le \norm{F_M\bar Z - T^{1/2}}_2 \le \norm{F_M\bar Z - \bar T^{1/2}}_2 + \sqrt{\frac{\beta}{d}} \norm{T^{1/2}}_2
	\end{align}
	and 
	\begin{align}\label{eq:beta4}
	\norm{\bar T - \bar T_k^{1/2}}_2 \le \norm{\bar T - T_k^{1/2}}_2 \le \norm{T - T_k^{1/2}}_2 +  \sqrt{\frac{\beta}{d}} \norm{T^{1/2}}_2.
	\end{align}
	To prove \eqref{fss:spectral} we thus must show that there exists $M$ with
	\begin{align}
	\norm{F_M \bar Z - \bar T^{1/2}}_2^2  &\le \frac{3}{2} \norm{\bar T^{1/2}-\bar T^{1/2}_k}_2 + \frac{1}{2k} \norm{\bar  T^{1/2}-\bar T^{1/2}_k}_F^2. \label{fss:spectral2}
	\end{align}
	\eqref{fss:spectral} follows from combining \eqref{fss:spectral2} with \eqref{eq:beta2}, \eqref{eq:beta3} and \eqref{eq:beta4} and applying AMGM.
	
	We proceed to prove \eqref{fss:frob2} and \eqref{fss:spectral2}, which in combination complete the lemma.
	Let  $U_1 = \bar T^{-1/2} F_S D^{1/2}$. We have $U_1 U_1^* = \bar  T^{-1/2} F_S D F_S^* \bar T^{-1/2} = I$ so $U_1$ is unitary. Additionally, for any $M$, 
	$$\bar Z U_1 = F_M^+ \bar T^{1/2} \bar T^{-1/2} F_S D^{1/2} = F_M^+(F_S D^{1/2}).$$
	So to show the lemma, it suffices to prove \eqref{fss:frob2} and \eqref{fss:spectral2} for lefthand side $F_SD^{1/2}- F_M (F_M^+(F_S D^{1/2}))$, which is just the original the lefthand side with the rotation $U_1$ applied. If $F$ were real, the existence of column submatrix $F_M$  with $m \le ck$ for some fixed $c$ satisfying \eqref{fss:frob2} is shown e.g. in  Theorem 1.1 of \cite{guruswami2012optimal}. The existence of $F_M$ with $m \le ck$ satisfying \eqref{fss:spectral2} is a corollary  of Theorem 27 of \cite{cohen2015dimensionality}. Simply taking the union of these two column subsets yields a subset satisfying both guarantees simultaneously since projection to a larger set can only reduce spectral and Frobenius norm respectively. In our setting $F_S$ is complex, however, the proofs of the above results can be seen to still hold.
	Alternatively, we can prove the lemma directly from the existing results for real matrices by noticing that, by Lemma \ref{lem:fourier}, the frequencies of $F_S$ come in conjugate pairs with equal corresponding entries in $D$. Thus $F_S$ can be rotated via a $2 \times 2$ block diagonal matrix to a real matrix $F'$ to which these results can be applied.

Specifically, 
as noted, it suffices  to prove \eqref{fss:frob2} and \eqref{fss:spectral2} for $F_SD^{1/2}-F_M (F_M^+(F_S D^{1/2})) $. Since by Lemma \ref{lem:fourier} the frequencies of $F_S$ come in conjugate pairs with equal corresponding entries in $D$, there exists unitary $U_2$ such that $F_S D^{1/2} U_2 = F_S U_2 D^{1/2} =  F' D^{1/2}$ where $F'$ is real. In particular, if $f_j$ is conjugate to itself ($e^{2 \pi i f_j} = e^{-2 \pi i f_j}$) then the $j^{th}$ column of $U_2$ is just the $j^{th}$ standard basis vector. For $j \neq j'$ that are conjugate to each other, on the principal submatrix corresponding to $j,j'$, $U_2$ equals: $$\begin{bmatrix} 
\frac{1}{\sqrt{2}} & \frac{-i}{\sqrt{2}}\\
 \frac{1}{\sqrt{2}} & \frac{i}{\sqrt{2}}\
\end{bmatrix}$$ 

Since $F'$ is real, prior work \cite{guruswami2012optimal,cohen2015dimensionality} implies that, for some universal constant $c$, there exists a subset $M'$ of $m' \le ck$ columns of $F'$ (which we denote $F_{M'}'$) such that $F_{M'}' ({F_{M'}'}^+(F' D^{1/2})) - F'D^{1/2}$ satisfies \eqref{fss:frob2} and \eqref{fss:spectral2}. We can then write $F'_{M'}$ in the span of at most $2ck$ columns of $F_S$ using the transformation $U_2$.  This proves the existence of $F_M$  such that $F_M (F_M^+(F_S D^{1/2})) - F_SD^{1/2}$ satisfies \eqref{fss:frob2} and \eqref{fss:spectral2}, completing the lemma.
\end{proof}

\begin{replemma}{lem:existenceT}
	Consider PSD $T \in \R^{d\times d}$ and $X \in \R^{d \times n}$ with columns drawn i.i.d. from $\frac{1}{\sqrt{n}} \cdot \mathcal{N}(0,T)$. For any rank $k$, $\eps,\delta \in (0,1]$, $m \ge c_1 k$, and $n \ge c_2 \left (m + \log(1/\delta) \right )$ for sufficiently large $c_1,c_2$, with probability $\ge 1- \delta$,  there exists $M = \{f_1,\ldots,f_m\} \subset [0,1]$ such that, letting $Z = F_M^+ X$, $\norm{Z}_2^2 \le \frac{c_3 \cdot d^2 \norm{T}_2}{\eps^2}$ for some fixed $c_3$ and:
	\begin{align*}
	\norm{F_M Z Z^* F_M^* -XX^*}_F &\le 10 \sqrt{\norm{T-T_k}_2 \cdot \tr(T) + \frac{\tr(T-T_k) \cdot \tr(T)}{k}} + \frac{\eps}{2} \norm{T}_2.
	\end{align*}
\end{replemma}
\begin{proof}
	Applying Lemmas \ref{lem:fss2} and \ref{lem:pcp} (with error parameters $\beta = \frac{\eps^2}{d \cdot 12 \cdot 32}$ and $\gamma = 1/2$ respectively), for sufficiently large $c_1,c_2$ and $m \ge c_1 k$, $n \ge c_2( m + \log(1/\delta))$, with probability $\ge 1-\delta/2$ there exists $M = \{f_1,\ldots,f_m\} \subset [0,1]$ such that, letting $Z = F_M^+ X$,
	\begin{align*}
	\norm{F_M Z - X}_2^2 &\le 1.5 \norm{F_M F_M^+ T^{1/2} - T^{1/2}}_2^2 + \frac{1}{2k} \norm{F_M F_M^+ T^{1/2} - T^{1/2}}_F^2 \\
	&\le 4.5 \norm{T^{1/2}-T_k^{1/2}}_2^2 + \frac{4.5}{k} \norm{T^{1/2}-T_k^{1/2}}_F^2 + 9 \beta \norm{T}_2 +  \frac{1.5}{k} \norm{T^{1/2}-T_k^{1/2}}_F^2 + \frac{3 \beta}{k} \norm{T}_2\\
	&\le 4.5\norm{T^{1/2}-T_k^{1/2}}_2^2 + \frac{6}{k} \norm{T^{1/2}-T_k^{1/2}}_F^2 +\frac{\eps^2}{32d} \norm{T}_2
	\end{align*}
	Using that for any $A,B$, $\norm{AA^* - BB^*}_F \le \norm{AB^* - BB^*}_F + \norm{AA^* - AB^*}_F \le \norm{A-B}_2 \cdot (\norm{A}_F + \norm{B}_F)$ we can bound:
	\begin{align}\label{eq:aabbPlug22}
	\norm{F_M Z Z^* F_M^* -XX^*}_F &\le \norm{F_{M} Z -X}_2 \cdot \left (\norm{X}_F + \norm{F_{M} Z}_F \right )\nonumber\\
	&\le \sqrt{4.5 \norm{T^{1/2}-T_k^{1/2}}_2^2 + \frac{6}{k} \norm{T^{1/2}-T_k^{1/2}}_F^2 +  \frac{\eps}{32d} \norm{T}_2}  \cdot \left (\norm{X}_F + \norm{F_{M} Z}_F \right ).
	\end{align}
	We have $\norm{X}_F = \norm{T^{1/2} G}_F$ where $G$ has columns distributed as $\frac{1}{\sqrt{n}} \cdot \mathcal{N}(0,I)$. By standard Gaussian concentration, for $n \ge c \log(1/\delta)$ for large enough $c$, $\norm{X}_F \le 2 \norm{T^{1/2}}_F$ with probability $\ge 1-\delta/2$. Conditioning on this event we also have $\norm{F_M Z}_F = \norm{F_MF_M^+ Z}_F \le \norm{X}_F \le 2\norm{T^{1/2}}_F$ since this is just a projection of $X$ onto the column space of $F_M$. Plugging back into \eqref{eq:aabbPlug22}:
	\begin{align*}
	&\norm{F_M Z Z^* F_M^* -XX^*}_F \\&\le 4\sqrt{4.5 \norm{T^{1/2}-T_k^{1/2}}_2^2 \cdot \norm{T^{1/2}}_F^2 + \frac{6}{k} \norm{T^{1/2}-T_k^{1/2}}_F^2 \cdot \norm{T^{1/2}}_F^2 + \frac{\eps}{32d} \norm{T}_2 \cdot \norm{T^{1/2}}_F^2}.
	\end{align*}
	Applying our stable rank bound, we have $\norm{T^{1/2}}_F^2 = \tr(T) \le s \norm{T}_2$. Further, $\norm{T^{1/2}-T_k^{1/2}}_2^2 \le \frac{\norm{T^{1/2}}_F^2}{k} \le \frac{s}{k} \norm{T}_2$.
	Finally, we can bound the above by:
	\begin{align*}
	\norm{F_M Z Z^* F_M^* -XX^*}_F &\le 10 \sqrt{\norm{T-T_k}_2 \cdot \tr(T) + \frac{\tr(T-T_k) \cdot \tr(T)}{k}} + \frac{\eps}{2} \norm{T}_2,
	\end{align*}
	which gives the error bound.  It just remains to bound $\norm{Z}_2^2 \le \norm{F_M^+}_2^2\cdot  \norm{X}_F^2 \le \frac{2}{\beta} \cdot 4 \norm{T^{1/2}}_F^2 \le \frac{8d \norm{T}_2}{\beta} \le \frac{(8 \cdot 12 \cdot 32) d^2 \norm{T}_2}{\eps^2}$.
	The total success probability, union bounding over the event that a suitable $M$ exists and that $\norm{F_MZ}_F \le \norm{X}_F \le 2 \norm{T^{1/2}}_F$ is at least $1-\delta$. 
\end{proof}

\begin{replemma}{lem:mnetT}
	Consider PSD $T \in \R^{d\times d}$ and $X \in \R^{d \times n}$ with columns drawn i.i.d. from $\frac{1}{\sqrt{n}} \cdot \mathcal{N}(0,T)$.
	For any rank $k$ and $\delta,\eps \in (0,1]$, consider $m \ge c_1 k$, $n \ge c_2 \left (m + \log(1/\delta) \right )$,
	and $N = \{0,\alpha,2\alpha,\ldots 1\}$ for $\alpha = \frac{\eps^2}{c_3 d^{3.5}}$
	for sufficiently large constants $c_1,c_2,c_3$.  With probability $\ge 1- \delta$, there exists $M = \{f_1,\ldots,f_m\} \subset N$ with:
	\begin{align*}
	\min_{W \in \C^{m \times m}} \norm{F_M W F_M^* - XX^T}_F \le 10 \sqrt{\norm{T-T_k}_2 \cdot \tr(T) + \frac{\tr(T-T_k) \cdot \tr(T)}{k}} + \eps \norm{T}_2.
	\end{align*}
\end{replemma}
\begin{proof}
	If suffices to show that, for the subset  $M = \{f_1,\ldots, f_m\} \subset [0,1]$ shown to exist  (with probability $\ge 1-\delta$) in Lemma \ref{lem:existenceT}, letting $\hat M = \{\hat f_1,\ldots, \hat f_m\} \subset N$ be the set obtained by  rounding each element of $M$ to the nearest element in $N$, and letting $Z = F_M^+ X$, 
	\begin{align}\label{eq:round}
	\norm{F_{\hat M} ZZ^* F_{\hat M}^* - F_M ZZ^* F_M^*}_F \le \frac{\eps}{2} \norm{T}_2.
	\end{align}
	$M$ has $\norm{F_M ZZ^* F_M^* - XX^T}_F \le  10 \sqrt{\norm{T-T_k}_2 \cdot \tr(T) + \frac{\tr(T-T_k) \cdot \tr(T)}{k}} + \frac{\eps}{2} \norm{T}_2$. Thus by triangle inequality, we will have:
	\begin{align*}
	\min_{W \in \C^{m \times m}} \norm{F_{\hat M} W F_{\hat M}^* - XX^T}_F &\le \norm{F_{\hat M} ZZ^* F_{\hat M}^* - XX^T}_F\\
	&\le 10 \sqrt{\norm{T-T_k}_2 \cdot \tr(T) + \frac{\tr(T-T_k) \cdot \tr(T)}{k}} + \eps \norm{T}_2.
	\end{align*}
	We thus turn to proving \eqref{eq:round}. We can bound:
	\begin{align*}
	\norm{F_{\hat M} ZZ^* F_{\hat M}^* - F_M ZZ^* F_M^*}_F &\le \norm{F_{\hat M}ZZ^* F_M^* - F_{M}ZZ^* F_M^*}_2 + \norm{F_{\hat M}ZZ^* F_{\hat M}^* - F_{\hat M}ZZ^* F_M^*}_2 \\
	& \le \norm{F_{\hat M} - F_M}_2 \cdot \norm{Z}_2 \cdot \left (\norm{F_M Z}_2+\norm{F_{\hat M}Z}_2 \right )\\
	&\le \norm{F_{\hat M} - F_M}_2 \cdot \norm{Z}_2 \cdot \left (2\norm{F_M Z}_2+\norm{F_{\hat M }-F_{M }}_2\norm{Z}_2 \right )
	\end{align*}
	By Lemma \ref{lem:existenceT}, $\norm{Z}_2^2 \le \frac{c \cdot d^2 \norm{T}_2}{\eps^2}$ for some constant $c$. Also, $\norm{F_MZ}_2 \le \norm{F_MF_M^+ X}_F \le \norm{X}_F \le 2\norm{T}_F \le 2\sqrt{d} \norm{T}_2$ with probability $\ge 1-\delta$ as argued in the proof of Lemma \ref{lem:existenceT}. Thus: 
	\begin{align*}
	\norm{F_{\hat M} ZZ^* F_{\hat M}^* - F_M ZZ^* F_M^*}_F &\le \norm{F_{\hat M} - F_M}_2 \cdot \frac{2\sqrt{c}d^{1.5}  \norm{T}_2}{\eps} +\norm{F_{\hat M} - F_M}_2^2 \cdot \frac{cd^2 \norm{T}_2}{\eps^2}.
	\end{align*}
	To prove the lemma it thus suffices to show that $\norm{F_{\hat M} - F_M}_2^2 \le \norm{F_{\hat M} - F_M}_F^2 \le \frac{\eps^4}{32cd^3}$.
	We have 
	\begin{align*}
	\left | (F_{\hat M} - F_M)_{kj} \right | = \left | e^{-2\pi i f_j (k-1)} - e^{-2\pi i \hat f_j (k-1)} \right | \le \int_{2 \pi f_j(k-1)}^{2 \pi \hat f_j (k-1)} \left |e^{-2\pi i x (k-1)} \right | dx \le 2\pi (k-1) \cdot \alpha.
	\end{align*}
	Since $(k-1) < d$, for $\alpha = \frac{1}{2 \pi d} \cdot \frac{\eps^2}{6 \sqrt{c} d^{2.5}}$, $ \left | (F_{\hat M} - F_M)_{kj} \right |^2 \le \frac{\eps^4}{36 c d^5}$ and so $\norm{F_{\hat M} - F_M}_F^2 \le \frac{\eps^4}{36 c d^3}$, giving the required bound.
\end{proof}

\subsection{A priori leverage bounds for Fourier matrices}
In Section \ref{sec:fourier}, we apply Claim \ref{claim:lev_sampling} to the sample matrix $X \in \R^{d\times n}$, which has columns drawn according to $\mathcal{N}(0,T)$. Since we do not want to examine all $d$ entries in each sample, we cannot afford to explicitly compute $X$'s leverage scores. Instead, we use that $X$ can be written as $F_SY$ (exactly or approximately) for some matrix $Y$, where $F_S \in \C^{d\times s}$ is a Fourier matrix with few columns. Accordingly, it suffices to bound the leverage scores of $F_S$. 

Surprisingly this can be done without any computation: the leverage scores of a Fourier matrix $F_S$ with \emph{any} frequencies $S =\{f_1,\ldots f_s\}$ (see Definition \ref{def:fourier_matrix}) can be upper bounded by a simple closed-form formula \emph{which does not depend on $S$}. I.e. the upper bound is valid no matter how large, small, clustered, or separated the frequencies are.  Moreover, this closed form upper bound is nearly tight (up to a logarithmic factor) for any $S$. In particular, in our setting where $F_s$ is full rank and $s \leq d$, the true leverage scores satisfy $\sum_{j = 1}^d \tau_j(F_S) = s$ while the closed form produces leverage score upper bounds which sum to just $O(s \log s)$. 

Formally, we prove two results which can be combined to obtain a nearly tight upper bound:

\begin{lemma}[Leverage score upper bounds]\label{lem:main_lev_bound}
	For any Fourier matrix $F_S \in \C^{d\times s}$ with $s\leq d$,
	\begin{align}
	\label{eq:simple_lev_bound}
	\tau_j(F_s) &\leq \frac{s}{\min(j,d+1-j)} & &\text{for all $j \in [d]$}
	\end{align} 
	and, for some fixed constant $c$,
	\begin{align}
	\label{eq:uniform_lev_bound}
	\tau_j(F_s) &\leq \frac{cs^6\log^3(s+1)}{d} &  &\text{for all $j \in [d]$}.
	\end{align} 
\end{lemma}
For continuous Fourier operators, a bound similar to  \eqref{eq:simple_lev_bound} was proven in  \cite{ChenPrice:2018} and improved in \cite{AvronKapralovMusco:2019} using a different technique, which we adapt. We prove \eqref{eq:uniform_lev_bound} using techniques from \cite{ChenKanePrice:2016} which establishes a similar bound for continuous operators. These results were very recently improved by an $s$ factor in \cite{ChenPrice:2019}, so it might also be possible to improve our bound as well. However, polynomial improvements in \eqref{eq:uniform_lev_bound} will only lead to constant factor improvements in our final application of Lemma \ref{lem:main_lev_bound}. 

For \emph{discrete} (i.e. on-grid) Fourier matrices with $s$ frequencies, the leverage score of every row is simply equal to $\frac{s}{d}$ since all columns are orthogonal and all rows have squared norm $s/d$. This bound gives an approach to establishing sparse recovery results for subsampled discrete Fourier matrices (e.g., via the restricted isometry property). 

Beginning with the results of \cite{ChenKanePrice:2016}, bounds like Lemma \ref{lem:main_lev_bound}, on the other hand, can be used to establish similar results for off-grid Fourier matrices. In contrast to earlier work on ``off-grid'' sparse recovery problems (e.g., \cite{TangBhaskarShah:2013,CandesFernandez-Granda:2014,TangBhaskarRecht:2015,BoufounosCevherGilbert:2015}), results based on  such bounds require no assumptions on the frequencies in $S$, including no ``separation'' assumption that $|f_i - f_j|$ is not too small for all $i,j$. We refer the reader to \cite{ChenKanePrice:2016} and \cite{PriceSong:2015} for a more in-depth discussion of this issue. In the context of our work, Lemma \ref{lem:main_lev_bound} allows us to avoid a separation assumption on the frequencies in $T$'s Vandermonde decomposition, which in necessary in some related work \cite{ChenChiGoldsmith:2015}.

Before proving Lemma \ref{lem:main_lev_bound} we state a simple corollary:
\begin{corollary}\label{cor:sum_bound}
	For any positive integers $d$ and $s \leq d$, there is an explicit set of values $\tilde{\tau}_1^{(s)}, \ldots, \tilde{\tau}_d^{(s)} \in (0,1]$ such that, for any Fourier matrix $F_S \in \C^{d\times s}$,
	\begin{align*}
	\forall j,\,\, &\tilde{\tau}_j^{(s)}(F_s) \geq \tau_j & &\text{and} & \sum_{j=1}^d \tilde{\tau}_j^{(s)} &= O(s \log s).
	\end{align*}
\end{corollary}
\begin{proof}
	Letting $c$ be the constant from Lemma \ref{lem:main_lev_bound}, we can satisfy the conditions of Corollary \ref{cor:sum_bound} by choosing:  
	\begin{align*}
	\tilde{\tau}_j^{(s)} = \min\left(1, \frac{s}{\min(j,d+1-j)}, \frac{cs^6\log^3(s+1)}{d}\right).
	\end{align*}
	By Lemma \ref{lem:main_lev_bound} and the fact that $\tau_j(F_s) \leq 1$ for all $s$ (by Fact \ref{fact:min_char}) we immediately have $\tilde{\tau}_j^{(s)} \geq \tau_j(F_s)$. So we are just left to bound $\sum_{j=1}^d \tilde{\tau}_j^{(s)}$. Let $q = \frac{cs^6\log^3(s+1)}{d}$ denote the right hand side of \eqref{eq:uniform_lev_bound}.
	\begin{align*}
	\sum_{j=1}^d \tilde{\tau}_j^{(s)}  &= \sum_{j=1}^{\lfloor 1/q \rfloor} \tilde{\tau}_j^{(s)} + \sum_{j=\lfloor 1/q \rfloor + 1}^{\lceil d - 1/q \rceil - 1} \tilde{\tau}_j^{(s)} + \sum_{j=\lceil d - 1/q \rceil}^{d} \tilde{\tau}_j^{(s)} \\ 
	&\leq 1 + \sum_{j=\lfloor 1/q \rfloor + 1}^{\lceil d - 1/q \rceil - 1}  \frac{s}{\min(j,d+1-j)} + 1 \\
	&\leq 2 + 2s\cdot \sum_{j=\lfloor 1/q \rfloor + 1}^{\lceil d/2 \rceil} \frac{1}{j} \\
	&\leq 2 + 2s\cdot \left(\log(\lceil d/2 \rceil) + 1 - \log(\lfloor 1/q \rfloor )\right) \\
	&= 2 + 2s\cdot\log\left(O(dq)\right)  = O\left(s\log s\right)
	\end{align*}
	The last inequality follows from the fact that for any integer $z$, $\log z \leq \sum_{j=1}^z \frac{1}{j} \leq 1+ \log z $.
\end{proof}
All our applications of leverage score bounds go through Corollary \ref{cor:sum_bound}: we sample rows from $X$ randomly according to the closed form probabilities $\tilde{\tau}_1, \ldots, \tilde{\tau}_d$. The upper bound on the sum of these probabilities allows us to bound how many samples need to be taken in expectation. 

Next, we need to prove the upper bounds of \eqref{eq:simple_lev_bound} and \eqref{eq:uniform_lev_bound}, which are handled separately. Both rely on the following claim, which allow us to restrict our attention to $\tau_{j}(F_S)$ for $j \leq (d+1)/2$: 
\begin{claim}[Leverage score symmetry]\label{claim:lev_score_sym}
	For any Fourier matrix $F_S \in \C^{d\times s}$ with $s\leq d$, 
	\begin{align}
	\tau_{d+1 - j}(F_S) &= \tau_{j}(F_S) & & \text{for all $j \in [d]$}
	\end{align}
\end{claim}
\begin{proof}
	Let $a_1, \ldots, a_d \in \C^s$ denote the rows of $F_{S}$ and let $f_1, \ldots, f_s$ denote the frequencies of its columns.
	By the minimization characterization of the leverage scores from Fact \ref{fact:min_char}, we know that for any $j \in [d]$ there is some $y\in \C^d$ with $\|y\|_2^2 = \tau_{j}(F_S)$ and $y^T F_S = a_j$. I.e., for all $w \in [s]$,
	\begin{align}
	\label{eq:row_sum_equal_1}
	e^{-2 \pi i f_w (j-1)} = \sum_{z=1}^d y_z e^{-2 \pi i f_w (z-1)}.
	\end{align}
	To establish \eqref{eq:row_sum_equal_1} refer to Definition \ref{def:fourier_matrix} for the entries of $F_S$. 
	Taking the complex conjugate of both sides, we have: 
	\begin{align*}
	e^{2 \pi i f_w (j-1)} = \sum_{z=1}^d (y_z)^* e^{2 \pi i f_w (z-1)}.
	\end{align*}
	Finally, multiplying both sides by $e^{-2 \pi i f_w d-1}$ yields:
	\begin{align*}
	e^{-2 \pi i f_w (d - j+1 - 1)} = \sum_{z=1}^d (y_z)^* e^{-2 \pi i f_w (d - z)} \sum_{g=1}^d (y_{d+1-g})^* e^{-2 \pi i f_w (g -1)}.
	\end{align*}
	In other words, if we define $\tilde{y} \in \C^d$ by setting $\tilde{y}_g = (y_{d+1-g})^*$ for $g \in [d]$ then we have $\tilde{y}^TF_S = a_{d-j + 1}$. Of course $\|\tilde{y}\|_2^2 = \|{y}\|_2^2$, so we conclude that $\tau_{d+1 - j}(F_S) \leq \tau_{j}(F_S)$. Since this bound holds for all $j \in [d]$, we also have $\leq \tau_{j}(F_S) \leq \tau_{d+1 - j}(F_S)$,, which establishes the claim.
\end{proof}

\begin{proof}[Proof of Lemma \ref{lem:main_lev_bound}] With Claim \ref{claim:lev_score_sym} we are ready to prove the main result: as long as we establish \eqref{eq:simple_lev_bound} and \eqref{eq:uniform_lev_bound} for $j \leq (d+1)/2$ then the claim implies that it holds for all $j \in [d]$. 
	
	\medskip\noindent\textbf{Proof of \eqref{eq:simple_lev_bound}:} Again let $a_1, \ldots, a_d \in \C^s$ denote the rows of $F_{S}$ and suppose $S = \{f_1, \ldots, f_s\}$. Let $F_{S}^{(1,j)} \in \C^{j\times s}$ contain the first $j$ rows of $F_S$. As discussed in Definition \ref{def:lev_score}, $\sum_{z=1}^j \tau_z(F_{S}^{(1,j)}) = \rank(F_{S}^{(1,j)}) \leq s$. So it must be that there exists some $m \in [j]$ such that $\tau_m(F_{S}^{(1,j)}) \leq \frac{s}{j}$. I.e. there is some integer $m \leq j$ whose leverage score in $F_{S}^{(1,j)}$ is upper bounded by exactly the leverage score upper bound that we want to prove for $j$. 
	
	To take advantage of this fact, we use Fact \ref{fact:min_char}, which implies that since $a_m$'s leverage score is small, it can be written as a linear combination of $a_1, \ldots, a_j$ with small coefficients. We can ``shift'' this linear combination to write $a_j$ as a linear combination of $a_{1 + j - m}, \ldots, a_{2j - m}$ with the same coefficients. Since $j \leq \frac{d}{2}$ and thus $2j - m \leq d$, this is a valid linear combination of rows in $F_S$ to form $a_j$, which implies an upper bound on $\tau_j(F_S)$. 
	
	Formally, Fact \ref{fact:min_char} implies that there is some $y \in \C^j$ with $\|y\|_2^2 \leq \frac{s}{j}$ and $a_m = y^TF_{S}^{(1,j)}$. I.e., $\forall\,w \in [s]$,
	\begin{align}
	\label{eq:row_sum_equal}
	e^{-2 \pi i f_w (m-1)} = \sum_{z=1}^j y_z e^{-2 \pi i f_w (z-1)}.
	\end{align}
	Then we claim that for any $w \in [s]$ we also have:
	\begin{align}
	\sum_{z=1}^j y_z e^{-2 \pi i f_w (z-1 + j-m)} &= e^{-2 \pi i f_w ( j-m)} \sum_{z=1}^j y_z e^{-2 \pi i f_w (z-1)} \\&= e^{-2 \pi i f_w ( j-m)}e^{-2 \pi i f_w (m-1)} =  e^{-2 \pi i f_w (j-1)}. \label{eq:shift_bound}
	\end{align}
	The second equality follows from \eqref{eq:row_sum_equal}.
	Now let $\tilde{y} \in \C^d$ be a vector which is zero everywhere, except that for $z \in [j]$ we set $\tilde{y}_{z + j-m} = {y}_{z}$\footnote{Note that this operation is valid for $j \leq (d+1)/2$ because then $z+j - m\leq 2j -m \leq d$ (since $m \geq 1)$}. From \eqref{eq:shift_bound} we have that $\tilde{y}^TF_s = a_j$. Furthermore, $\|\tilde{y}\|_2^2 = \|y\|_2^2$, and thus $\|\tilde{y}\|_2^2 \leq \frac{s}{j}$. By Fact \ref{fact:min_char}, we conclude that $\tau_j(F_S) \leq \frac{s}{j}$ for $j \leq (d+1)/2$, as desired. 
	
	
	\medskip\noindent\textbf{Proof of \eqref{eq:uniform_lev_bound}:}
	While the proof above uses the minimization characterization of the leverage scores, here we will apply the \emph{maximization characterization} of Fact \ref{fact:max_char}. It follows from this characterization that \eqref{eq:uniform_lev_bound} can be proven by showing that, for any Fourier matrix $F_S \in \C^{d\times s}$ and any $y$ which can be written as $F_Sx$ for some $x\in \C^s$,
	\begin{align}
	\label{eq:smoothness_to_prove}
	\frac{|y_j|^2}{\|y\|_2^2} &\leq \frac{O(s^6 \log^3 s)}{d} & & \text{for all} & j&\in [d].
	\end{align} 
	In other words, we need to establish that \emph{any Fourier sparse function} (which can be written as $F_Sx$ for some set of frequencies $S$ and coefficients $x$) cannot be too concentrated at any point on the integer grid. In particular, the squared magnitude at a point $|y_j|^2$ cannot exceed the average squared magnitude $\|y\|_2^2/d$ by more than $\tilde{O}(s^6)$. 
	
	Tools for proving this sort of bound were recently developed in \cite{ChenKanePrice:2016}. We will rely on one particular result from that paper, which we restate below in a less general form than what was originally proven:
	\begin{claim}[Corollary of Claim 5.2 in the arXiv version of \cite{ChenKanePrice:2016}]\label{claim:small_coeff_sum}
		For any Fourier matrix $F_S \in \C^{d\times s}$ and any $y$ that can be written $y = F_Sx$ for some $x \in \C^s$, there is a fixed constant $c$ such that, if $m = cs^2\log(s+1)$, then for any $\tau \in 1, \ldots, \lfloor d/2m \rfloor$ and $j \leq (d+1)/2$, there always exist $C_1, \ldots, C_m \in \C$ such that:
		\begin{enumerate}
			\item $|C_z| \leq 11$ for all  $z \in [m]$.
			\item $y_j = \sum_{z=1}^m C_z y_{\left[j + z\tau\right]}$.
		\end{enumerate}
	\end{claim}
	
	Claim \ref{claim:small_coeff_sum} implies that, if $y$ has an $s$ sparse Fourier transform, then any value $y_j$ can be written as a {small coefficient} sum of entries of $y$ on \emph{any uniform grid} of size $\tilde{O}(s^2)$ that starts at $j$. The implication is that $|y_j|$ cannot be that much larger than the value of $y$ on this uniform grid, which is a first step to proving  \eqref{eq:smoothness_to_prove}, which asserts that $|y_j|$ cannot be much larger than the average absolute value of all entries in $y$.
	Formally, from Cauchy-Schwarz inequality we have that, for any $\tau \in 1, \ldots, \lfloor d/2m \rfloor$,
	\begin{align*}
	|y_j|^2 \leq m\sum_{z=1}^m |C_z|^2 |y_{\left[j + z\tau\right]}|^2.
	\end{align*}
	It follows that:
	\begin{align*}
	|y_j|^2 \leq O(m) \frac{1}{d/2m}\sum_{\tau = 1}^{\lfloor d/2m \rfloor}\sum_{z=1}^m |y_{\left[j + z\tau\right]}|^2 = \frac{O(m^2)}{d} \sum_{z=1}^m \sum_{\tau = 1}^{\lfloor d/2m \rfloor}|y_{\left[j + z\tau\right]}|^2  
	&\leq \frac{O(m^2)}{d} \sum_{z=1}^m  \|y\|_2^2 \\
	&= \frac{O(m^3)}{d}  \|y\|_2^2.
	\end{align*}
	Since $m = O(s^2 \log(s+1))$, \eqref{eq:uniform_lev_bound} follows.
\end{proof}

%

\subsection{Prony's method with inexact root-finding}\label{sec:badRoots}
In this section, we give algorithms that recover rank-$k$ Toeplitz covariance matrices, but when we only have access to an inexact root-finding algorithm.
We will use the following standard subroutine for approximately finding roots of complex univariate polynomials as a black-box.

\begin{lemma}[\cite{pan2002univariate}]
\label{lem:approx-poly}
For all integers $k$, all $\beta \geq k \log k$, there exists an algorithm $\mathtt{FindRoots}$ which takes in a degree-$k$ polynomial $p: \mathbb{C} \to \mathbb{C}$ so that if $z_1, \ldots, z_k$ are the roots of $p$, and satisfy $|z_i| \leq 1$ for all $i = 1, \ldots, k$, then $\mathtt{FindRoots}(q, k, \beta)$ returns $z_1^*, \ldots, z_k^*$ so that there exists a permutation $\pi: [k] \to [k]$ so that $|z_{\pi(i)} - z_i^*| \leq 2^{2 - \beta / k}$.
Moreover, the algorithm runs in time $O(k \log^2 k \cdot (\log^2 k + \log \beta))$.
\end{lemma}
\noindent
We will also require the following bound on how large coordinates of Fourier sparse functions can get.
This can be viewed as a discrete analog of Lemma 5.5 in~\cite{ChenKanePrice:2016}.
\begin{lemma}
\label{lem:extrapolation}
Let $z \in \R^d$ be given by $z = F_R y$ for some $R$ with $|R| \leq m$.
Then, we have
\[
\Norm{z}_2^2 \leq \Paren{2^m \cdot m^{m + 7} \log^3 m}^{m \log (d / m)} \Norm{z_{[m]}}_2^2 = d^{\Paren{m^{O(1)}}} \Norm{z_{[m_t]}}_2^2  \; .
\]
\end{lemma}
\begin{proof}
Let $m_t = (1 + 1/m)^t m$.
We proceed by induction, by showing the stronger claim that
\begin{equation}
\label{eq:extrapolation-induction}
\Norm{z_{[m_{t + 1}]}}_2^2 \leq 2^m \cdot m^{m + 7} \log^3 m \Norm{z_{[m_t]}}_2^2 \; ,
\end{equation}
for all $t$.
The desired claim then follows by repeatedly applying this claim.
The base case $t = 0$ is trivial.
For any $t > 0$, assume that~\eqref{eq:extrapolation-induction} holds for $[m_t]$.
Let $j \in [(1 + 1/m) m_t]$.
Then we can write $j$ as $j = j_0 + m \tau$ for some $j_0, \tau \in m_t / m$, so that $j_0 + (k - 1) \tau \in [m_t]$.
By Lemma 5.3 in~\cite{ChenKanePrice:2016}, we can write $z_j = \sum_{j = 0}^{m - 1} a_j z_{j_0 + j \tau}$, for some coefficients $a_j$ satisfying $|a_j| \leq (2m)^m$.
Therefore, we have
\begin{align*}
|z_j|^2 &= \left| \sum_{j = 0}^{m - 1} a_j z_{j_0 + j \tau} \right|^2 \leq m \sum_{j = 0}^{m - 1} |a_j| \cdot |z_{j_0 + j \tau}|^2 \\
&\leq 2^m \cdot m^{m + 2} \max_{j \in [m_t]} |z_{j}|^2 \\
&\leq \frac{2^m \cdot m^{m + 8} \log^3 m}{m_t} \Norm{z_{[m_t]}}_2^2 \; ,
\end{align*}
where the last line follows from~\eqref{eq:simple_lev_bound}.
Thus overall we have
\[
\Norm{z_{[m_{t + 1}]}}_2^2 \leq \Norm{z_{[m_t]}}_2^2 + \frac{m_t}{m} \frac{2^m \cdot m^{m + 8} \log^3 m}{m_t} \Norm{z_{[m_t]}}_2^2 = 2^m \cdot m^{m + 7} \log^3 m \Norm{z_{[m_t]}}_2^2 \; ,
\]
as claimed.
\end{proof}

For completeness, in Algorithm~\ref{alg:prony-inexact}, we give the formal pseudocode for Prony's method using $\texttt{FindRoots}$, although it is almost identical to the pseudocode in Algorithm~\ref{alg:exactSFT-1}.
The two differences are that (1) we round the frequencies we get to a grid, and (2) we regress to find the coefficients of the frequencies, rather than exactly solving a system of linear equations.
 \begin{algorithm}[H]
\caption{\algoname{Prony's method with inexact root-finding}}
{\bf input}:  A sample $x = F_S y$, where $S$ has size at most $k$ \\
{\bf parameters}: rank $k$, accuracy parameter $\beta > k \log k$.\\
{\bf output}: A set of frequencies $R$ and $\hat y \in \R^k$ so that $F_{R} \hat{y}$ approximates $x$.
\begin{algorithmic}[1]
\State{Solve $P_k (x) c = b_k(x)$.}
\State{Let $p(t) = \sum_{s=1}^k c_s t^s$ and let $R' = \texttt{FindRoots}(p, \beta)$.}
\State{Let $R$ be the result of rounding every element of $R'$ to the nearest integer multiple of $2^{3 - \beta / k}$}
\State{Let $F_R \in \C^{2 k \times k}$ be the Fourier matrix (Def \ref{def:fourier_matrix}) with $2 k$ rows and  frequencies $R$.}
\State{Let $\hat y = \argmin \Norm{F_R y - x_{[2k]}}_2$.} \\
\Return{$R$ and $\hat y$}.
\end{algorithmic}
\label{alg:prony-inexact}
\end{algorithm}

Before we analyze the algorithm, it will be helpful to establish some additional notation.
For any set $S = \{f_1, \ldots, f_k \}$ of $k$ frequencies, let $\bar{S} = \{f_1', \ldots, f_{k'}'\}$ formed by rounding the elements of $S$ to the nearest integer multiple of $2^{3-\beta / k}$.
Let $\rho: [k] \to [k']$ be the function which takes $j \in [k]$ and maps it to the unique element $\ell \in [k']$ so that $|e^{2 \pi i f_j} - e^{2 \pi i f'_\ell}| \leq 2^{3 - \beta / k}$.
That is, it takes each frequency to the rounded version that the algorithm recovers.
For any $y \in \R^k$, let $y' \in \R^{k'}$ be given by $(y')_j = \sum_{\ell \in \rho^{-1} (j)} y_\ell$.
Then, putting together the guarantees of Lemmata~\ref{lem:approx-poly} and~\ref{lem:prony} yields:

\begin{corollary}
\label{cor:prony-inexact}
Let $\beta \geq k \log k$, and let $S$ be  set of $k$ frequencies.
Let $y \in \R^k$, let $x = F_S y$, and let $I = \mathrm{supp} (y)$.
Then, on input $x$, Algorithm~\ref{alg:prony-inexact} outputs the set of frequencies $S' = \{e^{2 \pi i f_{\rho(j)}'}: j \in I \}$, and $\hat{y} \in \R^{|S'|}$ so that $\norm{x - F_{S'} \hat{y}}_2 \leq 2^{-\beta / k} d^{\Paren{k^{O(1)}}} \Norm{y}_2$.
Moreover, Algorithm~\ref{alg:prony-inexact} requires only $2k$ entrywise queries to $x$.
\end{corollary}
\begin{proof}
By Lemma~\ref{lem:prony}, and the linear independence of columns of Fourier matrices, we know that for all $j \in I$, there exists $f \in S'$ so that $|e^{2 \pi i f_j} - e^{2 \pi i f}| \leq 2^{3 - \beta / k}$.
Since the exact roots of the polynomial in Lemma~\ref{lem:prony} is exactly the set $\{e^{2 \pi i f_j}: j \in I\}$, the inexact roots that the algorithm of Lemma~\ref{lem:approx-poly} outputs must be rounded to the set $S'$.
Hence Algorithm~\ref{alg:prony-inexact} will output this set of frequencies.

Letting $F'$ denote the restriction of $F_{S'}$ to the first $2k$ rows, we have 
\begin{align*}
\Norm{x_{[2k]} - F' y'}_2 &= \Norm{\sum_{j \in I} y_i \Paren{F(f_j)_{S'} - F(f_{\rho(j)})_{S'}} }_2 \leq \sum_{i = 1}^{2k} |y_i| \Norm{F(f_j)_{S'} - F(f_{\rho(j)})_{S'}}_2 \\
&\leq \sum_{i = 1}^{2k} |y_i| \cdot O(2^{-\beta / k} k) = O \Paren{2^{- \beta / k} k^2 \Norm{y}_2} \; .
\end{align*}
Thus, $\hat{y}$ satisfies $\Norm{x_{[2k]} - F' \hat{y}}_2 \leq O \Paren{2^{- \beta / k} k^2 \Norm{y}_2}$.
Let $z = x - F_{S'} \hat{y}$.
The vector $z$ has Fourier support of size at most $2k$, and moreover, the above calculation demonstrates that $\Norm{z_{[2k]}}_2 \leq O(2^{-\beta / k} k^2 \Norm{y}_2)$.
Thus, by Lemma~\ref{lem:extrapolation}, we have that $\Norm{x - F_{S'} \hat{y}}_2 = \Norm{z}_2 \leq 2^{-\beta / k} d^{\Paren{k^{O(1)}}} \Norm{y}_2$, as claimed.
\end{proof}
\subsubsection{A sample-efficient and $\poly(d)$ time algorithm}
Building on Algorithm~\ref{alg:prony-inexact}, we first demonstrate how to give a simple algorithm which requires $O(k \log d / \eps^2)$ samples and runtime $\poly(d, k, 1 / \eps)$ to learn a rank-$k$ Toeplitz matrix.
The algorithm is very straightforward: for every sample, simply use Algorithm~\ref{alg:prony-inexact} to recover it to very good accuracy from $2k$ entries, and then run the full sample algorithm (i.e. Algorithm~\ref{alg:sparse} with the full ruler) on the set of samples we recover.
The formal pseudocode is given in Algorithm~\ref{alg:prony-poly-d}.

\begin{algorithm}[H]
\caption{\algoname{Recovering rank-$k$ Toeplitz covariance via Prony's method}}
{\bf input}:  Samples $x^{(1)}, \ldots, x^{(n)} \sim \normal (0, T)$, where $T$ is a rank-$k$ Toeplitz covariance matrix \\
{\bf parameters}: rank $k$, accuracy parameter $\beta > k \log k$.\\
{\bf output}: A matrix $\hat T$ that approximates $T$
\begin{algorithmic}[1]
\For{$j = 1, \ldots, n$}
	\State{Let $\hat{x}^{(j)}$ be the output of Algorithm~\ref{alg:prony-inexact} on input $x^{(j)} $ with accuracy parameter $\beta$}
\EndFor
\State{Let $\hat T$ be the output of Algorithm~\ref{alg:sparse} with $R = [d]$ and samples $\hat{x}^{(1)}, \ldots, \hat{x}^{(n)}$}\\
\Return{$\hat T$}
\end{algorithmic}
\label{alg:prony-poly-d}
\end{algorithm}
\noindent
We prove:
\begin{theorem}
\label{thm:prony-inexact-1}
Let $\eps > 0$.
Let $T \in \R^{d \times d}$ be a PSD rank-$k$ Toeplitz matrix, let $\beta > k \log k$, and let $\eta = 2^{-\beta / k} d^{\Paren{k^{O(1)}}}$.
Given $x^{(1)}, \ldots, x^{(n)} \sim \normal (0, T)$, Algorithm~\ref{alg:prony-poly-d} outputs $\hat T$ so that
\begin{equation}
\label{eq:prony-poly-d}
Pr \Brac{ \Norm{ \hat T - T}_2 > \Paren{\eps + 4 (1 + \eta) \eta} \Norm{T}_2 } \leq \frac{C d^2}{\eps} \exp \Paren{- \frac{c n \eps^2}{\log d}} + \exp \Paren{-c n} \; .
\end{equation}
Moreover, the algorithm has entrywise sample complexity $2k$, and runs in time $\poly(n, k, d, \log \beta)$.
In particular, for any $\delta > 0$, if we let $n = \Theta (\log (d / (\eps \delta)) \log d / \eps^2)$ and $\beta = k^{O(1)} \log d + k \log 1 / \eps$, then Algorithm~\ref{alg:prony-poly-d} outputs $\That$ so that $\Pr \Brac{\Norm{\That - T}_2 > 2 \eps \Norm{T}_2} < \delta$, and the algorithm a total sample complexity of $\Theta (k\log d / \eps^2)$, and runs in time $\poly (k, d, 1/\eps, \log 1 / \delta)$.
\end{theorem}
\begin{proof}
For all $j$, we have $x^{(j)} = F_S D^{1/2} U g^{(j)}$ for some $S$ with $|S| \leq k$, a fixed unitary matrix $U$, and $g^{(j)} \sim \normal (0, I)$.
Let $y^{(j)} = D^{1/2} U g^{(j)}$.
By Corollary~\ref{cor:prony-inexact}, we know that $\Norm{x^{(j)} - \hat{x}^{(j)}}_2 \leq \eta \Norm{y^{(j)}}_2$.
Hence, we have 
\begin{align*}
\Norm{\frac{1}{n} \sum_{j = 1}^n \Paren{x^{(j)}} \Paren{x^{(j)}}^\top - \frac{1}{n} \sum_{j = 1}^n \Paren{\hat{x}^{(j)}} \Paren{\hat{x}^{(j)}}^\top}_F &\leq \frac{1}{n} \sum_{j = 1}^n \Norm{\Paren{x^{(j)}} \Paren{x^{(j)}}^\top - \Paren{\hat{x}^{(j)}} \Paren{\hat{x}^{(j)}}^\top}_F \\
&\leq \frac{2}{n} \sum_{j = 1}^n \max \Paren{\Norm{x^{(j)}}_2, \Norm{\hat{x}^{(j)}}_2} \cdot \eta \Norm{y^{(j)}}_2 \\
&\leq \frac{2 (1 + \eta) \eta}{n} \sum_{j = 1}^n (1 + \eta) \cdot \eta \Norm{y^{(j)}}_2^2 \; ,
\end{align*}
where the last line follows since $F_S$ has maximum singular value at most $1$.
Since $U$ is unitary, we know that $\Norm{y^{(j)}}_2^2$ is distributed as a non-uniform $\chi^2$-random variable with $d$ degrees of freedom, with coefficients given by the diagonal entries of $D$.
As argued in the proof of Lemma~\ref{lem:fss2}, we know that each diagonal entry of $D$ is bounded by $\Norm{T}_2 / d$.
By standard concentration results for subexponential random variables, we know that for all $t > 0$,
\[
\Pr \Brac{\frac{1}{n} \sum_{j = 1}^n \Norm{y^{(j)}}_2^2 > 2 \Norm{T}_2} \leq \exp \Paren{-c n} \; .
\]
for some universal constant $c$.
Moreover, by Theorem~\ref{thm:linear}, we know that 
\[
\Pr \Brac{ \Norm{ \avg \Paren{\frac{1}{n} \sum_{j = 1}^n \Paren{x^{(j)}} \Paren{x^{(j)}}^\top} - T}_2 > \eps \Norm{T}_2 } \leq \frac{C d^2}{\eps} \exp \Paren{- \frac{c n \eps^2}{\log d}} \; ,
\]
for universal constants $c, C$.
Thus, combining these two bounds, and observing that averaging can only decrease difference in Frobenius norm, yields~\eqref{eq:prony-poly-d}, as claimed.
\end{proof}

\subsection{An $O(\log \log d)$ time algorithm for well-conditioned Toeplitz matrices}
In this section, we demonstrate an algorithm with extremely efficient runtime in terms of dimension, when the ratio between the largest and smallest non-zero singular value of the matrix is bounded away from zero.
Specifically, for any rank $k$ PSD matrix $M \in \R^{d \times d}$ with non-zero eigenvalues $\sigma_1 \geq \ldots \geq \sigma_{k'} > 0$, where $k' \leq d$, let $\kappa (M) = \sigma_1 / \sigma_{k'}$.

Our algorithm for this setting will be almost identical to the algorithm presented in the main text that assumes an exact root finding oracle.
For each sample $x$, we will apply Algorithm~\ref{alg:prony-inexact} to obtain $S'$ and $y'$ so that $x$ is well-approximated by $F_{S'} y'$.
We then simply estimate the diagonals of the $k \times k$ matrix given by the second moment $y$ vectors, and output the set of frequencies as well as the set of diagonal entries estimated.
The formal pseudocode is presented in Algorithm~\ref{alg:prony-condition}.
\begin{algorithm}[H]
\caption{\algoname{Recovering rank-$k$ Toeplitz covariance with bounded condition number}}
{\bf input}:  Samples $x^{(1)}, \ldots, x^{(n)} \sim \normal (0, T)$, where $T$ is a rank-$k$ Toeplitz covariance matrix \\
{\bf parameters}: rank $k$, condition number $\kappa (T)$.\\
{\bf output}: A set of frequencies $R$ and a diagonal matrix $\hat D$ so that $F_R \hat{D} F_R^*$ approximates $T$
\begin{algorithmic}[1]
\State{Let $\beta = k^{O(1)} \log d + O(k \log (\kappa(T) / \eps))$}
\For{$j = 1, \ldots, n$}
	\State{Let $R_j, \hat{y}^{\paren{j}}$ be the output of Algorithm~\ref{alg:prony-inexact} on input $x^{(j)} $ with accuracy parameter $\beta$}
	\State{As we will argue, with probability $1$ we have $R_j = R_{j'}$ for all $j, j' \in [n]$}
\EndFor
\State{Let $D_\ell = \frac{1}{n} \sum_{j \in [n]} |\hat{y}^{(j)}_{\ell}|^2$, for $\ell \in [|R_1|]$.}
\State{Let $\hat D = \mathrm{diag} (D_1, \ldots, D_{R_1})$.}\\
\Return{$R_1$ and $\hat D$.}
\end{algorithmic}
\label{alg:prony-condition}
\end{algorithm}
Our guarantee is the following:
\begin{theorem}
\label{thm:prony-condition}
Let $\eps > 0$.
Let $T \in \R^{d \times d}$ be a PSD rank-$k$ Toeplitz matrix.
Given $x^{(1)}, \ldots, x^{(n)} \sim \normal (0, T)$ and $\kappa (T)$, Algorithm~\ref{alg:prony-condition} outputs $R$ and $\hat D$ so that if we let $\hat T = F_R \hat{D} F_R^*$, then
\begin{equation}
\label{eq:prony-condition}
\Pr \Brac{\Norm{\hat T - T}_2 > \eps \Norm{T}_2} < (k + 1) \exp \Paren{-c n \eps^2} \; .
\end{equation}
Moreover, the algorithm has entrywise sample complexity $2k$, and runs in time 
\[
\poly(n, k, \log \log d, \log \log (\kappa (T)), \log \log (1 / \eps)) \; .
\]
In particular, for any $\delta > 0$, if we let $n = \Theta (\log (1 / \delta) \log k / \eps^2)$, then Algorithm~\ref{alg:prony-poly-d} outputs $\That$ so that $\Pr \Brac{\Norm{\That - T}_2 > 2 \eps \Norm{T}_2} < \delta$, and the algorithm requires $\Theta (k \log (1 / \delta) \log d / \eps^2)$ entrywise samples, and runs in time 
\[
\poly(n, k, \log \log d, \log \log (\kappa (T)), 1 / \eps) \; .
\]

\end{theorem}
\noindent
We first prove the following geometric lemma, which states that any Toeplitz matrix with bounded $\kappa$ cannot have frequencies which are too close together:
\begin{lemma}
\label{lem:roots-separated}
Let $T \in R^{d \times d}$ be a PSD Toeplitz matrix, and let $T = F_S D F_S$ be the Vandermonde decomposition of $T$.
Then
\[
\min_{\substack{j, j' \in S\\j \neq j'}} |e^{2 \pi i f_j} - e^{2 \pi i f_{j'}}| = \Omega \Paren{\frac{ 1}{\sqrt{d \cdot \kappa (T)}}} \; .
\]
\end{lemma}
\begin{proof}
Let $j, j' \in S$ achieve the minimum on the LHS, and WLOG assume that $D_{j,j} \geq D_{j', j'}$.
Let $\Delta = |1 - e^{2 \pi i (f_j - f_{j'})}| = \frac{1}{2} |e^{2 \pi i f_j} - e^{2 \pi i f_{j'}}|^2$.
By calculation, observe that 
\[
\left|\Iprod{F(f_j), F(f_{j'})} \right| = \left| \sum_{\ell = 1}^d e^{2 \pi i (f_j - f_{j'}) \ell} \right| = d( 1 - O (d \Delta)) \; .
\]
As argued in Lemma~\ref{lem:fss2}, the largest eigenvalue of $T$ is at least $d \cdot D_{j,j}$.
Moreover, since the vectors $\{ F(f_\ell): \ell \in S \}$ are linearly independent and span the range of $T$, we know that the smallest non-zero eigenvalue of $T$ is given by $\min_{x \in U} \frac{x^* T x}{\Norm{x}_2^2}$, where $U = \mathrm{span} \Paren{F(f_\ell): \ell \in S \}}$.
Now we can write $F(f_j) = u + v$, where $u \in \mathrm{span} \Paren{\{ F(f_\ell): \ell \in S \setminus \{ j' \} \}}$, and $v$ is orthogonal to $F(f_\ell)$, for all $\ell \neq j'$.
Moreover, we have that $\Norm{u}_2^2 \geq \left|\Iprod{F(f_j), F(f_{j'})} \right|$, and so by the Pythagorean theorem we must have $\Norm{v}_2^2 \leq O(d \Delta)$, and clearly $v \in \mathrm{span} \Paren{\{ F(f_\ell): \ell \in S \}}$.
Therefore, the smallest nonzero eigenvalue of $T$ is at most 
\begin{align*}
\frac{v^* T v}{\Norm{v}_2^2} = D_{j', j'} \frac{\left| \Iprod{F(f_{j'}), u} \right|^2}{\Norm{u}_2^2} = D_{j', j'} \Norm{u}_2^2 \leq D_{j', j'} \cdot O(d \Delta) \; .
\end{align*}
Since $D_{j', j'} \leq D_{j,j}$, we conclude that the ratio between the largest eigenvalue of $T$ and the smallest nonzero eigenvalue of $T$ is at least $\Omega (1 / (d \Delta))$, as claimed.
\end{proof}

\begin{proof}[Proof of Theorem~\ref{thm:prony-condition}]
Let $T = F_S D F_S^*$ be the Vandermonde decomposition of $T$, so that $S = \{f_1, \ldots, f_{k'}\}$ for some $k' \leq k$.
For simplicity of exposition assume that $k' = k$.
By our choice of $\beta$, the guarantees of Lemma~\ref{lem:approx-poly}, and Lemma~\ref{lem:roots-separated}, we have that if $R = \{f'_1, \ldots, f'_{k''}\}$ is the set of roots obtained by Algorithm~\ref{alg:prony-condition}, then $k'' = |R| = k$, and moreover, for all $j \in R$, there exists a unique $\ell \in S$ so that $|e^{2 \pi i f'_j} - e^{2 \pi i f_\ell}| \leq 2^{3 - \beta / k}$.
In particular, by our choice of $\beta$, we have that for all $j, j' \in R$, $|e^{2 \pi i f'_j} - e^{2 \pi i f'_{j'}}| \geq \Omega \Paren{\frac{1}{\sqrt{d \kappa (T)}}}$.

Let $j^* = \argmax_{j \in [k]} D_{j,j}$, and let $\lambda_1 \geq \ldots \geq \lambda_k$ be the nonzero eigenvalues of $T$.
As argued before, we have $\lambda_1 \geq d \cdot D_{j^* j^*}$.
Moreover, since $D_{j^*} F_S F_S^* \succeq T$, we conclude that $\lambda_k \leq D_{j^* j^*} \sigma_k$, where $\sigma_k$ is the $k$th largest singular value of $F_S F_S^*$.
Putting these bounds together, we conclude that $\sigma_k \geq \lambda_k / D_{j^* j^*} \geq d / \kappa(T)$.
By our choice of $\beta$, we can also conclude that the smallest singular value of $F_R$ is at least $\Omega \left(\sqrt{\frac{d}{\kappa (T)}} \right)$.

Let $\eta = 2^{-\beta / k} d^{\Paren{k^{O(1)}}}$, and let $y^{(j)}$ and $g^{(j)}$ be as before, so that $x^{(j)} = F_S y^{(j)}$ and $y^{(j)} = D^{1/2} U g^{(j)}$, and $g^{(j)}$ is a standard normal Gaussian.
By Corollary~\ref{cor:prony-inexact}, we know that for all $j \in [n]$, we must have $\Norm{x^{(j)} - F_R \hat{y}^{(j)}}_2 \leq \eta \Norm{y}_2$.
By a triangle inequality, this in turn implies that $\Norm{F_R (y^{(j)} - \hat{y}^{(j)})}_2 \leq O(\eta) \cdot \Norm{y}_2$.
By our bound on the smallest singular value of $F_R$, we conclude that $\Norm{y^{(j)} - \hat{y}^{(j)}}_2 \leq O(\eta \sqrt{\kappa (T) / d}) \cdot \Norm{y}_2$.

Let $Z_\ell = \frac{1}{n} \sum_{j \in [n]} |y^{(j)}_{\ell}|^2$.
As argued in the proof of Theorem~\ref{thm:prony-exact-roots}, we know that for any $\ell \in [k]$, we must have
\[
\Pr \Brac{|Z_\ell - D_{\ell \ell}| > \eps D_{\ell \ell}} \leq 2 \exp \Paren{-c n \eps^2} \; ,
\]
for some constant $c$.
By standard concentration results for subexponential random variables, we know that for all $t > 0$,
\[
\Pr \Brac{\frac{1}{n} \sum_{j = 1}^n \Norm{y^{(j)}}_2^2 > 2 \Norm{T}_2} \leq \exp \Paren{-c n} \; .
\]
for some universal constant $c$.
Combining these two bounds with union bounds, we conclude that 
\[
\Pr \Brac{-\eps D - O(\eta \sqrt{\kappa (T) / d}) I \preceq \hat D - D \preceq \eps D + O(\eta \sqrt{\kappa (T) / d}) I} > 1 - (k + 1) \exp \Paren{-c n \eps^2} \; ,
\]
and thus
\[
\Pr \Brac{-\eps T - O(\eta \sqrt{\kappa (T) / d}) F_S F_S^* \preceq F_S \hat D F_S^* - T \preceq \eps T + O(\eta \sqrt{\kappa (T) / d}) F_S F_S^*} > 1 - (k + 1) \exp \Paren{-c n \eps^2} \; .
\]
To conclude, we note that $\Norm{F_S F_S^*}_2 \leq \Norm{F_S F_S^*}_F \leq kd$, and conditioned on our good event,
\begin{align*}
\Norm{F_R \hat{D} F_R - F_S \hat{D} F_S}_2 &\leq \Norm{F_R D F_R - F_S D F_S}_2 \\
&\leq (1 + \eps) \Norm{D}_2 \cdot \Norm{F_R - F_S}_F^2 \\
&\leq \Norm{T}_2 \cdot O(d 2^{-2\beta / k}) \; .
\end{align*}
 and hence 
 \[
\Pr \Brac{\Norm{\hat T - T}_2 > \Paren{\eps + O(\eta \sqrt{\kappa (T) / d}) + d 2^{-2\beta / k}} \Norm{T}_2} < (k + 1) \exp \Paren{-c n \eps^2} \; .
\]
The result then follows by adjusting $\eps$, and from our choice of $\beta$.
\end{proof}
\end{document}